\documentclass[a4paper,UKenglish,cleveref,autoref,thm-restate]{lipics-v2021}

\usepackage{todonotes}
\presetkeys{todonotes}{inline}{}

\usepackage{stackengine}
\usepackage{scalerel}
\usepackage{xspace}

\usepackage{tikz}
\usetikzlibrary{cd}
\usetikzlibrary{decorations.pathmorphing}

\usepackage{stmaryrd}
\usepackage{amsmath}
\usepackage{amsthm}
\usepackage{amssymb}
\usepackage{mathtools}
\usepackage{bm}

\usepackage{mathpartir}
\usepackage{fontawesome}
\usepackage{enumitem}

\newtheorem{construction}[theorem]{Construction}

\title{Relative induction principles for type theories} 

\author{Rafaël Bocquet}{Eötvös Loránd University, Budapest, Hungary}{bocquet@inf.elte.hu}{https://orcid.org/0000-0001-6484-9570}
{The author was supported by the European Union, co-financed by the European Social Fund (EFOP-3.6.3-VEKOP-16-2017-00002).}

\author{Ambrus Kaposi}{Eötvös Loránd University, Budapest, Hungary}{akaposi@inf.elte.hu}{http://orcid.org/0000-0001-9897-8936}
{The author was supported by the ``Application Domain Specific Highly Reliable IT Solutions'' project which has been implemented with the support provided from the National Research, Development and Innovation Fund of Hungary, financed under the Thematic Excellence Programme TKP2020-NKA-06 (National Challenges Subprogramme) funding scheme.}

\author{Christian Sattler}{Chalmers University of Technology, Sweden}{sattler@chalmers.se}{https://orcid.org/0000-0001-6374-4427}{}

\authorrunning{R. Bocquet and A. Kaposi and C. Sattler} 

\Copyright{Rafaël Bocquet and Ambrus Kaposi and Christian Sattler} 



\category{} 

\relatedversion{} 

\supplement{}

\acknowledgements{}

\nolinenumbers 

\hideLIPIcs  

\EventEditors{John Q. Open and Joan R. Access}
\EventNoEds{2}
\EventLongTitle{42nd Conference on Very Important Topics (CVIT 2016)}
\EventShortTitle{CVIT 2016}
\EventAcronym{CVIT}
\EventYear{2016}
\EventDate{December 24--27, 2016}
\EventLocation{Little Whinging, United Kingdom}
\EventLogo{}
\SeriesVolume{42}
\ArticleNo{23}

\input{macros.tex}

\begin{document}

\maketitle

\begin{abstract}
  We present new induction principles for the syntax of dependent type theories, which we call relative induction principles.
  The result of the induction principle relative to a functor $F$ into the syntax is stable over the codomain of $F$.
  We rely on the internal language of presheaf categories.
  In order to combine the internal languages of multiple presheaf categories, we use Dependent Right Adjoints and Multimodal Type Theory.
  Categorical gluing is used to prove these induction principles, but it not visible in their statements, which involve a notion of model without context extensions.
  As example applications of these induction principles, we give short and boilerplate-free proofs of canonicity and normalization for some small type theories, and sketch proofs of other metatheoretic results.
\end{abstract}

\section{Introduction}\label{sec:introduction}

\paragraph*{Induction principles}

Syntax without bindings or equations is characterized by its universal property as the initial object of some category of algebras, or equivalently by its induction principle as an inductive type.
The same can be said for syntax with equations, \eg quotient inductive-inductive types \cite{QIITs,DBLP:conf/lics/KovacsK20}.
As for syntax with bindings, we can encode it using syntax with equations but without bindings by making explicit the contexts and substitutions \cite{popl16}.

While this construction yields induction principles for syntax with bindings, most metatheoretic results are not direct applications of these induction principles.
They often involve a second step, in which the contexts over which the result holds are identified.
For example, canonicity only holds in the empty context, whereas normalization holds over every context, but is only stable under renamings.
This second step is most of the time handled in an ad-hoc manner.

Our main contribution is to show how this second step can be handled in a principled way and to introduce new induction principles for syntactic categories with bindings that merge the two steps into one.
More specifically, we give statements and proofs of so called \emph{relative induction principles} for a small dependent type theory $\Th_{\Pi,\BoolTy}$ with function space and booleans and for a minimal version of cubical type theory.
We use these theories to present our constructions, but they do not rely on any specific feature of these theories.
They could be generalized to arbitrary type theories (for some general definition of type theory, such as Uemura's definition~\cite{GeneralFrameworkSemanticsTT}).
In the appendix we show how to generalize to a hierarchy of universes closed under function space and booleans. We leave the full formal generalization to future work.

In the general case, we consider a functor $F : \CC \to \CS$, where $\CS$ is the syntax of our theory and $\CC$ a category which should satisfy some universal property.
We give induction principles which directly provide results that are stable over the morphisms of $\CC$ (hence the name \emph{relative}).
Under the hood, the relative induction principles use the universal properties of both $\CC$ and $\CS$.
The input data for a relative induction principle consists of a \emph{displayed model without context extensions}, along with some additional data depending on the universal property of $\CC$.

The following table lists example functors and the result that the induction principle relative to the given functor provides.
\begin{alignat*}{4}
  & \{\diamond\} && { }\to{ } && \Init_{\Th_{\Pi,\BoolTy}} && \quad\text{Canonicity \cite{CoquandNormalization} (Section \ref{sec:applCanon})} \\
  & \CRen && { }\to{ } && \Init_{\Th_{\Pi,\BoolTy}} && \quad\text{Normalization \cite{altenkirch_et_al:LIPIcs:2016:5972,CoquandNormalization} (Section \ref{sec:applNorm})} \\
  & \square && { }\to{ } && \Init_{\mathsf{CTT}} && \quad\text{(Homotopy/Strict) canonicity for cubical type theory \cite{coquand2021canonicity} (Section \ref{sec:applCanCTT})} \\
  & \mathcal{A}_{\square} && { }\to{ } && \Init_{\mathsf{CTT}} && \quad\text{Normalization for cubical type theory \cite{NormalizationCTT} (Section \ref{sec:applNormCTT})} \\
  & \Init_{\mathsf{ITT}} && { }\to{ } && \Init_{\mathsf{ETT}} && \quad\text{Conservativity of Extensional Type Theory over} \\
    & && && && \quad\text{Intensional Type Theory \cite{hofmann95conservativity,Oury2005,10.1145/3293880.3294095} (Section \ref{sec:applConserv})} \\
  & \Init_{\mathsf{HoTT}} && { }\to{ } && \Init_{\mathsf{2LTT}} && \quad\text{Conservativity of two-level type theory over HoTT \cite{CapriottiThesis}}
\end{alignat*}
There the initial model of a theory $\Th$ is denoted by $\Init_{\Th}$, $\Th_{\Pi,\BoolTy}$ is the small type theory that we consider in this paper, $\{\diamond\}$ is the terminal category, $\CRen$ is the category of renamings of $\Init_{\Th_{\Pi,\BoolTy}}$, $\square$ is the category of cubes, and $\mathcal{A}_{\square}$ is the category of cubical atomic contexts of~\cite{NormalizationCTT}.

We note some similarity with the \emph{worlds} of Twelf~\cite{Twelf} and with the \emph{context schemas} of Beluga~\cite{Beluga}.
The worlds and context schemas can be seen as descriptions of full subcategories spanned by contexts that are generated by a class of context extensions.
Our approach is more general, as we are not restricted to full subcategories.

In~\cite{SynCategoriesSketchingAndAdequacy}, an argument is made for the use of locally cartesian closed categories instead of Uemura's representable map categories in the semantics of type theories.
Using locally cartesian closed categories means that contexts can be extended by arbitrary judgments.
Indeed, the induction principle that we associate to $F : \CC \to \CS$ is left unchanged if $\CS$ is faithfully embedded into a category with additional context extensions.
However it depends on the context extensions of $\CC$; for instance canonicity is provable using $\{\diamond\} \to \Init_{\Th_{0}}$ only because $\{\diamond\}$ is not equipped with any way to extend contexts.
Thus the general notion of context extension remains important.

\paragraph*{Higher-order abstract syntax}

Higher-order abstract syntax~\cite{HOAS} is an encoding of bindings that relies on the binding structure of an ambient language.
It is closely related to Logical Frameworks~\cite{LF}.
As shown in~\cite{SemAnalysisHOAS} for the untyped lambda calculus, higher-order abstract syntax can be given semantics using presheaf categories.

The equivalence of a higher-order presentation of syntax with another presentation is usually called \emph{adequacy}.
Hofmann identified the crucial property justifying the adequacy of the higher-order presentation of untyped or simply-typed syntax: given a representable presheaf $\yo_{A}$ of a category $\CC$ with products, the presheaf exponential $(\yo_{A} \Ra B)$ can be computed as $\abs{\yo_{A} \Ra B}_{\Gamma} \triangleq \abs{B}_{\Gamma \times A}$.
This can be generalized to dependently typed syntax by considering locally representable presheaves.

The internal language of presheaf categories yields a definition of Categories with Families equipped with type-theoretic operations that are automatically stable under substitutions.
This gives a nice setting to work with a single model of type theory as used in e.g.\ \cite{paoloHOAS,awodey_2018}.
However, it does not immediately give a way to describe the general semantics of a type theory, since different models may live over different presheaf categories.
We solve this problem by using Multimodal Type Theory.

\paragraph*{Multimodal Type Theory}

The action on types of a morphism $F : \CC \to \CD$ of models can be seen as a natural transformation $F^{\Ty} : \Ty^{\CC} \to F^{\ast}\ \Ty^{\CD}$, where $F^{\ast} : \CPsh^{\CD} \to \CPsh^{\CC}$ is precomposition by $F$.
The action on terms is harder to describe.
As terms are dependent over types, we essentially need to extend $F^{\ast}$ to dependent presheaves.
The correct way to do this is to see $F^{\ast}$ as a dependent right adjoint~\cite{DRAs}; it satisfies a universal property that can be axiomatized and yields a modal extension of the internal languages of $\CPsh^{\CC}$ and $\CPsh^{\CD}$.

The actions of $F$ on types and terms are described in this extended language by:
\begin{alignat*}{3}
  & F^{\Ty} && :{ } && \forall (A : \Ty^{\CC})\ \mmod{F^{\ast}} \to \Ty^{\CD}, \\
  & F^{\Tm} && :{ } && \forall (A : \Ty^{\CC})\ (a : \Tm^{\CC}\ A)\ \mmod{F^{\ast}} \to \Tm^{\CD}\ (F^{\Ty}\ A\ \mmod{F^{\ast}}).
\end{alignat*}
where $\mmod{F^{\ast}}$ is an element of the syntax of dependent right adjoints that transitions between the presheaf models $\CPsh^{\CC}$ and $\CPsh^{\CD}$.
Multimodal Type Theory~\cite{MTT} is a further extension of this language that can deal with multiple dependent right adjoints at the same time.

Our strategy is to axiomatize just the structure and properties of the models, categories and functors that we need, so as to be able to perform most constructions internally to Multimodal Type Theory.
We describe our variant of the syntax of a dependent right adjoint in \cref{sec:dras}, and of the syntax of Multimodal Type Theory in \cref{sec:mtt}.

Other kinds of modalities have been used similar purposes in related work.
In~\cite{SemAnalysisContextualTypes}, the flat modality of crisp type theory is used to characterize the closed terms internally to a presheaf category.
In~\cite{NormalizationCTT}, a pair of open and closed modalities correspond respectively to the syntactic and semantic components of constructions performed internally to a glued topos.

One of the advantages of Multimodal Type Theory over other approaches is that additional modes and modalities can be added without requiring modification to constructions that rely on a fixed set of modes and modalities.

\paragraph*{Categorical gluing}

Some of the previous work on the metatheory of type theory has focused on the relation between logical relations and categorical gluing.
Some general gluing constructions have been given~\cite{GluingTT, GluingFlatFunctors}.
The input of these general gluing constructions is a suitable functor $F : \CC \to \CD$, where $\CC$ is a syntactic model of type theory, and $\CD$ is a semantic category (for instance a topos, or a model of another type theory with enough structure).
Gluing then provides a new glued model $\CP$ of the type theory, that combines the syntax of $\CC$ with semantic information from $\CD$.
Canonicity for instance can be proven by gluing along the global section functor $\Init_{\Th} \to \CSet$.
However the known proofs of normalization~\cite{NormalizationByGluingFreeLT,altenkirch_et_al:LIPIcs:2016:5972,CoquandNormalization} that rely on categorical gluing are not immediate consequences of these general constructions.

In the present work, we see the constructions of the glued category $\CP$ and of its type-theoretic structures as fundamentally different constructions.
We rely on the same base category $\CP$; but we equip it with type-theoretic structure using a different construction, that does not necessarily involve logical relations.

As mentioned earlier, the input data for the relative induction principles are displayed models without context extensions.
One of the central results of our work is that any displayed model without context extensions can be replaced by a displayed model with context extensions over a different base category.
In our proof this different base category is the glued category $\CP$.
However the concrete definition of $\CP$ does not matter in applications: the only thing that matters is that $\CP$ is equipped with a suitable replacement of the input displayed model without context extensions.

\paragraph*{Contributions}

Our main contribution is the statement and proofs of relative induction principles over the syntax of dependent type theory (\cref{sec:disp_models_wo_exts}), which take into account the fact that the results of induction should hold over a category $\CC$ with a functor into the syntactic category of the theory.
These induction principles are described using the new semantic notions of displayed models without context extensions and relative sections.

We show that the interpretation of dependent right adjoints and Multimodal Type Theory in diagrams of presheaf categories gives an internal language that is well-suited to the definitions of notions related to the semantics of syntax with bindings, including the notions of models, morphisms of models, the rest of the $2$-categorical structure of models, displayed models (both with and without context extensions), sections of displayed models (Sections \ref{sec:int_lang_psh}--\ref{sec:dras})
We never have to prove explicitly that any construction is stable under substitutions.

We show the application of our relative induction principles through the following examples in Section \ref{sec:applications}. We explain in detail an abstract version of Coquand's canonicity proof \cite{CoquandNormalization}, normalisation \cite{altenkirch_et_al:LIPIcs:2016:5972,CoquandNormalization} which we detail in Appendix \ref{sec:normalization}, canonicity \cite{coquand2021canonicity} and normalisation \cite{NormalizationCTT} for cubical type theory, syntactic parametricity \cite{bernardy12parametricity} and conservativity of ETT over ITT \cite{hofmann95conservativity}.



\section{Internal language of presheaf categories and models of type theory}\label{sec:int_lang_psh}

\begin{itemize}
  \item We work in a constructive metatheory, with a cumulative hierarchy $(\SSet_{i})$ of universes.
  \item If $\CC$ is a small category, we write $\abs{\CC}$ for its set of objects and $\CC(x \to y)$ for the set of morphisms from $x$ to $y$.
        We may write $(x : \CC)$ (or $(x : \CC^{\op})$) instead of $(x : \abs{\CC})$ to indicate that the dependence on $x$ is covariant (or contravariant).

        We write $(f \cdot g)$ or $(g\circ f)$ for the composition of $f : \CC(x \to y)$ and $g : \CC(y \to z)$.
  \item We rely on the internal language of presheaf categories. Given a small category $\CC$, the presheaf category $\CPsh^{\CC}$ is a model of extensional type theory, with a cumulative hierarchy of universes $\SPsh^{\CC}_{0} \subset \SPsh^{\CC}_{1} \subset \cdots \subset \SPsh^{\CC}_{i} \subset \cdots$, dependent functions, dependent sums, quotient inductive-inductive types, extensional equality types, \etc.
        For each of our definitions, propositions, theorems, \etc, we specify whether it should be interpreted externally or internally to some presheaf category.
  \item The Yoneda embedding is written $\yo : \CC \to \CPsh^{\CC}$.
        We denote the restriction of an element $x : \abs{X}_{\Gamma}$ of a presheaf $X$ along a morphism $\rho : \CC(\Delta \to \Gamma)$ by $x[\rho]_X : \abs{X}_\Delta$.
\end{itemize}

\subsection{Locally representable presheaves}

The notion of locally representable presheaf is the semantic counterpart of the notion of context extension.

\begin{definition} \label{locally-representable}
  Let $X$ be a presheaf over a category $\CC$ and $Y$ be a dependent presheaf over $X$.
  We say that $Y$ is \defemph{locally representable} if, for every $\Gamma : \abs{\CC}$ and $x : \abs{X}_{\Gamma}$, the presheaf
  \begin{alignat*}{3}
    & Y_{\mid x} && :{ } && \forall (\Delta : \CC^{\op}) (\rho : \CC(\Delta \to \Gamma)) \to \CSet \\
    & \abs{Y_{\mid x}}\ \rho && \triangleq{ } && \abs{Y}_{\Delta}\ (x[\rho])
  \end{alignat*}
  over the slice category $(\CC / \Gamma)$ is representable.

  In that case, we have, for every $\Gamma$ and $x$, an \defemph{extended context} $(\Gamma \rhd Y_{\mid x})$, a \defemph{projection map} $\bm{p}^{Y}_{x} : (\Gamma \rhd Y_{\mid x}) \to \Gamma$ and a \defemph{generic element} $\bm{q}^{Y}_{x} : \abs{Y_{\mid x}}\ \bm{p}_{x}^{Y}$ such that for every $\sigma : \Delta \to \Gamma$ and $y : \abs{Y_{\mid x}}\ \sigma$, there is a unique \defemph{extended morphism} $\angles{\sigma,y} : \Delta \to (\Gamma \rhd Y_{\mid x})$ such that $\angles{\sigma,y} \cdot \bm{p}^{Y}_{x} = \sigma$ and $\bm{q}^{Y}_{x}[\angles{\sigma,y}] = y$.
  \lipicsEnd
\end{definition}

Up to the correspondence between dependent presheaves and their total maps, locally representable dependent presheaves are also known as \emph{representable natural transformations} \cite{awodey_2018}.
We read this definition in a structured manner, with a local representability structure consisting of a choice of representing objects in the above definition.
The notion of local representability is \emph{local}: the restriction map from local representability structures for a dependent presheaf $Y$ over $X$ to coherent families of local representability structures of $Y_{\mid x}$ over $\yo_{\Gamma}$ for $x : \abs{X}_{\Gamma}$ is invertible.%
\footnote{Locality also holds if we consider local representability as a property.}

Assume that $\CC$ is an $i$-small category.
Internally to $\CPsh^{\CC}$ there is then, for every universe level $j$, a family $\isRep : \SPsh^{\CC}_{j} \to \SPsh^{\CC}_{\max(i, j)}$ of local representability structures over $j$-small presheaf families.
Due to the above locality property, we have for a dependent presheaf $Y$ over $X$ that elements of $X \vdash \isRep(Y)$ correspond to witnesses that $Y$ is locally representable over $X$.
This leads to universes $\SRepPsh^{\CC}_{j} \triangleq (A : \SPsh^{\CC}_{j}) \times \isRep\ A$ of $j$-small locally representable presheaf families.
As an internal category, it is equivalent to the $j$-small one that at $\Gamma : \abs{\CC}$ consists of an element of the slice of $\CC$ over $\Gamma$ together with a choice of base changes along any map $\CC(\Delta \to \Gamma)$. 

An alternative semantic for the presheaf $(y : Y) \to Z\ y$ of dependent natural transformations from $Y$ to $Z$ can be given when $Y$ is locally representable over $X : \CPsh^{\CC}$.
We could define $\abs{(y : Y) \to Z\ y}_{\Gamma}\ x \triangleq \abs{Z}_{\Gamma \rhd Y_{\mid x}}\ (x[\bm{p}^{Y}_{x}], \bm{q}^{Y}_{x})$.
This definition satisfies the universal property of the presheaf of dependent natural transformations from $Y$ to $Z$, and is therefore isomorphic to its usual definition.
The alternative definition admits a generalized algebraic presentation, which is important to justify the existence of initial models.

\subsection{Internal definition of models}

Our main running example is the theory $\Th_{\Pi,\BoolTy}$ of a family equipped with $\Pi$-types and a boolean type.
An internal model of $\Th_{\Pi,\BoolTy}$ in a presheaf category $\CPsh^{\CC}$ consists of the following elements.
\begin{alignat*}{3}
  & \Ty && :{ } && \SPsh^{\CC} \\
  & \Tm && :{ } && \Ty \to \SRepPsh^{\CC} \\
  & \Pi && :{ } && \forall (A : \Ty) (B : \Tm\ A \to \Ty) \to \Ty \\
  & \app && :{ } && \forall A\ B \to \Tm\ (\Pi\ A\ B) \simeq ((a : \Tm\ A) \to \Tm\ (B\ a))
\end{alignat*}
\begin{alignat*}{3}
  & \BoolTy && :{ } && \Ty \\
  & \true,\false && :{ } && \Tm\ \BoolTy \\
  & \elimBool && :{ } && \forall (P : \Tm\ \BoolTy \to \Ty)\ (t : \Tm\ (P\ \true))\ (f : \Tm\ (P\ \false))\ (b : \Tm\ \BoolTy) \to \Tm\ (P\ b) \\
  & - && :{ } && \elimBool\ P\ t\ f\ \true = t \\
  & - && :{ } && \elimBool\ P\ t\ f\ \false = f
\end{alignat*}
The inverse of $\app$ is written $\lam : \forall A\ B \to ((a : \Tm\ A) \to \Tm\ (B\ a)) \to \Tm\ (\Pi\ A\ B)$.

A model of $\Th_{\Pi,\BoolTy}$ is a category $\CC$ equipped with a terminal object and with a global internal model of $\Th_{\Pi,\BoolTy}$ in $\CPsh^{\CC}$.

\begin{remark}\label{rmk:QIIT}
  If we unfold the above internal definitions in presheaves, we see that a model of $\Th_{\Pi,\BoolTy}$ is the same externally as an algebra for the signature of a quotient inductive-inductive type (QIIT)~\cite{QIITs} describing $\Th_{\Pi,\BoolTy}$.
  That QIIT is significantly more verbose because it has sorts of contexts and substitutions and, for every component of the model, separately states the action at each context and coherent action of or coherence under substitution.
  The notion of morphism of models we will define in \cref{subsec:morphisms} unfolds externally to the verbose notion of algebra morphism for this QIIT, except that we do not require context extension to be preserved strictly.
  The same remark holds for the notion of displayed model to be defined in \cref{sec:disp_models_wo_exts}.
\end{remark}

We have a $(2,1)$-category $\CMod_{\Th_{\Pi,\BoolTy}}$ of models.
The morphisms are functors equipped with actions on types and terms that preserve the terminal object and the context extensions up to isomorphisms and the operations $\Pi$, $\app$, $\BoolTy$, $\true$, $\false$ and $\elimBool$ strictly.
The $2$-cells are the natural isomorphisms between the underlying functors.

We have just given an internal definition of the objects of $\CMod_{\Th_{\Pi,\BoolTy}}$ in the language of presheaf categories; we will give internal definitions of the other components using dependent right adjoints.

\subsection{Sorts and derived sorts}

A base sort of a CwF $\CC$ is a (code for a) presheaf (in $\CPsh^{\CC}$) of the form $\Ty$ or $\Tm(-)$.
The derived sorts are obtained by closing the base sorts under dependent products with arities in $\Tm(-)$.
A derived sort is either a base sort, or a presheaf of the form $(a : \Tm(-)) \to X(a)$ where $X(a)$ is a derived sort.
A derived sort can be written in the form $[X]Y$ where $X$ is a telescope of types and $Y$ is a base sort that depends on $X$.

The type of an argument of a type-theoretic operation or equation is always a derived sort.
We often omit dependencies when writing derived sorts; \eg{} we write $[\Tm]\Ty$ for the derived sort of the second argument of $\Pi$.


\section{Dependent Right Adjoints and morphisms of models}\label{sec:dras}

In this section, we review the syntax and semantics of dependent right adjoints (DRAs)~\cite{DRAs}, and use the syntax of the dependent right adjoint $(F_{!} \dashv \modcolor{F^{\ast}})$ to give an internal encoding of the notion of morphism of models of $\Th_{\Pi,\BoolTy}$.
Multimodal Type Theory is only needed for some of the proofs and constructions performed in the appendix.

\subsection{Dependent Right Adjoints}

Fix a functor $F : \CC \to \CD$.
The precomposition functor $F^{\ast} : \CPsh^{\CD} \to \CPsh^{\CC}$ has both a left adjoint $F_{!} : \CPsh^{\CC} \to \CPsh^{\CD}$ and a right adjoint $F_{\ast} : \CPsh^{\CC} \to \CPsh^{\CD}$.
The functors $F^{\ast}$ and $F_{\ast}$ are not only right adjoints of $F_{!}$ and $F^{\ast}$, they are dependent right adjoints, which means that they admit actions on the types and terms of the presheaf models $\CPsh^{\CC}$ and $\CPsh^{\CD}$ that interact with the left adjoints.
We distinguish the functor $F^{\ast}$ from the dependent right adjoint $\modcolor{F^{\ast}}$ by using different colors.
The dependent adjunction $(F^{\ast} \dashv \modcolor{F_{\ast}})$ is constructed in~\cite[Lemma 8.2]{MTT}, whereas $(F_{!} \dashv \modcolor{F^{\ast}})$ is constructed in~\cite[Lemma 2.1.4]{MTTtechreport}.
We recall their constructions in \cref{sec:dra_constructions}.

We focus on the description of the dependent right adjoint $\modcolor{F^{\ast}}$ as a syntactic and type-theoretic operation.
For every presheaf $X : \CPsh^{\CC}$ and dependent presheaf $A$ over $F_{!}\ X$, we have a dependent presheaf $\modcolor{F^{\ast}}\ A$ over $X$, such that elements of $A$ over $F_{!}\ X$ are in natural bijection with elements of $\modcolor{F^{\ast}}\ A$ over $X$.

This is analogous to the definition of $\Pi$-types: given a presheaf $X : \CPsh^{\CC}$, a dependent presheaf $Y(x)$ over the $(x : X)$ and a dependent presheaf $Z(x, y)$ over $(x : X, y : Y(x))$, the $\Pi$-type $(y : Y(x)) \to Z(x, y)$ over $(x : X)$ is characterized by the fact that its elements are in natural bijection with the elements of $Z(x,y)$ over $(x : X, y : Y(x))$.

Following this intuition, we use a similar syntax for $\Pi$-types and modalities.
We view the left adjoint $F_{!}$ as an operation on the contexts of the presheaf model $\CPsh^{\CC}$. If $(x : X)$ is a context of this presheaf model, we write $(x : X, \mmod{F^{\ast}})$ instead of $F_{!}\ X$.
Given a dependent presheaf $Y(x, \mmod{F^{\ast}})$\footnote{Here the notation $Y(x, \mmod{F^{\ast}})$ is an informal way to keep track of the fact that $Y$ is dependent over the context $(x : X, \mmod{F^{\ast}})$.} over $(x : X, \mmod{F^{\ast}})$, we write $(\mmod{F^{\ast}} \to Y(x, \mmod{F^{\ast}}))$ instead of $\modcolor{F^{\ast}}\ Y$.

We write the components of the bijection between elements of $Y(x, \mmod{F^{\ast}})$ over $(x : X, \mmod{F^{\ast}})$ and elements of $(\mmod{F^{\ast}} \to Y(x, \mmod{F^{\ast}}))$ over $(x : X)$ similarly to applications and $\lambda$-abstractions.
If $y(x, \mmod{F^{\ast}})$ is an element of $Y(x, \mmod{F^{\ast}})$ over $(x : X, \mmod{F^{\ast}})$, we write $(\lambda\ \mmod{F^{\ast}} \mapsto y(x, \mmod{F^{\ast}}))$ for the corresponding element of $(\mmod{F^{\ast}} \to Y(x, \mmod{F^{\ast}}))$.
Conversely, given an element $f(x)$ of $(\mmod{F^{\ast}} \to Y(x, \mmod{F^{\ast}}))$ over $(x : X)$, we write $f(x)\ \mmod{F^{\ast}}$ for the corresponding element of $Y(x, \mmod{F^{\ast}})$.
There is a $\beta$-rule $(\lambda\ \mmod{F^{\ast}} \mapsto y(x, \mmod{F^{\ast}}))\ \mmod{F^{\ast}} = y(x, \mmod{F^{\ast}})$ and an $\eta$-rule $(\lambda\ \mmod{F^{\ast}} \mapsto f(x)\ \mmod{F^{\ast}}) = f(x)$.

We may define elements of modal types by pattern matching. For instance, we may write $f(x)\ \mmod{F^{\ast}} \triangleq y(x,\mmod{F^{\ast}})$ to define $f(x)$ as the unique element satisfying the equation $f(x)\ \mmod{F^{\ast}} = y(x,\mmod{F^{\ast}})$, that is $f(x) \triangleq \lambda\ \mmod{F^{\ast}} \mapsto y(x,\mmod{F^{\ast}})$.

The operation $(\mmod{F^{\ast}} \to -)$ is a modality that enables interactions between the two presheaf models $\CPsh^{\CC}$ and $\CPsh^{\CD}$.
The symbols $\mmod{F^{\ast}}$ and $\mmod{F^{\ast}}$ and their places in the terms have been chosen to make keeping track of the modes of subterms as easy as possible.
For both symbols $\mmod{F^{\ast}}$ and $\mmod{F^{\ast}}$, the part of the term that is left of the symbol is at mode $\CPsh^{\CC}$, while the part that is right of the symbol is at mode $\CPsh^{\CD}$.
The type formers $(\mmod{F^{\ast}} \to -)$ and the term former $(\lambda\ \mmod{F^{\ast}} \mapsto -)$ go from the mode $\CPsh^{\CD}$ to $\CPsh^{\CC}$, whereas the term former $(-\ \mmod{F^{\ast}})$ goes from the mode $\CPsh^{\CC}$ to the mode $\CPsh^{\CD}$.

\subsection{Modalities are applicative functors}

As a first demonstration of the syntax of modalities, we equip the modality $(\mmod{F^{\ast}} \to -)$ with the structure of an \emph{applicative functor}~\cite{ApplicativeFunctors}, defined analogously to the \emph{reader monad} $(A \to -)$.
This structure is given by an operation
\begin{alignat*}{3}
  & (\_{} \circledast \_{}) && :{ } && \forall A\ B\ (f : \mmod{F^{\ast}} \to (a : A\ \mmod{F^{\ast}}) \to B\ \mmod{F^{\ast}}\ a) (a : \mmod{F^{\ast}} \to A\ \mmod{F^{\ast}}) \\
  &&&&& \to (\mmod{F^{\ast}} \to B\ \mmod{F^{\ast}}\ (a\ \mmod{F^{\ast}})) \\
  & f \circledast a && \triangleq{ } && \lambda\ \mmod{F^{\ast}} \mapsto (f\ \mmod{F^{\ast}})\ (a\ \mmod{F^{\ast}})
\end{alignat*}

This provides a concise notation to apply functions under the modality.
If $f$ is an $n$-ary function under the modality, and $a_{1},\dotsc,a_{n}$ are arguments under the modality, we can write the application $f \circledast a_{1} \circledast \cdots \circledast a_{n}$ instead of $(\lambda\ \mmod{F^{\ast}} \mapsto (f\ \mmod{F^{\ast}})\ (a_{1}\ \mmod{F^{\ast}})\ \cdots\ (a_{n}\ \mmod{F^{\ast}}))$.

When $f$ is a global function of the presheaf model $\CPsh^{\CD}$, we write $f \circleddollar a_{1} \circledast \cdots \circledast a_{n}$ instead of $(\lambda\ \mmod{F^{\ast}} \mapsto f) \circledast a_{1} \circledast \dotsc \circledast a_{n}$.

\subsection{Preservation of context extensions}

The last component that we need for an internal definition of morphism of models of $\Th_{\Pi}$ is an internal way to describe preservation of extended contexts of locally representable presheaves.
The preservation of context extensions can be expressed without assuming that the extended contexts actually exist, \ie{} without assuming that the presheaves are locally representable;
in that case we talk about preservation of virtual context extensions.

\begin{definition}[Internally to $\CPsh^{\CC}$]\label{def:preserves_ext}
  Let $A^{\CC} : \SPsh^{\CC}$ and $A^{\CD} : \mmod{F^{\ast}} \to \SPsh^{\CD}$ be presheaves over $\CC$ and $\CD$, and $F^{A} : \forall (a : A^{\CC}) \mmod{F^{\ast}} \to A^{\CD}\ \mmod{F^{\ast}}$ be an action of $F$ on the elements of $A^{\CC}$.
  We say that $F^{A}$ \defemph{preserves virtual context extensions} if for every dependent presheaf $P : \forall \mmod{F^{\ast}} (a : A^{\CD}\ \mmod{F^{\ast}}) \to \SPsh^{\CD}$, the canonical comparison map
  \begin{alignat*}{3}
    & \tau && :{ } && (\forall \mmod{F^{\ast}} (a : A^{\CD}\ \mmod{F^{\ast}}) \to P\ \mmod{F^{\ast}}\ a) \to (\forall (a : A^{\CC}) \mmod{F^{\ast}} \to P\ \mmod{F^{\ast}}\ (F^{A}\ a\ \ \mmod{F^{\ast}})) \\
    & \tau(p) && \triangleq{ } && \lambda a \mmod{F^{\ast}} \mapsto p\ \mmod{F^{\ast}}\ (F^{A}\ a\ \mmod{F^{\ast}})
  \end{alignat*}
  is an isomorphism.
  In other words, $F^{A}$ preserves virtual context extensions when the modality $(\mmod{F^{\ast}} \to -)$ commutes with quantification over $A^{\CC}$ and $A^{\CD}$.

  This provides a notation to define an element $p$ of $(\forall \mmod{F^{\ast}}(a : A^{\CD}\ \mmod{F^{\ast}}) \to P\ \mmod{F^{\ast}}\ a)$ using pattern matching: we write
  \[ p\ \mmod{F^{\ast}}\ (F^{A}\ a\ \mmod{F^{\ast}}) \triangleq q\ a\ \mmod{F^{\ast}} \]
  to define $p$ as the unique solution of that equation ($p = \tau^{-1}(\lambda a \mmod{F^{\ast}} \mapsto q\ a\ \mmod{F^{\ast}})$).
  \lipicsEnd
\end{definition}

In \cref{sec:preserv} we show that the internal description of preservation of context extensions coincides with the external notion of preservation up to isomorphism.

\subsection{Morphisms of models}
\label{subsec:morphisms}

Let $F : \CC \to \CD$ be a morphism of models of $\Th_{\Pi,\BoolTy}$.
We now show that its structure can fully be described in the internal language of $\CPsh^{\CC}$.

Its actions on types and terms can equivalently be given by the following global elements.
\begin{alignat*}{3}
  & F^{\Ty} && :{ } && (A : \Ty^{\CC}) \to (\mmod{F^{\ast}} \to \Ty^{\CD}) \\
  & F^{\Tm} && :{ } && \forall A\ (a : \Tm^{\CC}) \to (\mmod{F^{\ast}} \to \Tm^{\CD}\ (F^{\Ty}\ A\ \mmod{F^{\ast}}))
\end{alignat*}
The preservation of context extensions by $F$ is equivalent to the fact that $F^{\Tm}$ preserves virtual context extensions in the sense of \cref{def:preserves_ext}.
We can use that fact to obtain the following actions on derived sorts.
\begin{alignat*}{3}
  & F^{[X]\Ty} && :{ } && (A : X^{\CC} \to \Ty^{\CC}) \to (\forall \mmod{F^{\ast}}\ (x : X^{\CD}) \to \Ty^{\CD}) \\
  & F^{[X]\Tm} && :{ } && \forall A\ (a : (x : X^{\CC}) \to \Tm^{\CC}\ (A\ x)) \to (\forall \mmod{F^{\ast}}\ (x : X^{\CD}) \to \Tm^{\CD}(F^{[X]\Ty}\ A\ \mmod{F^{\ast}}(x)))
\end{alignat*}
They are defined as follows, using the pattern matching notation of \cref{def:preserves_ext}.
\begin{alignat*}{3}
  & F^{[X]\Ty}\ A\ \mmod{F^{\ast}}\ (F^{X}\ x\ \mmod{F^{\ast}}) && \triangleq{ } && F^{\Ty}\ (A\ x)\ \mmod{F^{\ast}} \\
  & F^{[X]\Tm}\ a\ \mmod{F^{\ast}}\ (F^{X}\ x\ \mmod{F^{\ast}}) && \triangleq{ } && F^{\Tm}\ (a\ x)\ \mmod{F^{\ast}}
\end{alignat*}
Finally, the preservation of the operations can simply be described by the following equations.
\begin{alignat*}{1}
  & F^{\Ty}\ (\Pi^{\CC}\ A\ B)\ \mmod{F^{\ast}} = \Pi^{\CD}\ (F^{\Ty}\ A\ \mmod{F^{\ast}})\ (F^{[\Tm]\Ty}\ B\ \mmod{F^{\ast}}) \\
  & F^{\Ty}\ (\app^{\CC}\ f\ a)\ \mmod{F^{\ast}} = \app^{\CD}\ (F^{\Tm}\ f\ \mmod{F^{\ast}})\ (F^{\Tm}\ a\ \mmod{F^{\ast}}) \\
  & F^{\Ty}\ \BoolTy^{\CC}\ \ \mmod{F^{\ast}} = \BoolTy^{\CD} \\
  & F^{\Ty}\ \true^{\CC}\ \ \mmod{F^{\ast}} = \true^{\CD} \\
  & F^{\Ty}\ \false^{\CC}\ \ \mmod{F^{\ast}} = \false^{\CD} \\
  & F^{\Ty}\ (\elimBool^{\CC}\ t\ f\ b)\ \mmod{F^{\ast}} = \elimBool^{\CD}\ (F^{[\Tm]\Ty}\ P\ \mmod{F^{\ast}})\ (F^{\Tm}\ t\ \mmod{F^{\ast}})\ (F^{\Tm}\ f\ \mmod{F^{\ast}})\ (F^{\Tm}\ b\ \mmod{F^{\ast}})
\end{alignat*}
We can then derive analogous equations for $F^{[X]\Ty}$ and $F^{[X]\Tm}$.
For instance,
\begin{alignat*}{1}
  & F^{[X]\Ty}\ (\lambda x \mapsto \Pi^{\CC}\ (A\ x)\ (B\ x))\ \mmod{F^{\ast}}\ x \\
  & \quad = \Pi^{\CD}\ (F^{[X]\Ty}\ A\ \mmod{F^{\ast}}\ x)\ (\lambda a \mapsto F^{[X,\Tm]\Ty}\ B\ \mmod{F^{\ast}}\ (x,a)).
\end{alignat*}
Indeed, by \cref{def:preserves_ext}, it suffices to show that equation when $x = F^{X}\ x'\ \mmod{F^{\ast}}$.
It then follows from the base equation for $F^{\Ty}\ (\Pi^{\CC}\ (A\ x')\ (B\ x'))$.

We can also derive strengthening equations.
For example, when $A$ does not depend on $X$, we have $F^{[X]\Ty}\ (\lambda x \mapsto A)\ \mmod{F^{\ast}} = \lambda x \mapsto F^{\Ty}\ A\ \mmod{F^{\ast}}$.

\begin{remark}
  The notion of morphism of models unfolds externally to the verbose notion of algebra morphism for the QIIT signature of \cref{rmk:QIIT}, except that we do not require context extension to be preserved strictly.
  A standard argument shows that initial algebras for the QIIT are biinitial in our sense.
  A similar remark holds for the notion of displayed model (and their sections) that will be defined in \cref{sec:disp_models_wo_exts}.
\end{remark}


\section{Relative induction principles}\label{sec:disp_models_wo_exts}

In this section we state our relative induction principles using the notion of displayed model without context extensions.
The full proofs of these relative induction principles are given in the appendix.

We fix a base model $\CS$ of $\Th_{\Pi,\BoolTy}$ and a functor $F : \CC \to \CS$.

\begin{definition}
  A \defemph{displayed model without context extensions} over $F : \CC \to \CS$ consists of the following components, specified internally to $\CPsh^{\CC}$:
  \begin{itemize}
    \item Presheaves of displayed types and terms.
          \begin{alignat*}{3}
            & \Ty^{\bullet} && :{ } && (A : \mmod{F^{\ast}} \to \Ty^{\CS}) \to \SPsh^{\CC} \\
            & \Tm^{\bullet} && :{ } && \forall A\ (A^{\bullet} : \Ty^{\bullet}\ A) (a : \mmod{F^{\ast}} \to \Tm^{\CS}\ (A\ \mmod{F^{\ast}})) \to \SPsh^{\CC}
          \end{alignat*}

          They correspond to the \emph{motives} of an induction principle.

    \item Displayed variants of the type-theoretic operations of $\Th_{\Pi,\BoolTy}$.
          They are the \emph{methods} of the induction principle.
          \begin{alignat*}{3}
            & \Pi^{\bullet} && :{ } && \forall A\ B\ (A^{\bullet} : \Ty^{\bullet}\ A) (B^{\bullet} : \{a\} (a^{\bullet} : \Ty^{\bullet}\ A^{\bullet}\ a) \to \Ty^{\bullet}\ (B \circledast a)) \\
            &&&&& \to \Ty^{\bullet}\ (\Pi^{\CS} \circleddollar A \circledast B) \\
            & \app^{\bullet} && :{ } && \forall A\ B\ f\ a\ (A^{\bullet} : \Ty^{\bullet}\ A) (B^{\bullet} : \{a\} (a^{\bullet} : \Ty^{\bullet}\ A^{\bullet}\ a) \to \Ty^{\bullet}\ (B \circledast a)) \\
            &&&&& \to (\Tm^{\bullet}\ (\Pi^{\bullet}\ A^{\bullet}\ B^{\bullet})\ f) \simeq ((a^{\bullet} : \Ty^{\bullet}\ A^{\bullet}\ a) \to \Tm^{\bullet}\ (B^{\bullet}\ a^{\bullet})\ (\app^{\CS} \circleddollar f \circledast a))
          \end{alignat*}
          \begin{alignat*}{3}
            & \BoolTy^{\bullet} && :{ } && \Ty^{\bullet}\ (\lambda \mmod{F^{\ast}} \mapsto \BoolTy) \\
            & \true^{\bullet} && :{ } && \Tm^{\bullet}\ \BoolTy^{\bullet}\ (\lambda \mmod{F^{\ast}} \mapsto \true) \\
            & \false^{\bullet} && :{ } && \Tm^{\bullet}\ \BoolTy^{\bullet}\ (\lambda \mmod{F^{\ast}} \mapsto \false) \\
            & \elimBool^{\bullet} && :{ } && \forall P\ t\ f\ b\ (P^{\bullet} : \forall x\ (x^{\bullet} : \Tm^{\bullet}\ \BoolTy^{\bullet}\ x) \to \Ty^{\bullet}\ (P \circledast x)) \\
            &&&&& \phantom{\forall} (t^{\bullet} : \Tm^{\bullet}\ (P^{\bullet}\ \true^{\bullet})\ t) (f^{\bullet} : \Tm^{\bullet}\ (P^{\bullet}\ \false^{\bullet})\ f) \\
            &&&&& \to (b^{\bullet} : \Tm^{\bullet}\ \BoolTy^{\bullet}\ b) \to \Tm^{\bullet}\ (P^{\bullet}\ b^{\bullet})\ b
          \end{alignat*}

    \item Satisfying displayed variants of the type-theoretic equations\footnote{Note that these equations are well-typed because of the corresponding equations in $\CS$. As presheaves support equality reflection, we don't have to write transports.} of $\Th_{\Pi,\BoolTy}$.
          \begin{alignat*}{3}
            & \elimBool^{\bullet}\ P^{\bullet}\ t^{\bullet}\ f^{\bullet}\ \true^{\bullet} && ={ } && t^{\bullet} \\
            & \elimBool^{\bullet}\ P^{\bullet}\ t^{\bullet}\ f^{\bullet}\ \false^{\bullet} && ={ } && f^{\bullet}
            \tag*{\lipicsEnd}
          \end{alignat*}
  \end{itemize}
\end{definition}

A displayed model without context extensions has context extensions when for any $A$ and $A^{\bullet}$, the first projection map
\[ (a : \mmod{F^{\ast}} \to \Tm^{\CS}\ (A\ \mmod{F^{\ast}})) \times (a^{\bullet} : \Tm^{\bullet}\ A^{\bullet}\ a) \xrightarrow{\lambda (a,a^{\bullet}) \mapsto a} (\mmod{F^{\ast}} \to \Tm^{\CS}\ (A\ \mmod{F^{\ast}})) \]
has a locally representable domain and preserves context extensions.

In \cref{sec:sections} we give an internal definition of section of displayed models with context extensions.
It is similar to the definition of morphism of models.
The induction principle of the biinitial model $\Init_{\Th_{\Pi,\BoolTy}}$ is the statement that any displayed model with context extensions over $\Init_{\Th_{\Pi,\BoolTy}}$ admits a section.

While (displayed) models without context extensions are not well-behaved, we show that they can be replaced by (displayed) models with context extensions.%
\begin{restatable}{definition}{restateDefFactorization}\label{def:disp_model_factorization}
  A \defemph{factorization} $(\CC \xrightarrow{Y} \CP \xrightarrow{G} \CS, \CS^{\dagger})$ of a global displayed model without context extensions $\CS^{\bullet}$ over $F : \CC \to \CS$ consists of a factorization $\CC \xrightarrow{Y} \CP \xrightarrow{G} \CS$ of $F$ and a displayed model with context extensions $\CS^{\dagger}$ over $G : \CP \to \CS$, such that $Y : \CC \to \CP$ is fully faithful and equipped with bijective actions on displayed types and terms.
  \lipicsEnd{}
\end{restatable}

\begin{restatable}{construction}{restateDispReplace}\label{con:disp_replace_0}
  We construct a factorization $(\CC \xrightarrow{Y} \CP \xrightarrow{G} \CS, \CS^{\dagger})$ of any model without context extensions $\CS^{\bullet}$ over $F : \CC \to \CS$.
\end{restatable}
\begin{proof}[Construction sketch]
  We give the full construction in the appendix.
  We see $\CP$ as analogous to the presheaf category over $\CC$, but in the slice $2$-category $(\CCat / \CS)$.
  Indeed, a generalization of the Yoneda lemma holds for $Y : \CC \to \CP$.
  In particular $Y : \CC \to \CP$ is fully faithful.

  Equivalently, it could be defined as the pullback along $\yo : \CS \to \widehat{\CS}$ of the Artin gluing $\CG \to \widehat{\CS}$ of $F_{\ast} : \widehat{\CS} \to \widehat{\CC}$.

  It is well-known that given a base model $\CC$ of type theory, that model can be extended to the presheaf category $\widehat{\CC}$ in such a way that the Yoneda embedding $\yo : \CC \to \widehat{\CC}$ is a morphism of models with bijective actions on types and terms.
  This is indeed the justification for one of the intended models of two-level type theory~\cite{TwoLevelTypeTheoryAndApplications}.
  This construction does not actually depend on the context extensions in the base model $\CC$.
  The construction of the displayed model $\CS^{\dagger}$ over $G : \CP \to \CS$ is a generalization of this construction to displayed models.
\end{proof}

We now assume that we have a section $S_{0}$ of the displayed model without context extensions $\CS^{\dagger}$ constructed in \cref{con:disp_replace_0}.

In general, we want more than just the section $S_{0}$.
Indeed, if we take a type $A$ of $\CS$ over a context $F\ \Gamma$ for some $\Gamma : \abs{\CC}$, we can apply the action of $S_{0}$ on types to obtain a displayed type $S_{0}^{\Ty}\ A$ of $\CS^{\dagger}$ over $S_{0}\ (F\ \Gamma)$.
We would rather have a displayed type of $\CS^{\bullet}$ over $\Gamma$.
It suffices to have a morphism $\alpha_{\Gamma} : Y\ \Gamma \to S_{0}\ (F\ \Gamma)$.
We can then transport $S_{0}^{\Ty}\ A$ to a displayed type $(S_{0}^{\Ty}\ A)[\alpha_{\Gamma}]$ of $\CS^{\dagger}$ over $Y\ \Gamma$.
Since $Y$ is equipped with a bijective action $Y^{\Ty}$ on displayed types, this provides a displayed type $Y^{\Ty,-1}\ (S_{0}^{\Ty}\ A)[\alpha_{\Gamma}]$ of $\CS^{\bullet}$ over $\Gamma$, as desired.
In general, we want to have a full natural transformation $\alpha : Y \Ra (F \cdot S_{0})$.

It is useful to consider the universal setting under which such a natural transformation is available.
\begin{definition}
  The \defemph{displayed inserter} $\CI(\CS^{\bullet})$ is a category equipped with a functor $I : \CI(\CS^{\bullet}) \to \CC$ and with a natural transformation $\iota : (I \cdot Y) \Ra (I \cdot F \cdot S_{0})$ such that $(\iota \cdot P) = 1_{(I \cdot F)}$.
  It is the final object among such categories: given any other category $\CJ$ with $J : \CJ \to \CC$ and $\beta : (J \cdot Y) \Ra (J \cdot F \cdot S_{0})$ such that $(\bm{\beta} \cdot P) = 1_{(J \cdot S_{0})}$, there exists a unique functor $X : \CJ \to \CI(\CS^{\bullet})$ such that $J = (X \cdot I)$ and $\beta = (X \cdot \alpha)$.
  \lipicsEnd{}
\end{definition}

Internally to $\CPsh^{\CI(\CS^{\bullet})}$, we then have the following operations:
\begin{alignat*}{3}
  & S_{\iota}^{[X]\Ty} && :{ } && \forall \mmod{I^{\ast}} (A : \mmod{F^{\ast}} \to X \to \Ty)\ x\ (x^{\bullet} : X^{\bullet}\ x) \to \Ty^{\bullet}\ (A \circledast x) \\
  & S_{\iota}^{[X]\Tm} && :{ } && \forall \mmod{I^{\ast}}\ A\ (a : \forall \mmod{F^{\ast}}\ x \to \Tm\ (A\ \mmod{F^{\ast}}\ x))\ x\ (x^{\bullet} : X^{\bullet}\ x) \\
  &&&&& \to \Tm^{\bullet}\ (S_{\iota}^{[X]\Ty}\mmod{I^{\ast}}\ A\ x\ x^{\bullet})\ (a \circledast x),
\end{alignat*}
where $X^{\bullet}$ is defined by induction on the telescope $X$.
They preserve all type-theoretic operations:
\begin{alignat*}{1}
  & S_{\iota}^{[X]\Ty}\ \mmod{I^{\ast}}\ (\lambda \mmod{F^{\ast}}\ x \mapsto \Pi\ (A\ \mmod{F^{\ast}}\ x)\ (B\ \mmod{F^{\ast}}\ x))\ x^{\bullet} \\
  & \quad = \Pi^{\bullet}\ (S_{\iota}^{[X]\Ty}\ \mmod{I^{\ast}}\ A\ x^{\bullet})\ (\lambda a^{\bullet} \mapsto S_{\iota}^{[X,A]\Ty}\ \mmod{I^{\ast}}\ B\ (x^{\bullet},a^{\bullet})) \\
  & S_{\iota}^{[X]\Tm}\ \mmod{I^{\ast}}\ (\lambda \mmod{F^{\ast}}\ x \mapsto \app\ (f\ \mmod{F^{\ast}}\ x)\ (a\ \mmod{F^{\ast}}\ x))\ x^{\bullet} \\
  & \quad = \app^{\bullet}\ (S_{\iota}^{[X]\Tm}\ \mmod{I^{\ast}}\ f\ x^{\bullet})\ (S_{\iota}^{[X]\Ty}\ \mmod{I^{\ast}}\ a\ x^{\bullet}) \\
  & S_{\iota}^{[X]\Ty}\ \mmod{I^{\ast}}\ (\lambda \mmod{F^{\ast}}\ x \mapsto \BoolTy)\ x^{\bullet} = \BoolTy^{\bullet} \\
  & S_{\iota}^{[X]\Tm}\ \mmod{I^{\ast}}\ (\lambda \mmod{F^{\ast}}\ x \mapsto \true)\ x^{\bullet} = \true^{\bullet} \\
  & S_{\iota}^{[X]\Tm}\ \mmod{I^{\ast}}\ (\lambda \mmod{F^{\ast}}\ x \mapsto \false)\ x^{\bullet} = \false^{\bullet} \\
  & S_{\iota}^{[X]\Tm}\ \mmod{I^{\ast}}\ (\lambda \mmod{F^{\ast}}\ x \mapsto \elimBool\ (P\ \mmod{F^{\ast}}\ x)\ (t\ \mmod{F^{\ast}}\ x)\ (f\ \mmod{F^{\ast}}\ x)\ (b\ \mmod{F^{\ast}}\ x))\ x^{\bullet} \\
  & \quad = \elimBool^{\bullet}\ (\lambda b^{\bullet} \mapsto S_{\iota}^{[X,\Tm]\Ty}\ \mmod{I^{\ast}}\ P\ (x^{\bullet},b^{\bullet})) \\
  & \phantom{{ }\quad = \elimBool^{\bullet}\ }(S_{\iota}^{[X]\Tm}\ \mmod{I^{\ast}}\ t\ x^{\bullet})\ (S_{\iota}^{[X]\Tm}\ \mmod{I^{\ast}}\ f\ x^{\bullet})\ (S_{\iota}^{[X]\Tm}\ \mmod{I^{\ast}}\ b\ x^{\bullet})
\end{alignat*}

\begin{restatable}{definition}{restateRelSection}\label{def:rel_section}
  A \defemph{relative section} $S_{\alpha}$ of a factorization $(\CC \xrightarrow{Y} \CP \xrightarrow{G} \CS, \CS^{\dagger})$ of a displayed model without context extensions $\CS^{\bullet}$ over $F : \CC \to \CS$ consists of a section $S_{0}$ of the displayed model with context extensions $\CS^{\dagger}$ along with a natural transformation $\alpha : Y \Ra (F \cdot S_{0})$ such that $(\alpha \cdot G) = 1_{F}$, or equivalently with a section $\angles{\alpha} : \CC \to \CI(\CS^{\bullet})$ of $I : \CI(\CS^{\bullet}) \to \CC$.
  \lipicsEnd{}
\end{restatable}

A relative section $S_{\alpha}$ has actions on types and terms, obtained by pulling $S_{\iota}^{[X]\Ty}$ and $S_{\iota}^{[X]\Tm}$ along $\angles{\alpha}$.
\begin{alignat*}{3}
  & S_{\alpha}^{[X]\Ty} && :{ } && \forall (A : \mmod{F^{\ast}} \to X \to \Ty)\ x\ (x^{\bullet} : X^{\bullet}\ x) \to \Ty^{\bullet}\ (A \circledast x) \\
  & S_{\alpha}^{[X]\Tm} && :{ } && \forall A\ (a : \forall \mmod{F^{\ast}}\ x \to \Tm\ (A\ \mmod{F^{\ast}}\ x))\ x\ (x^{\bullet} : X^{\bullet}\ x) \\
  &&&&& \to \Tm^{\bullet}\ (S_{\alpha}^{[X]\Ty}\ A\ x\ x^{\bullet})\ (a \circledast x),
\end{alignat*}

A displayed model without context extension over the biinitial model does not necessarily admit a relative section; this depends on the functor $F : \CC \to \Init_{\Th}$.
Depending on the universal property of $\CC$, we need to provide some additional data in order to construct $\angles{\alpha} : \CC \to \CI(\Init_{\Th}^{\bullet})$.
Thus, we get a different induction principle for every functor $F : \CC \to \Init_{\Th}$ into $\Init_{\Th}$, which we call the induction principle relative to $F$.
We now state several of these relative induction principles, for our example type theory $\Th_{\Pi,\BoolTy}$ and for cubical type theory $\mathsf{CTT}$.

\begin{restatable}[Induction principle relative to $\{\diamond\} \to \Init_{\Th_{\Pi,\BoolTy}}$]{lemma}{restateIndTerminal}\label{lem:ind_terminal}
  Denote by $\{\diamond\}$ the terminal category (which should rather be seen here as the initial category equipped with a terminal object), and consider the functor $F : \{\diamond\} \to \Init_{\Th_{\Pi,\BoolTy}}$ that selects the empty context of $\Init_{\Th_{\Pi,\BoolTy}}$.

  Any global displayed model without context extensions over $F$ admits a relative section.
  \qed{}
\end{restatable}

\begin{definition}
  A \defemph{renaming algebra} over a model $\CS$ of $\Th_{\Pi,\BoolTy}$ is a category $\CR$ with a terminal object, along with a functor $F : \CR \to \CS$ preserving the terminal object, a locally representable dependent presheaf of variables
  \[ \Var^{\CR} : (A : \mmod{F^{\ast}} \to \Ty^{\CS}) \to \SRepPsh^{\CR} \]
  and an action on variables $\var : \forall A\ (a : \Var\ A)\ \mmod{F^{\ast}} \to \Tm^{\CS}\ (A\ \mmod{F^{\ast}})$ that preserves context extensions.

  The category of renamings $\CRen_{\CS}$ over a model $\CS$ is defined as the biinitial renaming algebra over $\CS$.
  We denote the category of renamings of the biinitial model $\Init_{\Th_{\Pi,\BoolTy}}$ by $\CRen$.
\end{definition}

\begin{restatable}[Induction principle relative to $\CRen \to \Init_{\Th_{\Pi,\BoolTy}}$]{lemma}{restateIndRenamings}\label{lem:ind_renamings}
  Let $\Init_{\Th_{\Pi,\BoolTy}}^{\bullet}$ be a global displayed model without context extensions over $F : \CRen \to \Init_{\Th_{\Pi,\BoolTy}}$, along with, internally to $\CI(\CS^{\bullet})$, a global map
  \[ \var^{\bullet} : \forall \mmod{I^{\ast}} (A : \mmod{F^{\ast}} \to \Ty) (a : \Var\ A) \to \Tm^{\bullet}\ (S_{\iota}^{\Ty}\ \mmod{I^{\ast}}\ A)\ (\var\ a). \]

  Then there exists a relative section $S_{\alpha}$ of $\Init_{\Th_{\Pi,\BoolTy}}^{\bullet}$.
  \qed{}
\end{restatable}
The relative section also satisfies a computation rule that relates $S_{\alpha}^{\Tm}\ (\var_{A}\ a)$ and $\var^{\bullet}$.

We also state relative induction principles that can be used to prove canonicity and normalization of cubical type theory.
\begin{definition}
  A cubical CwF is a CwF $\CC$ equipped with a locally representable interval presheaf with two endpoints
  \begin{alignat*}{3}
    & \mathbb{I}^{\CC} && :{ } \SRepPsh^{\CC}, \\
    & 0^{\CC}, 1^{\CC} && :{ } \mathbb{I}^{\CC}.
    \tag*{\lipicsEnd}
  \end{alignat*}
\end{definition}

A model of cubical type theory ($\mathsf{CTT}$) is a cubical CwF equipped with some choice of type-theoretic structures, such as $\Pi$-types, path types, glue types, \etc.

\begin{definition}
  A (cartesian) \defemph{cubical algebra} over a model $\CS$ of $\mathsf{CTT}$ is a category $\CC$ with a terminal object, along with a functor $F : \CC \to \CS$ preserving the terminal object, a locally representable interval presheaf $\mathbb{I}^{\CC} : \SRepPsh^{\CC}$ with two endpoints $0^{\CC}, 1^{\CC} : \mathbb{I}^{\CC}$ and an action $\mathsf{int} : \mathbb{I}^{\CC} \to \mmod{F^{\ast}} \to \mathbb{I}^{\CS}$ that preserves context extensions and the endpoints.

  The category of cubes $\square_{\CS}$ over a model $\CS$ is defined as the biinitial cubical algebra over $\CS$.
  We denote by $\square$ the category of cubes of the biinitial model $\Init_{\mathsf{CTT}}$ of cubical type theory.
  \lipicsEnd{}
\end{definition}

\begin{restatable}[Induction principle relative to $\square \to \Init_{\mathsf{CTT}}$]{lemma}{restateIndCube}\label{lem:ind_cubes}
  Let $\Init_{\mathsf{CTT}}^{\bullet}$ be a global displayed model without context extensions over $F : \square \to \Init_{\mathsf{CTT}}$, along with a map
  \[ \mathsf{int}^{\bullet} : (i : \mathbb{I}^{\square}) \to \mathbb{I}^{\bullet}\ (\mathsf{int}\ i) \]
  such that $\mathsf{int}^{\bullet}\ 0^{\square} = 0^{\bullet}$ and $\mathsf{int}^{\bullet}\ 1^{\square} = 1^{\bullet}$.

  Then there exists a relative section $S_{\alpha}$ of $\Init_{\mathsf{CTT}}^{\bullet}$.
  \qed{}
\end{restatable}

\begin{definition}
  A (cartesian) \defemph{cubical atomic algebra} over a model $\CS$ of $\mathsf{CTT}$ is a category $\CC$ with a terminal object, along with a functor $F : \CC \to \CS$ preserving the terminal object and with the structures of a cubical algebra ($\mathbb{I}^{\CC}, 0^{\CC}, 1^{\CC}, \mathsf{int}$) and of a renaming algebra $(\Var^{\CC}, \var)$.

  The category of cubical atomic contexts $\CA_{\square}$ is the biinitial cubical algebra over the biinitial model $\Init_{\mathsf{CTT}}$ of cubical type theory.
  \lipicsEnd{}
\end{definition}

\begin{restatable}[Induction principle relative to $\CA_{\square} \to \Init_{\mathsf{CTT}}$]{lemma}{restateIndAtomicCube}\label{lem:ind_renaming_cubes}
  Let $\Init_{\mathsf{CTT}}^{\bullet}$ be a global displayed model without context extensions over $F : \CA_{\square} \to \Init_{\mathsf{CTT}}$, along with, internally to $\CPsh^{\CI(\Init_{\mathsf{CTT}}^{\bullet})}$, global maps
  \begin{alignat*}{3}
    & \var^{\bullet} && :{ } && \forall \mmod{I^{\ast}}\ (A : \mmod{F^{\ast}} \to \Ty) (a : \Var\ A) \to \Tm^{\bullet}\ (S_{\iota}^{\Ty}\ \mmod{I^{\ast}}\ A)\ (\var\ a), \\
    & \mathsf{int}^{\bullet} && :{ } && \forall \mmod{I^{\ast}}\ (i : \mathbb{I}^{\CA_{\square}}) \to \mathbb{I}^{\bullet}\ (\mathsf{int}\ i),
  \end{alignat*}
  such that $\mathsf{int}^{\bullet}\ \mmod{I^{\ast}}\ 0^{\CA_{\square}} = 0^{\bullet}$ and $\mathsf{int}^{\bullet}\ \mmod{I^{\ast}}\ 1^{\CA_{\square}} = 1^{\bullet}$.

  Then there exists a relative section $S_{\alpha}$ of $\Init_{\mathsf{CTT}}^{\bullet}$.
  \qed{}
\end{restatable}


\section{Applications}\label{sec:applications}

We give a few applications of our relative induction principles.
Only the canonicity proof is detailed here; for most of the other proofs we only give the definition of the displayed types.
A more detailed normalization proof is given in~\cref{sec:normalization}.

\subsection{Canonicity}\label{sec:applCanon}
In order to prove canonicity for $\Init_{\Th_{\Pi,\BoolTy}}$, we construct a displayed model without context extensions $\Init_{\Th_{\Pi,\BoolTy}}^{\bullet}$ over $F : \{\diamond\} \to \Init_{\Th_{\Pi,\BoolTy}}$, so as to apply \cref{lem:ind_terminal} to it.
It is defined in the in the internal language of $\CPsh^{\{\diamond\}}$ ($= \CSet$).

A type of $\Init_{\Th_{\Pi,\BoolTy}}^{\bullet}$ displayed over a type $A : \mmod{F^{\ast}} \to \Ty$ is a set-valued proof-relevant logical predicate over the terms of type $A$.
A term of $\Init_{\Th_{\Pi,\BoolTy}}^{\bullet}$ of type $A^{\bullet}$ displayed over a term $a : \mmod{F^{\ast}} \to \Tm\ (A\ \mmod{F^{\ast}})$ is an witness of the fact that $a$ satisfies the predicate $A^{\bullet}$.
\begin{alignat*}{3}
  & \Ty^{\bullet}\ A && \triangleq{ } && (a : \mmod{F^{\ast}} \to \Tm\ (A\ \mmod{F^{\ast}})) \to \SSet \\
  & \Tm^{\bullet}\ A^{\bullet}\ a && \triangleq{ } && A^{\bullet}\ a
\end{alignat*}

The logical predicate over $\Pi\ A\ B$ characterizes the functions that preserve the logical predicate of $A$ and $B$.
\begin{alignat*}{3}
  & \Pi^{\bullet}\ A^{\bullet}\ B^{\bullet} && \triangleq{ } && \lambda\ f \mapsto (\forall a\ (a^{\bullet} : A^{\bullet}\ a) \to B^{\bullet}\ (\app \circleddollar f \circledast a)) \\
  & \app^{\bullet}\ f^{\bullet}\ a^{\bullet} && \triangleq{ } && f^{\bullet}\ a^{\bullet}
\end{alignat*}

Finally, $\BoolTy^{\bullet} : (\mmod{F^{\ast}} \to \Tm\ \BoolTy) \to \SSet$ characterizes canonical natural numbers, and is defined as the inductive family generated by $\true^{\bullet} : \BoolTy^{\bullet}\ (\lambda\ \mmod{F^{\ast}} \mapsto \true)$ and $\false^{\bullet} : \BoolTy^{\bullet}\ (\lambda\ \mmod{F^{\ast}} \mapsto \false)$.
The displayed natural number eliminator $\elimBool^{\bullet}$ is then obtained from the elimination principle of $\BoolTy^{\bullet}$.

\Cref{lem:ind_terminal} now provides a relative section $S_{\alpha}$ of $\Init_{\Th_{\Pi,\BoolTy}}^{\bullet}$.

Internally to $\CPsh^{\{\diamond\}}$, take any global boolean term $(b : \mmod{F^{\ast}} \to \Tm\ \BoolTy)$.
Note that since $F : \{\diamond\} \to \Init_{\Th_{\Pi,\BoolTy}}$ selects the empty context, the dependent right adjoint $\modcolor{F^{\ast}}$ restricts presheaves over $\Init_{\Th_{\Pi,\BoolTy}}$ to the empty context.
Thus $b$ is indeed a closed boolean term.

We can apply the action of the relative section $S_{\alpha}$ to $b$.
We obtain an element $S_{\alpha}^{\Tm}\ b$ of $\Ty^{\bullet}\ (S_{\alpha}^{\Ty}\ (\lambda\ \mmod{F^{\ast}} \mapsto \BoolTy))\ b$.
We compute $S_{\alpha}^{\Ty}\ (\lambda\ \mmod{F^{\ast}} \mapsto \BoolTy) = \BoolTy^{\bullet}$.
Therefore we have an element of $\BoolTy^{\bullet}\ b$.
This proves that $b$ is canonical.

\subsection{Normalization}\label{sec:applNorm}

The normalization proof of~\cite{CoquandNormalization} can be expressed using the induction principle relative to $F : \CRen \to \Init_{\Th_{\Pi,\BoolTy}}$, namely~\cref{lem:ind_renamings}.

Internally to $\CPsh^{\CRen}$, we have inductively defined families $\NfTy : \Ty \to \SPsh^{\CRen}$, $\Nf : \forall A \to (\mmod{F^{\ast}} \to \Tm\ (A\ \mmod{F^{\ast}})) \to \SPsh^{\CRen}$ and $\Ne : \forall A \to (\mmod{F^{\ast}} \to \Tm\ (A\ \mmod{F^{\ast}})) \to \SPsh^{\CRen}$ describing the normal forms of types, the normal forms of terms and the neutral terms.

A displayed type $A^{\bullet} : \Ty^{\bullet}\ A$ is a tuple $(A^{\bullet}_{0},A^{\bullet}_{p},A^{\bullet}_{\alpha},A^{\bullet}_{\beta})$ where:
\begin{itemize}
  \item $A^{\bullet}_{0} : \NfTy\ A$ is a witness that the type $A$ admits a normal form;
  \item $A^{\bullet}_{p} : (\mmod{F^{\ast}} \to \Tm\ (A\ \mmod{F^{\ast}})) \to \SPsh^{\CRen}$ is a proof-relevant logical predicate over terms of type $A$, valued in presheaves over $\CRen$;
  \item $A^{\bullet}_{\alpha} : \forall a \to \Ne\ a \to A^{\bullet}_{p}\ a$ is a natural transformation, usually called \emph{reflect} or \emph{quote}, witnessing the fact that all neutral terms satisfy the logical predicate $A^{\bullet}_{p}$;
  \item $A^{\bullet}_{\beta} : \forall a \to A^{\bullet}_{p}\ a \to \Nf\ a$ is a natural transformation, usually called \emph{reify} or \emph{unquote}, witnessing the fact that terms satisfying the logical predicate $A^{\bullet}_{p}$ admit normal forms.
\end{itemize}

\subsection{Canonicity for cubical type theory}\label{sec:applCanCTT}

The proof of canonicity for cubical type theory from~\cite{HoCanonicityCTT} can be reformulated using the induction principle relative to $F : \square \to \Init_{\mathsf{CTT}}$, \ie{}  \cref{lem:ind_cubes}.
Internally to $\CPsh^{\square}$, we have a universe $\UU^{\mathsf{fib}}$ of fibrant cubical sets.
A displayed type $A^{\bullet} : \Ty^{\bullet}\ A$ is a logical predicate valued in fibrant cubical sets:
\[ A^{\bullet} : (\mmod{F^{\ast}} \to \Tm\ (A\ \mmod{F^{\ast}})) \to \UU^{\mathsf{fib}}. \]

\subsection{Normalization for cubical type theory}\label{sec:applNormCTT}

The proof of normalization for cubical type theory from~\cite{NormalizationCTT} can be reformulated using the induction principle relative to $F : \CA_{\square} \to \Init_{\mathsf{CTT}}$, that is \cref{lem:ind_renaming_cubes}.

\subsection{Syntactic parametricity}\label{sec:applPar}

Syntactic parametricity can be described by a displayed model without context extensions over $\id : \Init_{\Th} \to \Init_{\Th}$.
A displayed type $A^{\bullet} : \Ty^{\bullet}\ A$ is a type-valued logical predicate $A^{\bullet} : \Tm\ A \to \Ty$.

However, we do not have a relative section of this displayed model.
We have the displayed inserter category $\CI(\Init_{\Th}^{\bullet})$; but the map $I : \CI(\Init_{\Th}^{\bullet}) \to \Init_{\Th}$ does not admit a section.
Various applications of syntactic parametricity can use various functors into $\CI(\Init_{\Th}^{\bullet})$.
For instance, if we only care about closed terms, we can consider the functor $\{\diamond\} \to \CI(\Init_{\Th}^{\bullet})$.
This is sufficient to prove that any closed term $f : \Tm\ ((A : \UU) \to A \to A)$ is homotopic to the polymorphic identity function.

\subsection{Conservativity}\label{sec:applConserv}

The conservativity of extensional type theory (ETT) over intensional type theory (ITT) can be obtained using an induction principle relative to  $F : \CRen_{\mathsf{ITT}} \to \Init_{\mathsf{ETT}}$.

The proof involves some congruence $(\sim)$ over $\Init_{\mathsf{ITT}}$; this consists of equivalence relations on types and terms preserving all type-theoretic operations.
A displayed type $A^{\bullet} : \Ty^{\bullet}\ A$ is an element of the quotient of $(A_{0} : \Ty_{\mathsf{ITT}}) \times (F^{\Ty}\ A_{0} = A)$ by this relation.
Displayed terms are defined similarly.
Hofmann's proof involves the quotient model $(\Init_{\mathsf{ITT}}/\sim)$, but by working internally to $\CPsh^{\CRen_{\mathsf{ITT}}}$ we can avoid the (easy but tedious) construction of that model.


\section{Future work}


While we have focused on the type theory $\Th_{\Pi,\BoolTy}$ in this document, we hope that it is clear that these constructions generalize to other type theories.
Nevertheless, it would be good to actually prove that all of these constructions can be done for arbitrary type theories.
In~\cite{QIITs}, a syntactic definition of quotient inductive-inductive type signature is given, along with semantics.
It should be possible to extend this approach and give general definitions of models, morphisms, displayed models (without context extensions), \etc, for arbitrary type theory signatures following \cite{paoloHOAS}.
Other definitions of the general notion of type theory have been proposed recently~\cite{GeneralDefinitionDTT,GeneralFrameworkSemanticsTT}.
One advantage of the approach of~\cite{QIITs} is that its semantics are given by induction on the syntax of signatures; and thus the definitions of models, morphisms, \etc, for a given type theory signature can be computed.


The current proof assistants based on dependent type theory natively support various variants of inductive types.
We believe that the ideas presented in this paper could help towards the implementation of proof assistants that natively support syntax with bindings.


We have used our framework to give short proofs of canonicity, normalization, parametricity and conservativity results for dependent type theory.
We see them as non-trivial results that are also well-understood; and thus serve as good benchmarks for our induction principles.
We hope to apply this framework the proof of novel results in the future.


We would also like to extend this work to other kinds of context extensions and binding structures, such as affine binding structures.
An affine variable cannot be duplicated (in the absence of additional structure) and can therefore be used at most once.
This should give a description of the category of weakenings as the initial object of some category.
The category of weakenings is similar to the category of renamings, but without the ability to duplicate variables.
Using the category of weakenings in a normalization proof allows for non-linear equations in the type theory, such as the group equation $x \cdot x^{-1} = 1$.
The internal language of presheaves over the category of weakenings is also the right setting for proving the decidability of equality.


\newpage

\bibliography{main}

\appendix

\section{The dependent right adjoints \texorpdfstring{$F^{\ast}$}{F\textasciicircum *} and \texorpdfstring{$F_{\ast}$}{F\_{}*}}\label{sec:dra_constructions}

In this section we give explicit definitions of the adjunctions $F_{!} \dashv F^{\ast}$ and $F^{\ast} \dashv F_{\ast}$ and their dependent versions $F_{!} \dashv \modcolor{F^{\ast}}$ and $F^{\ast} \dashv \modcolor{F_{\ast}}$, given a functor $F : \CC \to \CD$.
These definitions are standard category theory, we only record them for the benefit of the reader.

The precomposition functor:
\begin{alignat*}{3}
  & F^{\ast} && :{ } && \CPsh^{\CD} \to \CPsh^{\CC} \\
  & \abs{F^{\ast}\ X'}_\Gamma && \triangleq{ } && \abs{X'}_{F\ \Gamma} \\
  & x[\rho]_{F^{\ast}\ X'} && \triangleq{ } && x[F\ \rho]_{X'} \\
  & \abs{F^{\ast}\ f'}_\Gamma\ x && \triangleq{ } && \abs{f'}_{F\ \Gamma}\ x
\end{alignat*}
Its left adjoint:
\begin{alignat*}{3}
  & F_{!} && :{ } && \CPsh^{\CC} \to \CPsh^{\CD} \\
  & \abs{F_{!}\ X}_{\Gamma'} && \triangleq{ } &&
  \big((\Gamma:\abs{\CC})\times\CD(\Gamma'\to F\ \Gamma)\times\abs{X}_{\Gamma}\big) /{\sim} \text{ where } (\Gamma,\delta',x[\rho]_X) \sim (\Delta,F\ \rho\circ\delta',x) \\
  & (\Gamma,\delta',x)[\rho']_{F_{!}\ X} && \triangleq{ } && (\Gamma,\delta'\circ\rho',x) \\
  & \abs{F_{!}\ f}_{\Gamma'}\ (\Gamma,\delta',x) && \triangleq{ } && (\Gamma,\delta',\abs{f}_\Gamma\ x)
\end{alignat*}
The unit of the adjunction $F_{!} \dashv F^{\ast}$ is given by
\begin{alignat*}{3}
  & \eta_X && :{ } && X \to (F^{\ast}\ (F_{!}\ X)) \\
  & \abs{\eta_X}_\Gamma\ x && \triangleq{ } && (\Gamma,\id_{F\ \Gamma},x)
\end{alignat*}
while the hom-set definition of the adjunction is given by an isomorphism
\[
  \phi : (F_{!}\ X\to X') \cong (X\to F^{\ast}\ X') : \phi^{-1}
\]
natural in $X$ and $X'$, where $\phi\ f' \triangleq
F^{\ast}\ f'\circ\eta_X$ \ie $\abs{\phi\ f'}_\Gamma\ x =
\abs{f'}_{F\ \Gamma}\ (\abs{\eta_X}_\Gamma\ x)$ and
$\abs{\phi^{-1}\ f}_{\Gamma'}\ (\Gamma,\delta',x) \triangleq
(\abs{f}_\Gamma\ x)[\delta']_{X}$. The dependent right adjoint of $F_{!}$:
\begin{alignat*}{3}
  & \modcolor{F^{\ast}} && :{ } && \CDepPsh^{\CD}\ (F_{!}\ X) \to \CDepPsh^{\CC}\ X \\
  & \abs{\modcolor{F^{\ast}}\ A'}_\Gamma\ x && \triangleq{ } && \abs{A'}_{F\ \Gamma}\ (\abs{\eta_X}_\Gamma\ x) \\
  & a'[\rho]_{\modcolor{F^{\ast}}\ A'} && \triangleq{ } && a'[F\ \rho]_{A'}
\end{alignat*}
We have $\modcolor{F^{\ast}}\ A' \circ f = \modcolor{F^{\ast}}\ (A'\circ F_{!}\ f)$. The dependent adjunction $F_{!} \dashv \modcolor{F^{\ast}}$ is an isomorphism
\[
\psi : \CPsh^{\CD}\big((x':F_{!}\ X) \to A'(x')\big) \cong \CPsh^{\CC}\big((x:X)\to (\modcolor{F^{\ast}}\ A')(x)\big) : \psi^{-1}
\]
natural in $X$, where $\abs{\psi\ f'}_\Gamma\ x \triangleq \abs{f'}_{F\ \Gamma}\ (\abs{\eta_X}_\Gamma\ x)$ and $\abs{\psi^{-1}\ f}_{\Gamma'}\ (\Gamma,\delta',x) \triangleq (\abs{f}_\Gamma x)[\delta']_{A'}$.

The right adjoint of $F^{\ast}$:
\begin{alignat*}{3}
  & F_{\ast} && :{ } && \CPsh^{\CC} \to \CPsh^{\CD} \\
  & \abs{F_{\ast}\ X}_{\Gamma'} && \triangleq{ } && \big\{ \alpha:(\Gamma:\abs{\CC})(\delta':\CD(F\ \Gamma\to\Gamma'))\to\abs{X}_{\Gamma} \mid \alpha\ \Gamma\ (\delta'\circ F\ \sigma) = (\alpha\ \Delta\ \delta')[\sigma]_X \big\} \\
  & \alpha[\rho']_{F_{\ast}\ X} && \triangleq{ } && \lambda \Gamma\ \delta' \mapsto \alpha\ \Gamma\ (\rho'\circ\delta') \\
  & \abs{F_{\ast}\ f}_{\Gamma'}\ \alpha && \triangleq{ } && \lambda \Gamma\ \delta' \mapsto \abs{f}_\Gamma\ (\alpha\ \Gamma\ \delta')
\end{alignat*}
The adjunction is an isomorphism $\phi : (F^{\ast}\ X'\to X) \cong
(X'\to F_{\ast}\ X) : \phi^{-1}$ natural in $X$ and $X'$ where $\abs{\phi\ f}_{\Gamma'}\ x' \triangleq \lambda \Gamma\ \delta' \mapsto \abs{f}_\Gamma\ (x'[\delta']_{X'})$ and
$\abs{\phi^{-1}\ f'}_{\Gamma}\ x' \triangleq \abs{f'}_{F\ \Gamma}\ x'\ \Gamma\ \id_{F\ \Gamma}$.
The dependent right adjoint of $F^{\ast}$:
\begin{alignat*}{5}
  & \modcolor{F_{\ast}} && :{ } && \CDepPsh^{\CC}\ (F^{\ast}\ X') \to \CDepPsh^{\CD}\ X' \\
  & \abs{\modcolor{F_{\ast}}\ A}_{\Gamma'}\ x' && \triangleq{ } && \big\{ \alpha:(\Gamma:\abs{\CC})(\delta':\CD(F\ \Gamma\to\Gamma'))\to\abs{A}_{\Gamma}\ (x'[\delta']_{X'}) \mid \\
  & && && \hphantom{\big\{{}} \alpha\ \Gamma\ (\delta'\circ F\ \sigma) = (\alpha\ \Delta\ \delta')[\sigma]_A \big\} \\
  & \alpha [\rho']_{\modcolor{F_{\ast}}\ A} && \triangleq{ } && \lambda \Gamma\ \delta' \mapsto \alpha\ \Gamma\ (\rho'\circ\delta')
\end{alignat*}
We have $\modcolor{F_{\ast}}\ A \circ f' = \modcolor{F^{\ast}}\ (A\circ F_{!}\ f')$. The dependent adjunction $F^{\ast} \dashv \modcolor{F_{\ast}}$ is an isomorphism
\[
\psi : \CPsh^{\CC}\big((x:F^{\ast}\ X') \to A(x)\big) \cong \CPsh^{\CD}\big((x':X')\to (\modcolor{F_{\ast}}\ A)(x')\big) : \psi^{-1}
\]
natural in X' where $\abs{\psi\ f}_{\Gamma'}\ x' \triangleq \lambda\Gamma\ \delta'\mapsto \abs{f}_\Gamma\ (x'[\delta']_{X'})$ and $\abs{\psi^{-1}\ f'}_\Gamma\ x' \triangleq \abs{f'}_{F\ \Gamma}\ x'\ \Gamma\ \id_{F\ \Gamma}$.


\section{Multimodal Type Theory}\label{sec:mtt}

The proofs and constructions performed in the appendix involve more than two presheaf categories and more than a single dependent right adjoint.
We rely on Multimodal Type Theory~\cite{MTT} to provide a single language that embeds the internal languages of all of those presheaf categories and the dependent right adjoints between them.

Our variant of Multimodal Type Theory differs from the one presented in~\cite{MTT} in a couple of ways.
We keep the same syntax for dependent right adjoints as in \cref{sec:dras}; whereas~\cite{MTT} uses weak dependent right adjoints instead, which come with a positive elimination rule instead of the operation $(-\ \mmod{\mu})$.
So as to remove some of the ambiguities of the informal syntax and improve readability in the presence of multiple modalities, we annotate locks with \emph{tick variables}.
The extension of the syntax of Multimodal Type Theory by ticks was used by~\cite{Transpension} for the same purpose.
Ticks were originally introduced in~\cite{ClocksAreTicking}.

\subsection{Multiple modalities}

Multiple modalities are given semantically by multiple dependent right adjoints.
Given a functor $F : \CC \to \CD$, we already have two dependent right adjoints $\modcolor{F^{\ast}}$ and $\modcolor{F_{\ast}}$, which give modalities $(\mmod{F^{\ast}} \to -)$ and $(\mmod{F_{\ast}} \to -)$.
Dependent right adjoints can be composed, and we also have modalities $(\mmod{F_{\ast}F^{\ast}} \to -)$, $(\mmod{F^{\ast}F_{\ast}} \to -)$, \etc, where $(\mmod{F_{\ast}F^{\ast}} \to -) = (\mmod{F_{\ast}} \mmod{F^{\ast}} \to -)$.

\subsubsection{Ticks}

In presence of multiple modalities, or of non-trivial relations between the modalities, the notation $(-\ \mmod{\mu})$ becomes ambiguous.
Suppose for instance that $\modcolor{\mu}$ is a idempotent dependent right adjoint ($\modcolor{\mu \mu} = \modcolor{\mu}$).
Then for any context $\Gamma$, we have $\Gamma, \mmod{\mu}, \mmod{\mu} = \Gamma, \mmod{\mu}$.
If we write $(a\ \mmod{\mu})$ in the ambient context $(\Gamma, \mmod{\mu})$, it is unclear whether the subterm $a$ should live in the context $\Gamma$ or $\Gamma, \mmod{\mu}$.

To avoid this kind of ambiguity, we will annotate locks with \emph{ticks}.
In the above example, we would have $\Gamma, \emod{\mfm}{\mu}, \emod{\mfn}{\mu} = \Gamma, \emod{\mfm\mfn}{\mu}$; and we would write either $(a\ \smod{\mfn})$ if $a$ lives over $\Gamma,\emod{\mfm}{\mu}$ or $(a\ \smod{\mfm \mfn})$ if $a$ lives over $\Gamma$.

We use $\tickcolor{\mfm}$, $\tickcolor{\mfn}$, $\tickcolor{\mfo}$, \etc for tick variables.
The tick variables refer to the locks of the ambient context.
A tick is a formal composition of tick variables, corresponding to the composition of some adjacent locks in the ambient context.
We write $\tickcolor{\bullet}$ for the empty tick, corresponding to the empty composition.
We write $\tickcolor{\overline{\mfm}}$, $\tickcolor{\overline{\mfn}}$, $\tickcolor{\overline{\mfo}}$, \etc to refer to an arbitrary tick.

If $\Gamma$ is a context, then the subterms of $(\emod{\mfm}{\mu} \to -)$ and $(\lambda\ \emod{\mfm}{\mu} \mapsto -)$ live over the context $\Gamma, \emod{\mfm}{\mu}$.

The operation $(-\ \smod{\overline{\mfm}})$ now unbinds the last tick variable of the context; or more generally some suffix of the tick variables.
The ordinary variables that occur after these tick variables are implicitly dropped from the current context.

We omit ticks when no ambiguity can arise.
In fact, we don't need to use ticks outside of this section.

\subsubsection{Morphisms between modalities}

Finally, we have morphisms between modalities.
If $\modcolor{\mu}$ and $\modcolor{\nu}$ are two parallel dependent right adjoints, whose left adjoints are respectively $L_{\mu}$ and $L_{\nu}$, a morphism $\natcolor{\alpha} : \modcolor{\mu} \Ra \modcolor{\nu}$ is a natural transformation $\alpha : L_{\mu} \Ra L_{\nu}$.
For example, given $F : \CC \to \CD$, we have a counit $\natcolor{\varepsilon^{F}} : \modcolor{F_{\ast}} \modcolor{F^{\ast}} \Ra \modcolor{1}$, and a unit $\natcolor{\eta^{F}} : \modcolor{1} \Ra \modcolor{F^{\ast}} \modcolor{F_{\ast}} $, induced by the adjunction $(F_{!} \dashv F^{\ast})$.

Given $\natcolor{\alpha} : \modcolor{\mu} \Ra \modcolor{\nu}$, we obtain a coercion operation $-\ekey{\alpha}{\overline{\mfm}}{\overline{\mfn}}$ that sends types and terms from the context $\Gamma, \emod{\overline{\mfn}}{\nu}$ to the context $\Gamma, \emod{\overline{\mfm}}{\mu}$.
Semantically, this operation is the presheaf restriction operation of types and terms along the morphism $\abs{\alpha}_{\Gamma} : (\Gamma, \emod{\overline{\mfm}}{\mu}) \to (\Gamma, \emod{\overline{\mfn}}{\nu})$.

This induces a map
\begin{alignat*}{3}
  & \mathsf{coe}_{\alpha} && :{ } && \forall (A : \emod{\mfn}{\nu} \to \SPsh) \to (\emod{\mfn}{\nu} \to A\ \smod{\mfn}) \to (\emod{\mfm}{\mu} \to (A\ \smod{\mfn})\ekey{\alpha}{\mfm}{\mfn}) \\
  & \mathsf{coe}_{\alpha}\ a && \triangleq{ } && \lambda\ \emod{\mfm}{\mu} \mapsto (a\ \smod{\mfn})\ekey{\alpha}{\mfm}{\mfn}
\end{alignat*}

For another example, consider composable functors $F : \CC \to \CD$ and $G : \CD \to \CE$.
We have a natural isomorphism $\alpha : (F G)_{!} \simeq F_{!} G_{!}$.
This induces isomorphisms $(\emod{\mfm}{{(FG)}^{\ast}} \to A) \simeq (\emod{\mff\mfg}{F^{\ast}G^{\ast}} \to A)$, whose components are
\[ \lambda a\ \emod{\mff\mfg}{F^{\ast}G^{\ast}} \mapsto (a\ \smod{\mfm})\ekey{\alpha}{\mff \mfg}{\mfm} \]
and
\[ \lambda a\ \emod{\mfm}{(FG)^{\ast}} \mapsto (a\ \smod{\mff\mfg})\ekey{\alpha\inv}{\mfm}{\mff \mfg}. \]

We omit the natural transformation when it can be inferred.
For instance we could have written $\skey{\mff\mfg}{\mfm}$ and $\skey{\mfm}{\mff\mfg}$ above.

More generally, the operation $\ekey{\alpha}{\overline{\mfm}}{\overline{\mfn}}$ can be applied to any type or term over a context of the form $(\Gamma, \emod{\overline{\mfn}}{\nu}, \Delta)$ to send it to the context $(\Gamma, \emod{\overline{\mfm}}{\mu}, \Delta\ekey{\alpha}{\overline{\mfm}}{\overline{\mfn}})$, where $\Delta$ is an extension of the context $(\Gamma, \emod{\overline{\mfn}}{\nu})$ by variable bindings and locks, and $\Delta\ekey{\alpha}{\overline{\mfm}}{\overline{\mfn}}$ applies the operation $\ekey{\alpha}{\overline{\mfm}}{\overline{\mfn}}$ to every type in $\Delta$.
In that case it is interpreted semantically by restriction along the weakening $(\Gamma, \emod{\overline{\mfm}}{\mu}, \Delta\ekey{\alpha}{\overline{\mfm}}{\overline{\mfn}}) \to (\Gamma, \emod{\overline{\mfm}}{\mu}, \Delta)$ of $\abs{\alpha}_{\Gamma} : (\Gamma, \emod{\overline{\mfm}}{\mu}) \to (\Gamma, \emod{\overline{\mfn}}{\nu})$.

The operation $\ekey{\alpha}{\overline{\mfm}}{\overline{\mfn}}$ commutes with all natural type-theoretic operations.
For example, $(A \times B)\ekey{\alpha}{\overline{\mfm}}{\overline{\mfn}} = (A\ekey{\alpha}{\overline{\mfm}}{\overline{\mfn}} \times B\ekey{\alpha}{\overline{\mfm}}{\overline{\mfn}})$.

It commutes with binders:
\[ ((a : A) \to B\ a)\ekey{\alpha}{\overline{\mfm}}{\overline{\mfn}} = (a : A\ekey{\alpha}{\overline{\mfm}}{\overline{\mfn}}) \to B\ekey{\alpha}{\overline{\mfm}}{\overline{\mfn}}\ a. \]
Note that $\ekey{\alpha}{\overline{\mfm}}{\overline{\mfn}}$ is not applied to the bound variable $a$, as it is already applied to the type $A$ of $a$.

It also commutes with $(\mmod{\mu} \to -)$ and $(\lambda\ \mmod{\mu} \mapsto -)$ for a dependent right adjoint $\mu$:
\[ (\emod{\mfo}{\mu} \to A)\ekey{\alpha}{\overline{\mfm}}{\overline{\mfn}} = (\emod{\mfo}{\mu} \to A\ekey{\alpha}{\overline{\mfm}}{\overline{\mfn}}) \]
\[ (\lambda\ \emod{\mfo}{\mu} \mapsto a)\ekey{\alpha}{\overline{\mfm}}{\overline{\mfn}} = (\lambda\ \emod{\mfo}{\mu} \mapsto a\ekey{\alpha}{\overline{\mfm}}{\overline{\mfn}}) \]

It commutes with $(-\ \smod{\mfo})$ when $\tickcolor{\mfo}$ is a tick variable that is bound in $\Delta$.

The operation $\ekey{\alpha}{\overline{\mfm}}{\overline{\mfn}\mfq}$ (where $\tickcolor{\overline{\mfn}\mfq}$ is a non-empty composition of tick variables ending in $\tickcolor{\mfq}$) can only get stuck on $(-\ \smod{\mfq})$ (or more generally on $(-\ \smod{\overline{\mfo}\mfq})$).
The operation $\ekey{\alpha}{\overline{\mfm}}{\bullet}$ (where $\tickcolor{\bullet}$ is the empty composition of ticks) can only be stuck on a variable.

Finally, these operations satisfy some $2$-naturality conditions.
Given two vertically composable morphisms $\natcolor{\alpha} : \modcolor{\mu} \Ra \modcolor{\nu}$ and $\natcolor{\beta} : \modcolor{\nu} \Ra \modcolor{\xi}$, we have
\[ (-)\ekey{\beta}{\mfn}{\mfx}\ekey{\alpha}{\mfm}{\mfn} = (-)\ekey{\alpha\beta}{\mfm}{\mfx}. \]
Given $\natcolor{\alpha} : \modcolor{\mu} \Ra \modcolor{\nu}$ and a dependent right adjoint $\modcolor{\xi}$ such that the whiskering $\natcolor{\alpha}\modcolor{\xi}$ can be formed, we have
\[ (-)\ekey{\alpha\modcolor{\xi}}{\mfm\mfx}{\mfn\mfx} = (-)\ekey{\alpha}{\mfm}{\mfn}. \]
Similarly, when we can form the whiskering $\modcolor{\xi}\natcolor{\alpha}$, we have
\[ (-)\ekey{\modcolor{\xi}\alpha}{\mfx\mfm}{\mfx\mfn} = (-)\ekey{\alpha}{\mfm}{\mfn}. \]


\section{Constructions and Proofs}\label{sec:proofs}

\subsection{Preservation of context extensions}\label{sec:preserv}

In this subsection, we show that preservation of context extensions as defined in \cref{def:preserves_ext} corresponds externally to preservation of context extension as usually defined.
Let $\int$ denote the category of elements functor.

\begin{lemma} \label{lem:loc-rep-via-left-adj}
A dependent presheaf $Y : \CDepPsh^\CC\ X$ is locally representable exactly if the projection functor $\int (X, Y) \to \int X$ is a left adjoint.
\end{lemma}

\begin{proof}
This is an easy computation.
In the notation of \cref{locally-representable}, at object $(\Gamma, x)$, the right adjoint is $(\Gamma \rhd Y_{\mid x}, x[\bm{p}^{Y}_{x}], \bm{q}^{Y}_{x})$ and the counit is $\bm{p}^{Y}_{x} : (\Gamma \rhd Y_{\mid x}, x[\bm{p}^{Y}_{x}]) \to (\Gamma, x)$.
\end{proof}

The setting of the following statement is a functor $F : \CC \to \CD$.
We had omitted universe indices in \cref{def:preserves_ext} for readability of the main body.
Elaborating, we say $F^A$ \emph{preserves $i$-small virtual context extension} if the stated condition is satisfied for $P$ an $i$-small dependent presheaf.
We say that $F^A$ \emph{preserves virtual context extension} if it preserves $i$-small virtual context extension for all $i$.

\begin{proposition} \label{preservation-of-extension}
  Let $X$ be an $i$-small presheaf over $\CC$, $A^{\CC}$ an $i$-small dependent presheaf over $X$, $A^{\CD}$ an $i$-small dependent presheaf over $F_{!}\ X$, and $F^{A}$ a dependent natural transformation from $A^{\CC}$ to $\modcolor{F^{\ast}}\ A^{\CD}$ over $X$.
  This data corresponds to the premises of \cref{def:preserves_ext}, interpreted over the context $X$ of the presheaf model $\CPsh^{\CC}$.

  The following conditions are equivalent:
  \begin{enumerate}[label=(\roman*)]
  \item \label{preservation-of-extension:virtual}
  $F^{A}$ preserves $i$-small virtual context extension,
  \item \label{preservation-of-extension:virtual-any}
  $F^{A}$ preserves $j$-small virtual context extension for a fixed $j \geq i$,
  \item \label{preservation-of-extension:virtual-all}
  $F^{A}$ preserves virtual context extension,
  \end{enumerate}
  and assuming that $A^{\CC}$ and $A^{\CD}$ are locally representable:
  \begin{enumerate}[label=(\roman*),resume]
  \item \label{preservation-of-extension:external}
  for $\Gamma : \abs{\CC}$ and $x : \abs{X}_{\Gamma}$, given a representation
  \[
  \bm{q}^{A^\CC}_x : \abs{A^{\CC}_{\mid x}}\ (\Gamma \rhd A^\CC_{\mid x}, \bm{p}^{A^\CC}_x)
  ,\]
  then
  \[
  \abs{F^{A}}_{\Gamma \rhd A^\CC_{\mid x}}\ x[\rho]\ \bm{q}^{A^\CC}_x : \abs{A^{\CD}_{\mid \left(\abs{\eta^F_X}_\Gamma x\right)}}\ (F\ (\Gamma \rhd A^\CC_{\mid x}), F\ \bm{p}^{A^\CC}_x)
  \]
  is a representation.
  \item \label{preservation-of-extension:ugly}
  for $\Gamma$ and $x$ as above, the comparison morphism
  \[ \angles{F\ \bm{p}^{A^{\CC}}_{x}, \abs{F^{A}}_{\Gamma \rhd A^{\CC}_{\mid x}}\ \left(x[\bm{p}^{A^{\CC}}_{x}]\right)\ \bm{q}^{A^{\CC}}_{x}} : F\ (\Gamma \rhd A^{\CC}_{\mid x}) \to (F\ \Gamma) \rhd A^{\CD}_{\mid\left(\abs{\eta^{F}_{X}}_{\Gamma}\ x\right)} \]
  is invertible.
  \end{enumerate}
  Here, $\eta^{F}_{X} : X \to F^{\ast}\ (F_{!}\ X)$ is the component at $X$ of the unit of the adjunction $(F_{!} \dashv F^{\ast})$.
\end{proposition}

\begin{proof}
We have a strictly commuting square
\begin{equation} \label{preservation-of-extension:categories}
\begin{tikzcd}
  \int (X, A^\CC)
  \ar[r, "v"]
  \ar[d, "p_\CC"]
&
  \int (F_! X, A^\CD)
  \ar[d, "p_\CD"]
\\
  \int X
  \ar[r, "u"]
&
  \int F_! X
\end{tikzcd}
\end{equation}
of categories where $p_\CC$ and $p_\CD$ are projections, $u$ is given by the unit of the adjunction $F_! \dashv F^*$, and $v$ is induced by $F^A$.

Restriction along the maps in~\eqref{preservation-of-extension:categories} induces a strictly commuting square
\begin{equation} \label{preservation-of-extension:restriction}
\begin{tikzcd}
  \CDepPsh^{\CC}\ (X, A^\CC)
&
  \CDepPsh^{\CD}\ (F_! X, A^\CD)
  \ar[l, "v^*"]
\\
  \CDepPsh^{\CC}\ X
  \ar[u, "(p_\CC)^*"]
&
  \CDepPsh^{\CD}\ \int F_! X
  \ar[u, "(p_\CD)^*"]
  \ar[l, "u^*"]
\rlap{.}
\end{tikzcd}
\end{equation}
We have adjunctions $(p_\CC)^* \dashv (p_\CC)_*$ and $(p_\CD)^* \dashv (p_\CD)_*$ where $(p_\CC)_*$ takes $\Pi$-types over $A^\CC$ and $(p_\CD)_*$ takes $\Pi$-types along $A^\CD$.
Regarding the identity as a natural transformation in~\eqref{preservation-of-extension:restriction} from the right-top to the bottom-left composite and taking its mate with respect to these adjunctions, we obtain a natural transformation
\begin{equation} \label{preservation-of-extension:mate-restriction}
\begin{tikzcd}
  \CDepPsh^{\CC}\ (X, A^\CC)
  \ar[d, "(p_\CC)_*"]
&
  \CDepPsh^{\CD}\ (F_! X, A^\CD)
  \ar[l, "v^*"]
  \ar[d, "(p_\CD)_*"]
\\
  \CDepPsh^{\CC}\ X
&
  \CDepPsh^{\CD}\ F_! X
  \ar[l, "u^*"]
  \ar[ul,shorten <>=20pt,Rightarrow,"\tau"]
\rlap{.}
\end{tikzcd}
\end{equation}
This is the comparison map $\tau$ from \cref{def:preserves_ext}.
Conditions~\ref{preservation-of-extension:virtual}, \ref{preservation-of-extension:virtual-any}, \ref{preservation-of-extension:virtual-all} state that $\tau$ is invertible at dependent presheaves over $(F_! X, A^\CD)$ that are $i$-small, $j$-small for a fixed $j \geq i$, and $j$-small for any $j$, respectively.

The above categories may be seen as presheaf categories.
Switching to that presentation and transposing the preceding square to left adjoints, we obtain the natural transformation
\begin{equation} \label{preservation-of-extension:mate-left-kan-extension}
\begin{tikzcd}
  \CPsh^{\int (X, A^\CC)}
  \ar[r, "v_!"]
  \ar[dr,shorten <>=20pt,Rightarrow,"\overline{\tau}"]
&
  \CPsh^{\int (F_! X, A^\CD)}
\\
  \CPsh^{\int X}
  \ar[r, "u_!"]
  \ar[u, "(p_\CC)^*"]
&
  \CPsh^{\int F_! X}
  \ar[u, "(p_\CD)^*"]
\rlap{.}
\end{tikzcd}
\end{equation}
We have that $\tau$ is invertible exactly if $\overline{\tau}$ is invertible.
By Yoneda, $\tau$ is invertible already if $\overline{\tau}$ is invertible at representables.
Reversely, for $\overline{\tau}$ to be invertible at representables, it suffices for $\tau$ to be invertible at applications of $v_! (p_\CC)^*$ to representables.
By our smallness assumptions on $X$, $A^\CC$, $A^\CD$, these applications are $i$-small.
This shows that conditions~\ref{preservation-of-extension:virtual}, \ref{preservation-of-extension:virtual-any}, \ref{preservation-of-extension:virtual-all} are equivalent and hold exactly if $\overline{\tau}$ is invertible at representables (which holds exactly if $\tau$ and $\overline{\tau}$ are invertible).

Assume now that $A^\CC$ and $A^\CD$ are locally representable.
By \cref{lem:loc-rep-via-left-adj}, the functors $p_\CC$ and $p_\CD$ have right adjoints $q_\CC$ and $q_\CD$, respectively.
Regarding the identity as a natural transformation from the left-bottom to the top-right composite in~\eqref{preservation-of-extension:categories}, we take its mate with respect to these adjunctions, obtaining the natural transformation
\begin{equation*} 
\begin{tikzcd}
  \int (X, A^\CC)
  \ar[r, "v"]
  \ar[dr,shorten <>=10pt,Rightarrow,"\theta"]
&
  \int (F_! X, A^\CD)
\\
  \int X
  \ar[r, "u"]
  \ar[u, "q_\CC"]
&
  \int F_! X
  \ar[u, "q_\CD"]
\rlap{.}
\end{tikzcd}
\end{equation*}
Condition~\ref{preservation-of-extension:ugly} states that $\theta$ is invertible when using the explicit description of $q_\CC$ and $q_\CD$ given in the proof of \cref{lem:loc-rep-via-left-adj}.
Condition~\ref{preservation-of-extension:external} states the same thing without referring to particular choices of context extension: given $M : \abs{\int X}$, if $(N, f)$ is terminal in $(p_\CC \downarrow \int X)$, then $(v\ N, u\ f)$ is terminal in $(p_\CD \downarrow \int F_! X)$.
To see that these conditions are equivalent, recall that terminal objects are unique up to (unique) isomorphisms.

The process of taking mates commutes with presheaf category formation.
Thus, the natural transformation $\tau$ in~\eqref{preservation-of-extension:mate-restriction} is equivalent to the restriction action of $\theta$.
It follows that $\overline{\tau}$ in~\eqref{preservation-of-extension:mate-left-kan-extension} is equivalent to the left Kan extension action of $\theta$.
Recall that left Kan extension of representables is given by the original functor.
Thus, $\overline{\tau}$ is invertible at representables exactly if $\theta$ is invertible.
This shows that~\ref{preservation-of-extension:virtual} and~\ref{preservation-of-extension:external} are equivalent.
\end{proof}

\subsection{Displayed categories}

Displayed categories were introduced in~\cite{DispCats}.
The data of a displayed category over a base $\CD$ is equivalent to the data of a category $\CC$ equipped with a functor into $\CD$.
Many structures on functors that may seem non-categorical because they involve equalities of objects, such as fibration structures, are actually well-behaved when seen as structures over displayed categories.
Some of the constructions that follow are more intuitive when thinking about displayed categories instead of functors.
Because Multimodal Type Theory does not have dependent modes, we have to see displayed categories as functors when working internally.

We write $F : \CC \ratri \CD$ if $F$ is a functor that exhibits $\CC$ as a displayed category over $\CD$.
Given an object $x$ of $\CD$, we may write $\abs{\CC}(x)$ (or $\CC(x)$ or $\CC^{\op}(x)$) for the set of objects of $\CC$ displayed over $x$, that is the set containing the objects $x' : \abs{\CC}$ such that $F\ x' = x$.
Given objects $x$ and $y$ of $\CC$ and a morphism $f : \CD(F\ x \to F\ y)$, we write $\CC(x \to_{f} y)$ for the set of morphisms of $\CC$ from $x$ to $y$ that are displayed over $f$.
In other words, $\CC(x \to_{f} y)$ is the set containing the morphisms $f' : \CC(x \to y)$ such that $F\ f' = f$.

\subsection{Sections of displayed models with context extensions}\label{sec:sections}

\begin{definition}
  A \defemph{section} of a displayed model with context extensions $\CS^{\bullet}$ over a functor $F : \CC \ratri \CS$ consists of a section $S : \CS \to \CC$ of $F$ (up to a natural isomorphism $(S \cdot F) \simeq 1_{\CS}$) along with (internally to $\CPsh^{\CS}$):
  \begin{itemize}
    \item Actions on types and terms.
          \begin{alignat*}{3}
            & S^{\Ty} && :{ } && \forall (A : \Ty^{\CS}) \mmod{S^{\ast}} \to \Ty^{\bullet}\ (\lambda \mmod{F^{\ast}} \to A\smkey{S^{\ast}F^{\ast}}{\bullet}) \\
            & S^{\Tm} && :{ } && \forall A\ (a : \Tm^{\CS}\ A)\ \mmod{S^{\ast}} \to \Tm^{\bullet}\ (S^{\Ty}\ A\ \mmod{S^{\ast}})\ (\lambda \mmod{F^{\ast}} \to a\smkey{S^{\ast}F^{\ast}}{\bullet})
          \end{alignat*}
          where $\smkey{S^{\ast}F^{\ast}}{\bullet}$ is coercion over the natural isomophism $(S \cdot F) \simeq 1_{\CS}$.
    \item Such that for every $A : \Ty^{\CS}$, the total action
          \[ (a : \Tm^{\CS}\ A) \to (\mmod{S^{\ast}} \to (a : \mmod{F^{\ast}} \to \Tm^{\CS}\ A\smkey{S^{\ast}F^{\ast}}{\bullet}) \times (\Tm^{\bullet}\ (S^{\Ty}\ A\ \mmod{S^{\ast}})\ a)) \]
          preserves context extensions.

          As in the definition of morphisms of models, we can then derive actions on derived sorts.
          Given any telescope $X$ of $\CS$, we can define, by induction on $X$, a family
          \[ X^{\bullet} : \mmod{S^{\ast}} \to (\mmod{F^{\ast}} \to X^{\CS}\smkey{S^{\ast}F^{\ast}}{\bullet}) \to \SPsh^{\CC}, \]
          along with an action
          \[ S^{X} : (x : X^{\CS})\ \mmod{S^{\ast}} \to X^{\bullet}\ (\lambda \mmod{F^{\ast}} \to x\smkey{\mfs\mff}{\bullet}) \]
          such that the induced total action
          \[ (x : X^{\CS}) \to (\mmod{S^{\ast}} \to (x : \mmod{F^{\ast}} \to X^{\CS}\smkey{\mfs\mff}{\bullet}) \times (x^{\bullet} : X^{\bullet}\ \mmod{S^{\ast}}\ x)) \]
          has a locally representable codomain and preserves context extensions.
          We can then define, using the pattern matching notation of \cref{def:preserves_ext},
          \begin{alignat*}{3}
            & S^{[X]\Ty} && :{ } && (A : X^{\CS} \to \Ty^{\CS})\ \mmod{S^{\ast}}\ x\ (x^{\bullet} : X^{\bullet}\ \mmod{S^{\ast}}\ x) \\
            &&&&& \to \Ty^{\bullet}\ (\lambda \mmod{F^{\ast}} \to A\smkey{S^{\ast}F^{\ast}}{\bullet}\ (x\ \mmod{F^{\ast}})) \\
            & S^{[X]\Ty} && \triangleq{ } && \lambda A\ \mmod{S^{\ast}}\ (S^{X}\ x\ \mmod{S^{\ast}}) \mapsto S^{\Ty}\ (A\ x)\ \mmod{S^{\ast}} \\
            & S^{[X]\Tm} && :{ } && \forall A\ (a : \Tm^{\CS}\ a)\ \mmod{S^{\ast}}\ x\ (x^{\bullet} : X^{\bullet}\ \mmod{S^{\ast}}\ x) \\
            &&&&& \to \Tm^{\bullet}\ (S^{\Ty}\ A\ \mmod{S^{\ast}})\ (\lambda \mmod{F^{\ast}} \to a\smkey{S^{\ast}F^{\ast}}{\bullet}) \\
            & S^{[X]\Tm} && \triangleq{ } && \lambda a\ \mmod{S^{\ast}}\ (S^{X}\ x\ \mmod{S^{\ast}}) \mapsto S^{\Tm}\ (a\ x)\ \mmod{S^{\ast}}
          \end{alignat*}
    \item And such that all type-theoretic operations are preserved. For example,
          \begin{alignat*}{3}
            & S^{\Ty}\ \BoolTy^{\CS} && ={ } && \lambda \mmod{S^{\ast}} \mapsto \BoolTy^{\bullet} \\
            & S^{\Ty}\ (\elimBool^{\CS}\ P\ t\ f\ n) && ={ } && \elimBool^{\bullet} \circleddollar S^{[\Tm]\Ty}\ P \circledast S^{\Tm}\ t \circledast S^{\Tm}\ f \circledast S^{\Tm}\ n \\
            & S^{\Ty}\ (\Pi^{\CS}\ A\ B) && ={ } && \Pi^{\bullet} \circleddollar S^{\Ty}\ A \circledast S^{[\Tm]\Ty}\ B
          \end{alignat*}

          As in the definition of morphisms of models, we can derive computation rules for $S^{[X]\Ty}$ and $S^{[X]\Tm}$.
  \end{itemize}
\end{definition}

\subsection{Displayed presheaf category}

In what follows, we need to consider categories of presheaves over large categories, and in particular categories of presheaves over categories of presheaves.
We have to be a bit careful about sizes.
If $\CC$ is a category, we write $\widehat{\CC}$ for the category of $\omega$-small presheaves (functors into $\CSet_{\omega}$) over $\CC$, and $\CPsh^{\CC}$ for the category of large presheaves (functors into $\CSet_{\omega+1}$) over $\CC$.
We only use the internal language of $\CPsh^{\CC}$.

The goal of this subsection is to construct the factorization of \cref{con:disp_replace_0}.
\restateDefFactorization*
\restateDispReplace*

We fix a model $\CS$ of $\Th_{\Pi,\BoolTy}$ and a functor $F : \CC \ratri \CS$ for this whose subsection.

\begin{definition}
  We define the \defemph{displayed presheaf category} $\CP$ along with a projection functor $G : \CP \ratri \CS$ (which exhibits $\CP$ as a displayed category over $\CS$) and the \defemph{displayed Yoneda embedding} $Y : \CC \to \CP$.
  The displayed Yoneda embedding is a displayed functor over $\CS$: it satisfies $(Y \cdot G) = F$.
  They are analogous to the usual category of presheaves and Yoneda embedding, but in the slice $2$-category $(\CCat / \CS)$, or equivalently in the $2$-category of displayed categories over $\CS$.

  \begin{mathpar}
    \begin{tikzcd}
      \CC \ar[rd, "F"', -{Triangle[open]}] \ar[rr, "Y"] && \CP \ar[ld, "G", -{Triangle[open]}] \\
      & \CS &
    \end{tikzcd}
  \end{mathpar}

  An object $\Gamma^{\dagger}$ of $\CP$ displayed over an object $\Gamma$ of $\CS$ is a dependent presheaf
  \[ \Gamma^{\dagger} : \forall (\Theta : \CC^{\op}) (\gamma : \CS(F\ \Theta \to \Gamma)) \to \CSet_{\omega}. \]

  A morphism $f^{\dagger} : \CP(\Gamma^{\dagger} \to_{f} \Delta^{\dagger})$ displayed over a morphism $f : \CS(\Gamma \to \Delta)$ is a dependent natural transformation
  \[ f^{\dagger} : \forall (\Theta : \CC^{\op}) (\gamma : \CS(F\ \Theta \to \Gamma)) \to \Gamma^{\dagger}\ \gamma \to \Delta^{\dagger}\ (\gamma \cdot f). \]

  Given an object $\Gamma : \abs{\CC}$, we define an object $\abs{Y}_{\Gamma}$ of $\CP$ displayed over $F\ \Gamma$ by
  \[ \abs{Y}_{\Gamma}\ \Theta\ \gamma \triangleq \CC(\Theta \to_{\gamma} \Gamma) \]
  As this is both contravariant in $\Gamma$ and covariant in $\Theta$, this extends to a displayed functor $Y : \CC \to \CP$.
\end{definition}

\begin{proposition}\label{prop:disp_as_gluing}
  The category $\CP$ is equivalent to the comma category $(\widehat{\CC} \downarrow N_{F})$, where $N_{F} : \CS \to \widehat{\CC}$ is the composition of the Yoneda embedding $\yo^{\CS} : \CS \to \widehat{\CS}$ with $F^{\ast} : \widehat{\CS} \to \widehat{\CC}$.
  \qed
\end{proposition}

We prove several core properties of $G : \CC \ratri \CP$ and $Y : \CC \to \CP$.

The first of these properties is the generalization of the Yoneda lemma to $\CP$.
\begin{lemma}\label{lem:disp_yoneda_lemma}
  There is a natural family of isomorphisms
  \[ r : \forall (\Gamma^{\dagger} : \CP) (\Delta : \CC^{\op}) (\gamma : \CS(F\ \Delta \to \Gamma)) \to \Gamma^{\dagger}\ \gamma \simeq \CP(\abs{Y}_{\Delta} \to_{\gamma} \Gamma^{\dagger}), \]
  whose components are given by
  \begin{alignat*}{3}
    & r\ \Gamma^{\dagger}\ \Delta\ \gamma\ \gamma^{\dagger} && \triangleq{ } && \lambda \Theta\ \delta\ (\delta' : \CC(\Theta \to_{\delta} \Delta)) \mapsto \gamma^{\dagger}[\delta'] \\
    & r^{-1}\ \Gamma^{\dagger}\ \Delta\ \gamma\ \gamma' && \triangleq{ } && \gamma'\ \Delta\ \id_{F\ \Delta}\ \id_{\Delta} \tag*{\qed}
  \end{alignat*}
\end{lemma}

\begin{proposition}\label{prop:disp_yoneda_ff}
  The functor $Y : \CC \to \CP$ is fully faithful.
\end{proposition}
\begin{proof}
  We prove that the actions of $Y$ on displayed morphisms are bijective; this implies that its total actions are also bijective.

  Take two objects $\Gamma$ and $\Delta$ of $\CC$ and a base morphism $f : \CS(F\ \Delta \to F\ \Gamma)$.
  By \cref{lem:disp_yoneda_lemma}, we have $\abs{Y}_{\Gamma}\ f \simeq \CP(\abs{Y}_{\Delta} \to_{f} \abs{Y}_{\Gamma})$; and $\abs{Y}_{\Gamma}\ f = \CC(\Delta \to_{f} \Gamma)$ by definition.
  This determines a function $\CC(\Delta \to_{f} \Gamma) \to \CP(\abs{Y}_{\Delta} \to_{f} \abs{Y}_{\Gamma})$ that coincides with the action of $Y$ on displayed morphisms, which is therefore bijective.
\end{proof}

\begin{proposition}\label{prop:disp_yoneda_unit}
  The unit $\eta^{Y} : 1_{\CPsh^{\CC}} \Ra (Y_{!} \cdot Y^{\ast})$ is an isomorphism.
\end{proposition}
\begin{proof}
  This follows from $Y$ being fully faithful (see~\cite[Prop 4.23]{Kelly05}).
\end{proof}
As both $Y_{!}$ and $Y^{\ast}$ admit dependent right adjoints, \cref{prop:disp_yoneda_unit} induces coercion operations internally to $\CPsh^{\CC}$ and $\CPsh^{\CP}$.

\begin{lemma} \label{lem:fib-pullback-left-adj}
Let $p : \CE \to \CC$ be a Grothendieck fibration.
Then pullback along $p$ preserves left adjoints.
\end{lemma}

\begin{proof}
A standard fact.
\end{proof}

In the following, we switch freely between the point of view of a dependent presheaf over $X$ and a map into $X$.
In particular, a dependent presheaf over $X$ is locally representable exactly if the corresponding map into $X$ is a representable morphism.

\begin{corollary} \label{cor:loc-rep-pullback-fib}
Let $p : \CE \to \CC$ be a Grothendieck fibration.
Let $f \in \CPsh^{\CC}(Y \to X)$ be locally representable.
Then $p^{\ast} f \in \CPsh^{\CE}(p^{\ast} Y \to p^{\ast} X)$ is locally representable.
\end{corollary}

\begin{proof}
  This is the combination of \cref{lem:loc-rep-via-left-adj,cor:loc-rep-pullback-fib}.
  For this, note that $\int p^{\ast} f$ is the pullback of $\int f$ along $p$.
\end{proof}

\begin{proposition}[Internally to $\CPsh^{\CP}$]\label{prop:disp_yoneda_rep_1}
  For every locally representable presheaf $A : \mmod{G^{\ast}} \to \SRepPsh^{\CS}$, the presheaf $(\mmod{G^{\ast}} \to A\ \mmod{G^{\ast}})$ is locally representable.
\end{proposition}

\begin{proof}
We have to show the judgment
$
A : \mmod{G^{\ast}} \to \SRepPsh^{\CS} \vdash \isRep(\mmod{G^{\ast}} \to A\ \mmod{G^{\ast}})
.$
Inhabitants of this type correspond to local representability structures of the dependent presheaf
$
A : \mmod{G^{\ast}} \to \SRepPsh^{\CS} \vdash \mmod{G^{\ast}} \to A\ \mmod{G^{\ast}}
.$
This is the image of the universal locally representable dependent presheaf
$
A : \SRepPsh^{\CS} \vdash A
$
under $G^{\ast}$.
From \cref{prop:disp_as_gluing}, we see that $G^{\ast}$ is a Grothendieck fibration.
We conclude by \cref{cor:loc-rep-pullback-fib}.
\end{proof}

\begin{proposition}[Internally to $\CPsh^{\CP}$]\label{prop:disp_yoneda_rep_2}
  Given an $\omega$-small presheaf $A : \mmod{Y_{\ast}} \to \SPsh_{\omega}^{\CC}$, the presheaf $(\mmod{Y_{\ast}} \to A\ \mmod{Y_{\ast}})$ is locally representable.
\end{proposition}

\begin{proof}
We have to show the judgment
$
A : \mmod{Y_{\ast}} \to \SPsh_{\omega}^{\CC} \vdash \isRep(\mmod{Y_{\ast}} \to A\ \mmod{Y_{\ast}})
.$
Inhabitants of this type correspond to local representability structures of the dependent presheaf
$
A : \mmod{Y_{\ast}} \to \SPsh_{\omega}^{\CS} \vdash \mmod{Y_{\ast}} \to A\ \mmod{Y_{\ast}}
.$
This is the image of the universal dependent presheaf $A : \SPsh^{\CS} \vdash A$ under $Y_*$.
So it suffices to show the following: given an $\omega$-small dependent presheaf $N$ over $M$ in $\SPsh^{\CC}$, the dependent presheaf $Y_* N$ over $Y_* M$ in $\SPsh^{\CP}$ is locally representable.

Let us inspect the action of the functor $Y_*$ on $M : \SPsh^{\CC}$.
From \cref{sec:dra_constructions}, we have for $\Gamma : \CC$ and $\Gamma^\dagger : \CP(\Gamma)$ that $|Y_* M|_{(\Gamma, \Gamma^\dagger)}$ consists of a dependent natural transformation
\[
(\Delta : \CC^{\op}) (g : \CP(Y \Delta, (\Gamma, \Gamma^\dagger))) \to |M|_\Delta.
\]
Regarding $\CC$ as displayed over $\CS$, this writes as
\[
(\Delta_\CS : \CS^{\op}) (\Delta_\CC : \CC^{\op}(\Delta_\CC)) (u : \CS(\Delta_\CS \to \Gamma)) (u^\dagger : \CP(Y \Delta_\CC \to_u \Gamma^\dagger)) \to |M|_{\Delta_\CC}.
\]
By \cref{lem:disp_yoneda_lemma} (displayed Yoneda), this is naturally isomorphic to the type of dependent natural transformations
\[
(\Delta_\CS : \CS^{\op}) (\Delta_\CC : \CC^{\op}(\Delta_\CC)) (u : \CS(\Delta_\CS \to \Gamma)) (u^\dagger : \Gamma^\dagger\ u)) \to |M|_{\Delta_\CC}.
\]

Similarly, for a dependent presheaf $N$ over $M$ in $\SPsh^{\CC}$, we can describe the dependent presheaf $Y_* N$ over $Y_* M$ in $\SPsh^{\CP}$ as follows.
Given $\Gamma : \CC$ and $\Gamma^\dagger : \CP(\Gamma)$ and $\alpha : |Y_* M|_{(\Gamma, \Gamma^\dagger)}$, then $|Y_* N|_{(\Gamma,\Gamma^{\dagger})}\ \alpha$ is the type of dependent natural transformations
\[
  (\Delta_\CS : \CS^{\op}) (\Delta_\CC : \CC^{\op}(\Delta_\CC)) (u : \CS(\Delta_\CS \to \Gamma)) (u^\dagger : \Gamma^\dagger\ u) \to |N|\ (\alpha\ \Delta_\CS\ \Delta_\CC\ u\ u^\dagger).
\]

We now show that $Y_{\ast} N$ is locally representable.
Let $(\Gamma,\Gamma^{\dagger})$ be an object of $\CP$ and $\alpha : \abs{Y_{\ast} M}_{(\Gamma,\Gamma^{\dagger})}$.
We have to define a representing object for the presheaf (over $(\CP/(\Gamma,\Gamma^{\dagger}))$)
\begin{alignat*}{3}
  & (Y_{\ast} N)_{\mid \alpha} && :{ } && ((\Omega,\Omega^{\dagger}) : \CP^{\op}) (\rho : \CS(\Omega \to \Gamma)) (\rho^{\dagger} : \CP(\Omega^{\dagger} \to_{\gamma} \Gamma^{\dagger})) \to \SSet \\
  & \abs{(Y_{\ast}N)_{\mid \alpha}}\ (\Omega,\Omega^{\dagger})\ \rho\ \rho^{\dagger} && \triangleq{ } && \abs{Y_{\ast} N}_{(\Omega,\Omega^{\dagger})}\ (\lambda \Delta_{\CS}\ \Delta_{\CC}\ u\ u^{\dagger} \mapsto \alpha\ \Delta_{\CS}\ \Delta_{\CC}\ (u \cdot \rho)\ (\rho^{\dagger}\ u\ u^{\dagger}))
\end{alignat*}

The extended context is $(\Gamma, \Gamma^{\rhd})$ where
\begin{alignat*}{3}
  & \Gamma^{\rhd}\ \Delta_{\CS}\ \Delta_{\CC}\ (u : \CS(\Delta_{\CS} \to \Gamma)) && \triangleq{ } &&
  (u^{\dagger} : \Gamma^{\dagger}\ \Delta_{\CS}\ \Delta_{\CC}\ u) \times (n : \abs{N}\ \Delta_{\CS}\ \Delta_{\CC}\ (\alpha\ \Delta_{\CS}\ \Delta_{\CC}\ u\ u^{\dagger}))
\end{alignat*}
Note that this is only well-defined because $N$ is $\omega$-small.

The projection morphism $\bm{p} : \CP(\Gamma^{\rhd} \to_{\id} \Gamma^{\dagger})$ forgets the component $n$.
The generic element $\bm{q} : \abs{(Y_{\ast} N)_{\mid \alpha}}\ (\Gamma,\Gamma^{\rhd})\ \id\ \bm{p}$ is given by $\bm{q}\ \Delta_{\CS}\ \Delta_{\CC}\ u\ (u^{\dagger}, n) \triangleq n$.

Finally, we have to check the universal property of the extended context.
Given $(\Delta,\Delta^{\dagger}) : \CP$, a morphism from $(\Omega, \Omega^{\dagger})$ to $(\Gamma,\Gamma^{\rhd})$ consists of a morphism $\rho : \CS(\Omega \to \Gamma)$ and a dependent natural transformation
\[ \rho^{\rhd} : \forall \Delta_{\CS}\ \Delta_{\CC}\ (u : \CS(\Delta_{\CS} \to \Omega)) \to \Omega^{\dagger}\ u \to \Gamma^{\rhd}\ (u \cdot \rho). \]
By definition of $\Gamma^{\rhd}$, this is equivalently given by a pair of dependent natural transformations
\begin{alignat*}{3}
 & \rho^{\dagger} && :{ } && \forall \Delta_{\CS}\ \Delta_{\CC}\ (u : \CS(\Delta_{\CS} \to \Omega)) \to \Omega^{\dagger}\ u \to \Gamma^{\dagger}\ (u \cdot \rho) \\
 & \rho^{n} && :{ } && \forall \Delta_{\CS}\ \Delta_{\CC}\ (u : \CS(\Delta_{\CS} \to \Omega))\ (u^{\dagger} : \Omega^{\dagger}\ u) \to \abs{N}\ \Delta_{\CS}\ \Delta_{\CC}\ (\alpha\ \Delta_{\CS}\ \Delta_{\CC}\ (u \cdot \rho)\ (\rho^{\dagger}\ u\ u^{\dagger})),
\end{alignat*}
\ie. by $\rho^{\dagger} : \CP(\Omega^{\dagger} \to_{\rho} \Gamma^{\dagger})$ and $\rho^{n} : \abs{(Y_{\ast} N)_{\mid \alpha}}\ (\Omega,\Omega^{\dagger})\ \rho\ \rho^{\dagger}$.
This shows that $(\Gamma, \Gamma^{\rhd})$ satisfies the universal property of a representing object of $(Y_{\ast} N)_{\mid \alpha}$.
\end{proof}

\begin{proposition}[Internally to $\CPsh^{\CP}$]\label{prop:disp_yoneda_rep_3}
  For every $\omega$-small presheaf $A : \mmod{Y_{\ast}} \to \SPsh_{\omega}^{\CC}$, the unique map
  \[ (\mmod{Y_{\ast}} \to A\ \mmod{Y_{\ast}}) \to (\mmod{G^{\ast}} \to \Unit) \]
  preserves context extensions.
  Equivalently, the constant map
  \[ (\mmod{G^{\ast}} \to B\ \mmod{G^{\ast}}) \to ((\mmod{Y_{\ast}} \to A\ \mmod{Y_{\ast}}) \to (\mmod{G^{\ast}} \to B\ \mmod{G^{\ast}})) \]
  is an isomorphism for every $B : \mmod{G^{\ast}} \to \SPsh^{\CS}$.
\end{proposition}
\begin{proof}
  This follows from the fact that the projection map $\bm{p} : \CP((\Gamma,\Gamma^{\rhd}) \to (\Gamma,\Gamma^{\dagger}))$ constructed in the proof of \cref{prop:disp_yoneda_rep_2} is sent by $G$ to the identity morphism $\id : \CS(\Gamma \to \Gamma)$.
\end{proof}

We can now forget the definitions of $\CP$, $G$ and $Y$, as we will only rely on these properties.

We will work internally to $\CPsh^{\CP}$, $\CPsh^{\CC}$ and $\CPsh^{\CS}$ and use the dependent right adjoints $\modcolor{F^{\ast}}$, $\modcolor{G^{\ast}}$, $\modcolor{Y^{\ast}}$, $\modcolor{Y_{\ast}}$ and their compositions.
There is actually, up to isomorphism, only a single new composite dependent right adjoint $\modcolor{Y_{\ast}Y^{\ast}}$.







\begin{construction}\label{con:disp_replace_model}
  We construct a displayed model with context extensions $\CS^{\dagger}$ over $G : \CP \ratri \CS$.
  Furthermore, we equip $Y : \CC \to \CP$ with actions on displayed types and terms that preserve all displayed type-theoretic operations.
\end{construction}
\begin{proof}[Construction]
  We pose, internally to $\CPsh^{\CP}$,
  \begin{alignat*}{3}
    & \Ty^{\dagger} && :{ } && (A : \mmod{G^{\ast}} \to \Ty^{\CS}) \to \SPsh^{\CP} \\
    & \Ty^{\dagger}\ A && \triangleq{ } &&
    \mmod{Y_{\ast}} \to \Ty^{\bullet}\ (\lambda \mmod{F^{\ast}} \mapsto (A\ \mmod{G^{\ast}})\smkey{Y_{\ast}F^{\ast}}{G^{\ast}}) \\
    & \Tm^{\dagger} && :{ } && \forall A\ (A^{\dagger} : \Ty^{\dagger}\ A) (a : \mmod{G^{\ast}} \to \Tm^{\CS}\ (A\ \mmod{G^{\ast}})) \to \SPsh^{\CP} \\
    & \Tm^{\dagger}\ A^{\dagger}\ a && \triangleq{ } &&
    \mmod{Y_{\ast}} \to \Tm^{\bullet}\ (A^{\dagger}\ \mmod{Y_{\ast}})\ (\lambda \mmod{F^{\ast}} \mapsto (a\ \mmod{G^{\ast}})\smkey{Y_{\ast}F^{\ast}}{G^{\ast}})
  \end{alignat*}
  By \cref{prop:disp_yoneda_rep_1}, the family $(\mmod{G^{\ast}} \to \Tm^{\CS}\ (A\ \mmod{G^{\ast}}))$ is locally representable.
  By \cref{prop:disp_yoneda_rep_2}, the family $\Tm^{\dagger}\ A^{\dagger}\ a$ is also locally representable.
  Thus the dependent sum
  \[ (a : \mmod{G^{\ast}} \to \Tm^{\CS}\ (A\ \mmod{G^{\ast}})) \times (a^{\dagger} : \Tm^{\dagger}\ A^{\dagger}\ a) \]
  is a locally representable family of presheaves.
  The fact that the first projection map
  \[ (a : \mmod{G^{\ast}} \to \Tm^{\CS}\ (A\ \mmod{G^{\ast}})) \times (a^{\dagger} : \Tm^{\dagger}\ A^{\dagger}\ a) \xrightarrow{\lambda (a,a^{\dagger}) \mapsto a} (\mmod{G^{\ast}} \to \Tm^{\CS}\ (A\ \mmod{G^{\ast}})) \]
  preserves context extensions follows from \cref{prop:disp_yoneda_rep_3}.

  We have, internally to $\CPsh^{\CC}$, the following isomorphisms
  \begin{alignat*}{3}
    & Y^{\Ty} && :{ } &&
    (A : \mmod{F^{\ast}} \to \Ty^{\CS}) \\
    &&&&& \to \Ty^{\bullet}\ A \simeq (\mmod{Y^{\ast}} \to \Ty^{\dagger}\ (\lambda \mmod{G^{\ast}} \mapsto (A\ \mmod{F^{\ast}})\smkey{Y^{\ast}G^{\ast}}{F^{\ast}})) \\
    & Y^{\Ty} && \triangleq{ } &&
    \lambda A\ A^{\bullet}\ \mmod{Y^{\ast}Y_{\ast}} \mapsto A^{\bullet}\smkey{Y^{\ast}Y_{\ast}}{\bullet} \\
    & Y^{\Ty,-1} && \triangleq{ } &&
    \lambda A\ A^{\dagger}\ \mapsto (A^{\dagger}\ \mmod{Y^{\ast}Y_{\ast}})\smkey{\bullet}{Y^{\ast}Y_{\ast}} \\
    & Y^{\Tm} && :{ } &&
    \forall (A : \mmod{F^{\ast}} \to \Ty^{\CS})\ (A^{\bullet} : \Ty^{\bullet}\ A)\ (a : \mmod{F^{\ast}} \to \Tm^{\CS}\ (A\ \mmod{F^{\ast}})) \\
    &&&&& \to \Tm^{\bullet}\ A^{\bullet}\ a \simeq (\mmod{Y^{\ast}} \to \Tm^{\dagger}\ (Y^{\Ty}\ A^{\bullet}\ \mmod{Y^{\ast}})\ (\lambda \mmod{G^{\ast}} \mapsto (a\ \mmod{F^{\ast}})\smkey{Y^{\ast}G^{\ast}}{F^{\ast}})) \\
    & Y^{\Tm} && \triangleq{ } &&
    \lambda A^{\bullet}\ a\ a^{\bullet}\ \mmod{Y^{\ast}Y_{\ast}} \mapsto a^{\bullet}\smkey{Y^{\ast}Y_{\ast}}{\bullet} \\
    & Y^{\Tm,-1} && \triangleq{ } &&
    \lambda A^{\bullet}\ a\ a^{\dagger}\ \mapsto (a^{\dagger}\ \mmod{Y^{\ast}Y_{\ast}})\smkey{\bullet}{Y^{\ast}Y_{\ast}}
  \end{alignat*}

  More generally, for every (global) telescope $X$ of $\CS$, we have the following isomorphisms
  \begin{alignat*}{3}
    & Y^{[X]\Ty} && :{ } &&
    (A : \mmod{F^{\ast}} \to X \to \Ty^{\CS}) \\
    &&&&& \to (\forall x\ (x^{\bullet} : X^{\bullet}\ x) \to \Ty^{\bullet}\ (A \circledast x)) \\
    &&&&& \simeq (\forall \mmod{Y^{\ast}}\ x\ (x^{\dagger} : X^{\dagger}\ x) \to \Ty^{\dagger}\ (\lambda \mmod{G^{\ast}} \mapsto (A\ \mmod{F^{\ast}})\smkey{Y^{\ast}G^{\ast}}{F^{\ast}})\ (x\  \mmod{G^{\ast}})) \\
    & Y^{[X]\Ty} && \triangleq{ } &&
    \lambda A\ A^{\bullet}\ \mmod{Y^{\ast}}\ x\ x^{\dagger}\ \mmod{Y_{\ast}} \mapsto A^{\bullet}\smkey{Y^{\ast}Y_{\ast}}{\bullet}\ (Y^{X,-1}\ (\lambda \mmod{Y^{\ast}} \mapsto x^{\dagger}\smkey{Y_{\ast}Y^{\ast}}{\bullet})) \\
    & Y^{[X]\Ty,-1} && \triangleq{ } &&
    \lambda A\ A^{\dagger}\ x^{\bullet} \mapsto (A^{\dagger}\ \mmod{Y^{\ast}}\ (Y^{X}\ x^{\bullet}\ \mmod{Y^{\ast}})\ \mmod{Y_{\ast}})\smkey{\bullet}{Y^{\ast}Y_{\ast}} \\
    & Y^{[X]\Tm} && :{ } &&
    \forall (A : \mmod{F^{\ast}} \to X \to \Ty^{\CS})\ (A^{\bullet} : \forall x\ (x^{\bullet} : X^{\bullet}\ x) \to \Ty^{\bullet}\ (A \circledast x)) \\
    &&&&& \phantom{\forall{ }} (a : \mmod{F^{\ast}} \to (x : X) \to \Tm^{\CS}\ (A\ \mmod{F^{\ast}}\ x)) \\
    &&&&& \to (\forall x\ (x^{\bullet} : X^{\bullet}\ x) \to \Tm^{\bullet}\ (A^{\bullet}\ x^{\bullet})\ (a \circledast x)) \\
    &&&&& \simeq (\forall \mmod{Y^{\ast}}\ x\ (x^{\dagger} : X^{\dagger}\ x) \to \\
    &&&&& \hphantom{\simeq (} \Tm^{\dagger}\ (Y^{[X]\Ty}\ A^{\bullet}\ \mmod{Y^{\ast}}\ x^{\dagger})\ (\lambda \mmod{G^{\ast}} \mapsto (a\ \mmod{F^{\ast}})\smkey{Y^{\ast}G^{\ast}}{F^{\ast}})\ (x\  \mmod{G^{\ast}})) \\
    & Y^{[X]\Tm} && \triangleq{ } &&
    \lambda a\ a^{\bullet}\ \mmod{Y^{\ast}}\ x\ x^{\dagger}\ \mmod{Y_{\ast}} \mapsto a^{\bullet}\smkey{Y^{\ast}Y_{\ast}}{\bullet}\ (Y^{X,-1}\ (\lambda \mmod{Y^{\ast}} \mapsto x^{\dagger}\smkey{Y_{\ast}Y^{\ast}}{\bullet})) \\
    & Y^{[X]\Tm,-1} && \triangleq{ } &&
    \lambda a\ a^{\dagger}\ x^{\bullet} \mapsto (a^{\dagger}\ \mmod{Y^{\ast}}\ (Y^{X}\ x^{\bullet}\ \mmod{Y^{\ast}})\ \mmod{Y_{\ast}})\smkey{\bullet}{Y^{\ast}Y_{\ast}}
  \end{alignat*}
  where $X^{\bullet}$, $X^{\dagger}$ and $Y^{X}$ can be defined by induction on $X$.
  In particular we obtain bijective actions of $Y$ on every derived sort of the theory.

  It remains to define the type-theoretic operations of $\CS^{\dagger}$.
  Each operation of $\CS^{\dagger}$ should be derived from the corresponding operation of $\CS^{\bullet}$.
  We want all of the displayed operations to be preserved by $Y^{\Ty}$ and $Y^{\Tm}$.

  In the case of the $\Pi$ type former, this translates to the following equation, internally to $\CPsh^{\CC}$.
  \begin{alignat*}{3}
    & Y^{\Ty}\ (\Pi^{\bullet}\ A^{\bullet}\ B^{\bullet})\ \mmod{Y^{\ast}} && ={ } && \Pi^{\dagger}\ (Y^{\Tm}\ A^{\bullet}\ \mmod{Y^{\ast}})\ (Y^{[\Tm]\Ty}\ B^{\bullet}\ \mmod{Y^{\ast}})
  \end{alignat*}

  To define $\Pi^{\dagger}$, we essentially have solve this equation
  We look for a candidate with the following shape (where $\bm{A^{\ddagger}}$, $\bm{B^{\ddagger}}$, \etc, are still to be determined).
  \begin{alignat*}{3}
    & \Pi^{\dagger} && ={ } && \lambda A^{\dagger}\ B^{\dagger}\ \mmod{Y_{\ast}} \mapsto \Pi^{\bullet}\ \bm{A^{\ddagger}}(A^{\dagger}, \mmod{Y_{\ast}})\ \bm{B^{\ddagger}}(B^{\dagger}, \mmod{Y_{\ast}})
  \end{alignat*}

  The following equation should then hold internally to $\CPsh^{\CP}$.
  \begin{alignat*}{3}
    & \lambda \mmod{Y_{\ast}} \mapsto \Pi^{\bullet}\ A^{\bullet}\ B^{\bullet} && ={ } && \lambda \mmod{Y_{\ast}} \mapsto \Pi^{\bullet}\ \bm{A^{\ddagger}}(Y^{\Ty}\ A^{\bullet}\ \mmod{Y^{\ast}}, \mmod{Y_{\ast}})\ \bm{B^{\ddagger}}(Y^{[\Tm]\Ty}\ B^{\bullet}\ \mmod{Y^{\ast}}, \mmod{Y_{\ast}})
  \end{alignat*}

  We use the following definitions for $\bm{A^{\ddagger}}$ and $\bm{B^{\ddagger}}$.
  \begin{alignat*}{3}
    & \bm{A^{\ddagger}}(A^{\dagger}, \mmod{Y_{\ast}}) && \triangleq{ } && A^{\dagger}\ \mmod{Y_{\ast}} \\
    & \bm{B^{\ddagger}}(B^{\dagger}, \mmod{Y_{\ast}}) && \triangleq{ } &&
    \lambda a^{\bullet} \mapsto (B^{\dagger}\smkey{Y_{\ast}Y^{\ast}}{\bullet}\ (Y^{\Tm}\ a^{\bullet}\ \mmod{Y^{\ast}})\ \mmod{Y_{\ast}})\smkey{\bullet}{Y^{\ast}Y_{\ast}}
  \end{alignat*}

  Thus we obtain
  \begin{alignat*}{3}
    & \Pi^{\dagger} && ={ } && \lambda A^{\dagger}\ B^{\dagger}\ \mmod{Y_{\ast}} \mapsto \Pi^{\bullet}\ (A^{\dagger}\ \mmod{Y_{\ast}})\ (\lambda a^{\bullet} \mapsto (B^{\dagger}\smkey{Y_{\ast}Y^{\ast}}{\bullet}\ (Y^{\Tm}\ a^{\bullet}\ \mmod{Y^{\ast}})\ \mmod{Y_{\ast}})\smkey{\bullet}{Y^{\ast}Y_{\ast}})
  \end{alignat*}

  In general, we have to solve equations of the form
  \[ \bm{y^{\ddagger}}(Y^{[X]\Ty}\ y^{\bullet}\ \mmod{Y^{\ast}}, \mmod{Y_{\ast}}) = y^{\bullet}. \]
  or
  \[ \bm{y^{\ddagger}}(Y^{[X]\Tm}\ y^{\bullet}\ \mmod{Y^{\ast}}, \mmod{Y_{\ast}}) = y^{\bullet}. \]
  They admit solutions with the following shape:
  \[ \bm{y^{\ddagger}}(y^{\dagger}, \mmod{Y_{\ast}}) \triangleq \lambda x^{\bullet} \mapsto (y^{\dagger}\smkey{Y_{\ast}Y^{\ast}}{\bullet}\ (Y^{X}\ x^{\bullet}\ \mmod{Y^{\ast}})\ \mmod{Y_{\ast}})\smkey{\bullet}{Y^{\ast}Y_{\ast}}. \]

  This provides a definition of all displayed type-theoretic operations of $\CS^{\dagger}$.
  \begin{alignat*}{3}
    & \Pi^{\dagger}\ A^{\dagger}\ B^{\dagger} && ={ } && \lambda \mmod{Y_{\ast}} \mapsto \Pi^{\bullet}\ (A^{\dagger}\ \mmod{Y_{\ast}})\ (\lambda a^{\bullet} \mapsto (B^{\dagger}\smkey{Y_{\ast}Y^{\ast}}{\bullet}\ (Y^{\Tm}\ a^{\bullet}\ \mmod{Y^{\ast}})\ \mmod{Y_{\ast}})\smkey{\bullet}{Y^{\ast}Y_{\ast}}) \\
    & \app^{\dagger}\ f^{\dagger}\ a^{\dagger} && ={ } && \lambda \mmod{Y_{\ast}} \mapsto \app^{\bullet}\ (f^{\dagger}\ \mmod{Y_{\ast}})\ (a^{\dagger}\ \mmod{Y_{\ast}}) \\
    & \lam^{\dagger}\ b^{\dagger} && ={ } && \lambda \mmod{Y_{\ast}} \mapsto \lam^{\bullet}\ (\lambda a^{\bullet} \mapsto (b^{\dagger}\smkey{Y_{\ast}Y^{\ast}}{\bullet}\ (Y^{\Tm}\ a^{\bullet}\ \mmod{Y^{\ast}})\ \mmod{Y_{\ast}})\smkey{\bullet}{Y^{\ast}Y_{\ast}}) \\
    &&&&& \dots
  \end{alignat*}

  These operations satisfy the equations that hold in $\CS^{\bullet}$.
\end{proof}

\subsection{Displayed inserters}

We fix two parallel displayed functors $K, L : \CC \to \CD$ over a base category $\CS$.

\begin{mathpar}
  \begin{tikzcd}
    \CC \ar[rd, "F"', -{Triangle[open]}] \ar[rr, "K", shift left] \ar[rr, "L"', shift right] && \CD \ar[dl, "G", -{Triangle[open]}] \\
    &\CS&
  \end{tikzcd}
\end{mathpar}

\begin{definition}
  The \defemph{displayed inserter} of $K$ and $L$ is a displayed category $I : \CI \ratri \CC$ over $\CC$.
  \begin{itemize}
    \item A object of $\CI$ displayed over an object $x$ of $\CC$ is a displayed morphism
          \[ \iota^{x} : \CD(K\ x \to_{\id_{F\ x}} L\ x). \]
    \item A morphism of $\CI$ from $\iota^{x}$ to $\iota^{y}$ displayed over $f : \CC(x \to y)$ is a proof of the commutation of the square
          \begin{mathpar}
            \begin{tikzcd}
              K\ x \ar[r, "\iota^{x}"] \ar[d, "K\ f"'] & L\ x \ar[d, "L\ f"] \\
              K\ y \ar[r, "\iota^{y}"] & L\ y
            \end{tikzcd}
          \end{mathpar}
  \end{itemize}

  There is a natural transformation $\iota : (I \cdot K) \Ra (I \cdot L)$ formed by the morphisms $\iota^{x}$.

  The category $\CI$ satisfies the following universal property: for every category $\CA$ along with a functor $A : \CA \to \CC$ and a natural transformation $\gamma : (A \cdot K) \Ra (A \cdot L)$, there exists a unique functor $B : \CA \to \CI$ such that $A = (B \cdot I)$ and $\gamma = (B \cdot \iota)$.
  \lipicsEnd{}
\end{definition}

\begin{proposition}[Internally to $\CPsh^{\CI}$]\label{prop:inserter_rep}
  Assume given locally representable presheaves
  \begin{alignat*}{3}
    & A^{\CS} && :{ } && \mmod{I^{\ast}F^{\ast}} \to \SRepPsh^{\CS} \\
    & A^{\CC} && :{ } && \mmod{I^{\ast}} \to \SRepPsh^{\CC} \\
    & A^{\CD} && :{ } && \mmod{I^{\ast}L^{\ast}} \to \SRepPsh^{\CD}
  \end{alignat*}
  along with actions of $K$, $L$ and $G$ on these presheaves
  \begin{alignat*}{3}
    & K^{A} && :{ } && \mmod{I^{\ast}} \to A^{\CC}\ \mmod{I^{\ast}} \to \mmod{K^{\ast}} \to (A^{\CD}\ \mmod{I^{\ast}F^{\ast}})\emkey{\iota}{I^{\ast}K^{\ast}}{I^{\ast}L^{\ast}} \\
    & L^{A} && :{ } && \mmod{I^{\ast}} \to A^{\CC}\ \mmod{I^{\ast}} \to \mmod{L^{\ast}} \to A^{\CD}\ \mmod{I^{\ast}L^{\ast}} \\
    & G^{A} && :{ } && \mmod{I^{\ast}L^{\ast}} \to A^{\CD}\ \mmod{I^{\ast}L^{\ast}} \to \mmod{G^{\ast}} \to (A^{\CS}\ \mmod{I^{\ast}F^{\ast}})\smkey{L^{\ast}G^{\ast}}{F^{\ast}}
  \end{alignat*}
  such that $L^{A}$ and $G^{A}$ preserve context extensions and the two composed actions of $F$ on $A$ coindide, \ie the equality
  \begin{alignat*}{1}
    & (G^{A}\ \mmod{I^{\ast}L^{\ast}}\ (L^{A}\ \mmod{I^{\ast}}\ a\ \mmod{L^{\ast}})\ \mmod{G^{\ast}})\smkey{F^{\ast}}{L^{\ast}G^{\ast}} \\
    & \quad { }= ((G^{A}\ \mmod{I^{\ast}L^{\ast}})\emkey{\iota}{I^{\ast}K^{\ast}}{I^{\ast}L^{\ast}}\ (K^{A}\ \mmod{I^{\ast}}\ a\ \mmod{K^{\ast}})\ \mmod{G^{\ast}})\smkey{F^{\ast}}{K^{\ast}G^{\ast}}
  \end{alignat*}
  holds over the context $(\mmod{I^{\ast}}, a : A^{\CC}\ \mmod{I^{\ast}}, \mmod{F^{\ast}})$.

  Then the presheaf
  \[ A^{\CI} \triangleq \mmod{I^{\ast}} \to \{ a : A^{\CC}\ \mmod{I^{\ast}} \mid \mmod{K^{\ast}} \to (K^{A}\ \mmod{I^{\ast}}\ a\ \mmod{K^{\ast}}) = (L^{A}\ \mmod{I^{\ast}}\ a\ \mmod{L^{\ast}})\emkey{\iota}{I^{\ast}K^{\ast}}{I^{\ast}L^{\ast}} \}\]
  is locally representable and the first projection map
  \[ I^{A} : A^{\CI} \to \mmod{I^{\ast}} \to A^{\CC}\ \mmod{I^{\ast}} \]
  preserves context extensions.
\end{proposition}
\begin{proof}
  We translate the statement externally.
  Fix an object $(x,\iota^{x})$ of $\CI$.
  We have locally representable dependent presheaves
  \begin{alignat*}{3}
    & A^{\CS} && :{ } && \forall (y : \CS^{\op}) (\rho : \CS(y \to F\ x)) \to \CSet \\
    & A^{\CC} && :{ } && \forall (y : \CC^{\op}) (\rho : \CC(y \to x)) \to \CSet \\
    & A^{\CD} && :{ } && \forall (y : \CD^{\op}) (\rho : \CD(y \to L\ x)) \to \CSet
  \end{alignat*}
  and dependent natural transformations
  \begin{alignat*}{3}
    & K^{A} && :{ } && \forall (y : \CC^{\op}) (\rho : \CC(y \to x)) \to A^{\CC}\ \rho \to A^{\CD}\ (K\ \rho \cdot \iota^{x}) \\
    & L^{A} && :{ } && \forall (y : \CC^{\op}) (\rho : \CC(y \to x)) \to A^{\CC}\ \rho \to A^{\CD}\ (L\ \rho) \\
    & G^{A} && :{ } && \forall (y : \CD^{\op}) (\rho : \CC(y \to L\ x)) \to A^{\CD}\ \rho \to A^{\CS}\ (G\ \rho)
  \end{alignat*}
  such that $L^{A}$ and $G^{A}$ preserves context extensions and such that for every $\rho : \CC(y \to x)$ and $a : A^{\CC}\ \rho$,
  we have $G^{A}\ (L\ \rho)\ (L^{A}\ \rho\ a) = G^{A}\ (K\ \rho \cdot \iota^{x})\ (K^{A}\ \rho\ a)$.

  We have to show that the dependent presheaf
  \begin{alignat*}{3}
    & A^{\CI} && :{ } && \forall ((y,\iota^{y}) : \CI^{\op}) (\rho : \CI((y,\iota^{y}) \to (x,\iota^{x}))) \to \CSet \\
    & A^{\CI}\ (y,\iota^{y})\ \rho && \triangleq{ } && \{ a : A^{\CC}\ \rho \mid K^{A}\ \rho\ a = (L^{A}\ \rho\ a)[\iota^{y}] \}
  \end{alignat*}
  is locally representable.

  Fix some object $(y,\iota^{y})$ of $\CI$ along with $\rho : \CI((y,\iota^{y}) \to (x,\iota^{x}))$.
  Recall that $\rho$ is a morphism $\rho : \CC(y \to x)$ such that $\iota^{y} \cdot K\ \rho = L\ \rho \cdot \iota^{x}$.

  We have to show that the presheaf
  \begin{alignat*}{3}
    & A^{\CI}_{\mid (y,\iota^{y})} && :{ } && \forall (z : \CI^{\op}) (\sigma : \CI(z \to (y,\iota^{y}))) \to \CSet \\
    & A^{\CI}_{\mid (y,\iota^{y})}\ \sigma && \triangleq{ } && A^{\CI}\ (\sigma \cdot \rho)
  \end{alignat*}
  is representable.

  We know that $A^{\CC}_{\mid y}$ and $A^{\CD}_{\mid L\ y}$ are representable and that $L^{A}$ preserves context extensions.
  Thus we have some representing object $\bm{p} : \CC(y^{\rhd} \to y)$ of $A^{\CC}_{\mid y}$, and we know that $L\ \bm{p} : \CD(L\ y^{\rhd} \to L\ y)$ is a representing object of $A^{\CD}_{\mid L\ y}$.
  We have a generic element $\bm{q} : A^{\CC}\ (\bm{p} \cdot \rho)$ for $A^{\CC}_{\mid y}$, and $L^{A}\ (\bm{p} \cdot \rho)\ \bm{q} : A^{\CD}\ (L\ (\bm{p} \cdot \rho))$ is a generic element for $A^{\CD}_{\mid L\ y}$.

  We construct a morphism $\iota^{y^{\rhd}} : \CD(K\ y^{\rhd} \to L\ y^{\rhd})$ such that the square
  \begin{mathpar}
    \begin{tikzcd}
      K\ y^{\rhd} \ar[d, "K\ \bm{p}"'] \ar[r, "\iota^{y^{\rhd}}"] & L\ y^{\rhd} \ar[d, "L\ \bm{p}"] \\
      K\ y \ar[r, "\iota^{y}"] & L\ y
    \end{tikzcd}
  \end{mathpar}
  commutes and such that $G\ \iota^{y^{\rhd}} = \id_{F\ y}$.

  By the universal property of $L\ y^{\rhd}$, we define $\iota^{y^{\rhd}}$ as the extension of $K\ \bm{p} \cdot \iota^{y}$ by the element $K^{A}\ (\bm{p} \cdot \rho)\ \bm{q} : A^{\CD}\ (K\ \bm{p} \cdot \iota^{y} \cdot L\ \rho)$.
  \[ \iota^{y^{\rhd}} \triangleq \angles{K\ \bm{p} \cdot \iota^{y}, K^{A}\ (\bm{p} \cdot \rho)\ \bm{q}}. \]
  We can then compute
  \begin{alignat*}{3}
    & G\ \iota^{y^{\rhd}}
    && ={ } && G\ \angles{K\ \bm{p} \cdot \iota^{y}, K^{A}\ (\bm{p} \cdot \rho)\ \bm{q}} \\
    &&& ={ } && \angles{G\ (K\ \bm{p}) \cdot G\ \iota^{y}, G^{A}\ (K\ (\bm{p} \cdot \rho) \cdot \iota^{x})\ (K^{A}\ (\bm{p} \cdot \rho)\ \bm{q})} \\
    &&& ={ } && \angles{G\ (L\ \bm{p}), G^{A}\ (L\ (\bm{p} \cdot \rho))\ (L^{A}\ (\bm{p} \cdot \rho)\ \bm{q})} \\
    &&& ={ } && G\ \angles{L\ \bm{p}, L^{A}\ (\bm{p} \cdot \rho)\ \bm{q}} \\
    &&& ={ } && G\ (L\ \angles{\bm{p},\bm{q}}) \\
    &&& ={ } && \id
  \end{alignat*}

  This defines an object $(y^{\rhd}, \iota^{y^{\rhd}})$ of $\CI$, equipped with a projection $\bm{p}$ into $(y,\iota^{y})$.
  It remains to show that this object represents $A^{\CI}_{\mid (y,\iota^{y})}$.

  Fix an object $(z,\iota^{z})$ of $\CI$ along with a morphism $\sigma : \CI((z,\iota^{z}) \to (y,\iota^{z}))$.
  We know that factorizations of $\sigma : \CC(z \to y)$ through $\bm{p} : \CC(y^{\rhd} \to y)$ are in natural bijection with elements of $A^{\CC}\ (\sigma \cdot \rho)$.
  We extend this bijection to $\CI$.
  Because a displayed morphism of $\CI$ over a morphism of $\CC$ only consists of propositional data, we only need to construct a logical equivalence at the level of $A^{\CI}$.

  Take an element $a : A^{\CI}\ (\sigma \cdot \rho)$.
  By the universal property of $y^{\rhd}$, we have an extended morphism $\angles{\sigma, a} : \CC(z \to y^{\rhd})$.
  Let's show that the square
  \begin{mathpar}
    \begin{tikzcd}
      K\ z \ar[r, "\iota^{z}"] \ar[d, "K\ \angles{\sigma,a}"'] & L\ z \ar[d, "L\ \angles{\sigma,a}"] \\
      K\ y^{\rhd} \ar[r, "\iota^{y^{\rhd}}"] & L\ y^{\rhd}
    \end{tikzcd}
  \end{mathpar}
  commutes.

  By the universal property of $L^{y^{\rhd}}$, both $K\ \angles{\sigma,a} \cdot \iota^{y^{\rhd}}$ and $\iota^{z} \cdot L\ \angles{\sigma,a}$ can be written as the extension of some morphism in $\CD(K\ z \to L\ y)$ by some element of $A^{\CD}$.
  We compute
  \begin{alignat*}{3}
    & K\ \angles{\sigma,a} \cdot a^{y^{\rhd}} && ={ } && K\ \angles{\sigma,a} \cdot \angles{K\ \bm{p} \cdot \iota^{y}, K^{A}\ (\bm{p} \cdot \rho)\ \bm{q}} \\
    &&& ={ } && \angles{K\ (\angles{\sigma,a} \cdot \bm{p}) \cdot \iota^{y}, (K^{A}\ (\bm{p} \cdot \rho)\ \bm{q})[K\ \angles{\sigma,a}]} \\
    &&& ={ } && \angles{K\ \sigma \cdot \iota^{y}, (K^{A}\ (\bm{p} \cdot \rho)\ \bm{q})[K\ \angles{\sigma,a}]} \\
    &&& ={ } && \angles{K\ \sigma \cdot \iota^{y}, K^{A}\ (\sigma \cdot \rho)\ a}
  \end{alignat*}
  and
  \begin{alignat*}{3}
    & \iota^{z} \cdot L\ \angles{\sigma,a} && ={ } && \iota^{z} \cdot \angles{L\ \sigma, L^{A}\ (\sigma \cdot \rho)\ a} \\
    &&& ={ } && \angles{\iota^{z} \cdot L\ \sigma, (L^{A}\ (\sigma \cdot \rho)\ a)[\iota^{z}]}
  \end{alignat*}
  Now $K\ \sigma \cdot \iota^{y} = \iota^{z} \cdot L\ \sigma$ because $\sigma$ is a morphism of $\CI$, and $K^{A}\ (\sigma \cdot \rho)\ a = (L^{A}\ (\sigma \cdot \rho)\ a)[\iota^{z}]$ because $a$ is an element of $A^{\CI}$.
  Thus we have $K\ \angles{\sigma,a} \cdot \iota^{y^{\rhd}} = \iota^{z} \cdot L\ \angles{\sigma,a}$, as needed.

  This concludes the proof that $A^{\CI}_{\mid (y,\iota^{y})}$ is representable.
  The dependent presheaf $A^{\CI}$ is thus locally representable.

  Finally, we also have to check that the representing objects that we have constructed are natural in $(x,\iota^{x})$.
  This follows from the fact that the representing objects of $A^{\CC}$ are natural in $x$.
\end{proof}

We now return to the setting of our induction principles.
We fix a base model $\CS$ of $\Th_{\Pi,\BoolTy}$, a functor $F : \CC \to \CS$ equipped with a displayed model without context extensions $\CS^{\bullet}$, a factorization $(\CC \xrightarrow{Y} \CP \xrightarrow{G} \CS, \CS^{\dagger})$ and a section $S$ of $\CS^{\dagger}$.
We let $\CI(\CS^{\bullet})$ be the displayed inserter of $Y$ and $F \cdot S$.

The following diagram describes the categories and functors in play.
\begin{mathpar}
 \begin{tikzcd}
    \CI(\CS^{\bullet}) \ar[d, "I", -{Triangle[open]}] && \\
    \CC \ar[rr, "Y"] \ar[dr, "F", -{Triangle[open]}] &&
    \CP \ar[dl, "G"', -{Triangle[open]}] \\
    & \CS \ar[ur, "S_{0}"', bend right] &
  \end{tikzcd}
\end{mathpar}

\begin{proposition}\label{prop:inserter_terminal}
  If $\CC$ has a terminal object that is preserved by $F : \CC \to \CS$, then $\CI(\CS^{\bullet})$ has a terminal object that is preserved by $I : \CI(\CS^{\bullet}) \to \CC$.
\end{proposition}
\begin{proof}
  The terminal object of $\CI(\CS^{\bullet})$ is $(\diamond,\iota^{\diamond})$, where $\diamond$ is the terminal object of $\CC$ and $\iota^{\diamond} : \CP(Y\ \diamond \to_{\id} S\ (F\ \diamond))$.
  Since both $S$ and $F$ both preserve terminal objects, there is a unique such $\iota^{\diamond}$.
  This also implies that $(\diamond,\iota^{\diamond})$ is terminal.
\end{proof}

We can also specialize \cref{prop:inserter_rep} to this situation.
\begin{proposition}[Internally to $\CPsh^{\CI(S^{\bullet})}$]\label{prop:inserter_rep_cor}
  Assume given the following data:
  \begin{itemize}
    \item A locally representable presheaf $A^{\CS} : \mmod{I^{\ast}F^{\ast}} \to \SRepPsh^{\CS}$.
    \item A dependent presheaf $A^{\bullet} : \mmod{I^{\ast}} \to (\mmod{F^{\ast}} \to A^{\CS}\ \mmod{I^{\ast}F^{\ast}}) \to \SPsh^{\CC}$.
    \item A dependent presheaf $A^{\dagger} : \mmod{I^{\ast}F^{\ast}S_{0}^{\ast}} \to (\mmod{G^{\ast}} \to (A^{\CS}\ \mmod{I^{\ast}F^{\ast}})\smkey{S_{0}^{\ast}G^{\ast}}{\bullet}) \to \SPsh^{\CS}$
          such that the first projection map
          \begin{alignat*}{1}
            & (a : \mmod{G^{\ast}} \to (A^{\CS}\ \mmod{I^{\ast}F^{\ast}})\smkey{S_{0}^{\ast}G^{\ast}}{\bullet}) \times
            (a^{\dagger} : A^{\dagger}\ \mmod{I^{\ast}F^{\ast}S_{0}^{\ast}}\ a) \\
            & \quad { }\xrightarrow{\lambda (a,a^{\dagger}) \mapsto a}
            (\mmod{G^{\ast}} \to (A^{\CS}\ \mmod{F^{\ast}})\smkey{S_{0}^{\ast}G^{\ast}}{\bullet})
          \end{alignat*}
          has a locally representable domain and preserves context extensions.
    \item A bijective action
          \begin{alignat*}{3}
            & Y^{A} && :{ } && \mmod{I^{\ast}}\ \{a : \mmod{F^{\ast}} \to A^{\CS}\ \mmod{I^{\ast}F^{\ast}}\} \to \\
            &&&&& A^{\bullet}\ \mmod{I^{\ast}}\ a \simeq (\mmod{Y^{\ast}} \to (A^{\dagger}\ \mmod{I^{\ast}F^{\ast}S_{0}^{\ast}})\emkey{\iota}{I^{\ast}Y^{\ast}}{I^{\ast}F^{\ast}S_{0}^{\ast}}\ (\lambda \mmod{G^{\ast}} \mapsto (a\ \mmod{G^{\ast}})\smkey{F^{\ast}}{Y^{\ast}G^{\ast}})).
          \end{alignat*}
    \item An action
          \[ S^{A} : \mmod{I^{\ast}F^{\ast}}\ (a : A^{\CS}\ \mmod{I^{\ast}F^{\ast}})\ \mmod{S_{0}^{\ast}} \to A^{\dagger}\ \mmod{I^{\ast}F^{\ast}S_{0}^{\ast}}\ (\lambda \mmod{G^{\ast}} \mapsto a\smkey{S_{0}^{\ast}G^{\ast}}{\bullet}) \]
          whose induced total action preserves context extensions.
    \item A locally representable presheaf $A^{\CC} : \mmod{I^{\ast}} \to \SRepPsh^{\CC}$.
    \item An action $F^{A} : \mmod{I^{\ast}} \to A^{\CC}\ \mmod{I^{\ast}} \to \mmod{F^{\ast}} \to A^{\CS}\ \mmod{I^{\ast}F^{\ast}}$ that preserves context extensions.
    \item A map $f : \mmod{I^{\ast}} \to (a : A^{\CC}\ \mmod{I^{\ast}}) \to A^{\bullet}\ \mmod{I^{\ast}F^{\ast}}\ (F^{A}\ \mmod{I^{\ast}}\ a)$.
  \end{itemize}

  We pose
  \begin{alignat*}{3}
    & S^{A} && :{ } && \mmod{I^{\ast}}\ (a : \mmod{F^{\ast}} \to A^{\CS}\ \mmod{I^{\ast}F^{\ast}}) \to A^{\bullet}\ \mmod{I^{\ast}}\ a \\
    & S^{A}\ \mmod{I^{\ast}}\ a && \triangleq{ } && Y^{A,-1}\ (\lambda \mmod{Y^{\ast}} \mapsto (S^{A}\ (a\ \mmod{F^{\ast}})\ \mmod{S_{0}^{\ast}})\emkey{\iota}{I^{\ast}Y^{\ast}}{I^{\ast}F^{\ast}S_{0}^{\ast}})
  \end{alignat*}

  Then the presheaf
  \begin{alignat*}{3}
    & A^{\CI(S^{\bullet})} && \triangleq{ } && \mmod{I^{\ast}} \to \{ a : A^{\CC}\ \mmod{I^{\ast}} \\
    &&&&& \phantom{\mmod{I^{\ast}} \to { }} \mid S^{A}\ \mmod{I^{\ast}}\ (\lambda \mmod{F^{\ast}} \mapsto F^{A}\ \mmod{I^{\ast}}\ a\ \mmod{F^{\ast}}) = f\ \mmod{I^{\ast}}\ a \}
  \end{alignat*}
  is locally representable and the first projection map
  \[ A^{\CI(S^{\bullet})} \to \mmod{I^{\ast}} \to A^{\CC}\ \mmod{I^{\ast}} \]
  preserves context extensions.
\end{proposition}
\begin{proof}
  This follows from \cref{prop:inserter_rep}, applied to the following presheaves
  \begin{alignat*}{3}
    & B^{\CS}\ \mmod{I^{\ast}F^{\ast}}
    && \triangleq{ } && A^{\CS}\ \mmod{I^{\ast}F^{\ast}} \\
    & B^{\CC}\ \mmod{I^{\ast}}
    && \triangleq{ } && (a : \mmod{F^{\ast}} \to A^{\CS}\ \mmod{I^{\ast}F^{\ast}}) \times (A^{\bullet}\ \mmod{I^{\ast}}\ a) \\
    & B^{\CP}\ \mmod{I^{\ast}F^{\ast}S_{0}^{\ast}}
    && \triangleq{ } && (a : \mmod{G^{\ast}} \to (A^{\CS}\ \mmod{I^{\ast}F^{\ast}})\smkey{S_{0}^{\ast}G^{\ast}}{\bullet}) \times (A^{\dagger}\ \mmod{I^{\ast}F^{\ast}S_{0}^{\ast}}\ a)
  \end{alignat*}
  and to the following actions
   \begin{alignat*}{1}
     & Y^{A}\ \mmod{I^{\ast}}\ a\ \mmod{Y^{\ast}} \\
     & \quad \triangleq  (\lambda \mmod{G^{\ast}} \mapsto (F^{A}\ \mmod{I^{\ast}}\ a\ \mmod{F^{\ast}})\smkey{F^{\ast}}{Y^{\ast}G^{\ast}}, Y^{A}\ \mmod{I^{\ast}}\ (f\ \mmod{I^{\ast}}\ a)\ \mmod{Y^{\ast}}) \\
     & {(F \cdot S)}^{A}\ \mmod{I^{\ast}}\ a\ \mmod{F^{\ast}S_{0}^{\ast}} \\
     & \quad \triangleq (\lambda \mmod{G^{\ast}} \mapsto (F^{A}\ \mmod{I^{\ast}}\ a\ \mmod{F^{\ast}})\smkey{S_{0}^{\ast}G^{\ast}}{\bullet}, S^{A}\ \mmod{I^{\ast}F^{\ast}}\ (F^{A}\ \mmod{I^{\ast}}\ a\ \mmod{F^{\ast}})\ \mmod{S_{0}^{\ast}}) \\
     & G^{A}\ \mmod{I^{\ast}F^{\ast}S_{0}^{\ast}}\ (a, -)\ \mmod{G^{\ast}} \triangleq a\ \mmod{G^{\ast}}
     \tag*{\qedhere}
  \end{alignat*}
\end{proof}

\subsection{Relative sections}

From the point of view of $\CI$, the relative section of $\CS^{\bullet}$ already exists.
That is, we have, for every telescope $X$, operations
\begin{alignat*}{3}
  & S_{\iota}^{[X]\Ty} && :{ } && \forall \mmod{I^{\ast}} (A : \mmod{F^{\ast}} \to X \to \Ty)\ x\ (x^{\bullet} : X^{\bullet}\ x) \to \Ty^{\bullet}\ (A \circledast x) \\
  & S_{\iota}^{[X]\Ty} && \triangleq{ } && \lambda \mmod{I^{\ast}}\ A\ x^{\bullet} \mapsto Y^{[X]\Ty,-1}\ (\lambda \mmod{Y^{\ast}} \mapsto (S^{[X]\Ty}\ (A\ \mmod{F^{\ast}})\ \mmod{S_{0}^{\ast}})\emkey{\iota}{I^{\ast}Y^{\ast}}{I^{\ast}F^{\ast}S_{0}^{\ast}})\ x^{\bullet} \\
  & S_{\iota}^{[X]\Tm} && :{ } && \forall \mmod{I^{\ast}}\ A\ (a : \forall \mmod{F^{\ast}}\ x \to \Tm\ (A\ \mmod{F^{\ast}}\ x))\ x\ (x^{\bullet} : X^{\bullet}\ x) \\
  &&&&& \to \Tm^{\bullet}\ (S^{[X]\Ty}\mmod{I^{\ast}}\ A\ x\ x^{\bullet})\ (a \circledast x) \\
  & S_{\iota}^{[X]\Tm} && \triangleq{ } && \lambda \mmod{I^{\ast}}\ a\ x^{\bullet} \mapsto Y^{[X]\Tm,-1}\ (\lambda \mmod{Y^{\ast}} \mapsto (S^{[X]\Tm}\ (a\ \mmod{F^{\ast}})\ \mmod{S_{0}^{\ast}})\emkey{\iota}{I^{\ast}Y^{\ast}}{I^{\ast}F^{\ast}S_{0}^{\ast}})\ x^{\bullet},
\end{alignat*}
where $X^{\bullet} : (\mmod{F^{\ast}} \to X\ \mmod{F^{\ast}}) \to \SPsh^{\CC}$ is defined by induction on the telescope $X$.

Given any section $\angles{\alpha}$ of $I : \CI \to \CC$, we can then define
\begin{alignat*}{3}
  & S_{\alpha}^{[X]\Ty} && :{ } && \forall (A : \mmod{F^{\ast}} \to X \to \Ty)\ x\ (x^{\bullet} : X^{\bullet}\ x) \to \Ty^{\bullet}\ (A \circledast x) \\
  & S_{\alpha}^{[X]\Ty} && \triangleq{ } && \lambda A\ x^{\bullet} \mapsto (S^{[X]\Ty}_{\iota}\mmod{I^{\ast}})\smkey{\bullet}{\angles{\alpha}^{\ast}I^{\ast}}\ A\ x^{\bullet} \\
  & S_{\alpha}^{[X]\Tm} && :{ } && \forall A\ (a : \forall \mmod{F^{\ast}}\ x \to \Tm\ (A\ \mmod{F^{\ast}}\ x))\ x\ (x^{\bullet} : X^{\bullet}\ x) \\
  &&&&& \to \Tm^{\bullet}\ (S^{[X]\Ty}\ A\ x\ x^{\bullet})\ (a \circledast x) \\
  & S_{\alpha}^{[X]\Tm} && \triangleq{ } && \lambda a\ x^{\bullet} \mapsto (S^{[X]\Tm}_{\iota}\mmod{I^{\ast}})\smkey{\bullet}{\angles{\alpha}^{\ast}I^{\ast}}\ a\ x^{\bullet}
\end{alignat*}

Note that we have
\begin{alignat*}{3}
  & S_{\alpha}^{[X]\Ty} && ={ } && \lambda A\ x^{\bullet} \mapsto Y^{[X]\Ty,-1}\ (\lambda \mmod{Y^{\ast}} \mapsto (S^{[X]\Ty}\ (A\ \mmod{F^{\ast}})\ \mmod{S_{0}^{\ast}})\emkey{\alpha}{Y^{\ast}}{F^{\ast}S_{0}^{\ast}})\ x^{\bullet} \\
  & S_{\alpha}^{[X]\Tm} && ={ } && \lambda a\ x^{\bullet} \mapsto Y^{[X]\Tm,-1}\ (\lambda \mmod{Y^{\ast}} \mapsto (S^{[X]\Tm}\ (a\ \mmod{F^{\ast}})\ \mmod{S_{0}^{\ast}})\emkey{\alpha}{Y^{\ast}}{F^{\ast}S_{0}^{\ast}})\ x^{\bullet}
\end{alignat*}

These actions preserve all type-theoretic operations.
This follows from the fact that the actions of $S_{0}$ and $Y$ on types and terms preserve the type-theoretic operations.
We give the detailed proof for $\BoolTy$ and $\Pi$.
\begin{alignat*}{3}
  & S_{\alpha}^{[X]\Ty}\ (\lambda \mmod{F^{\ast}}\ x \mapsto \BoolTy)
  && ={ } && Y^{[X]\Ty,-1}\ (\lambda \mmod{Y^{\ast}}\ x^{\dagger} \mapsto (S^{[X]\Ty}\ \BoolTy\ \mmod{S_{0}^{\ast}}\ (x^{\dagger}\ \mmod{S_{0}^{\ast}}))\emkey{\alpha}{Y^{\ast}}{F^{\ast}S_{0}^{\ast}})
  \tag*{(definition of $S_{\alpha}^{[X]\Ty}$)} \\
  &&& ={ } && Y^{[X]\Ty,-1}\ (\lambda \mmod{Y^{\ast}}\ x^{\dagger} \mapsto \BoolTy^{\dagger}\emkey{\alpha}{Y^{\ast}}{F^{\ast}S_{0}^{\ast}})
  \tag*{($S_{\alpha}^{[X]\Ty}$ preserves $\BoolTy$)} \\
  &&& ={ } && Y^{[X]\Ty,-1}\ (\lambda \mmod{Y^{\ast}}\ x^{\dagger} \mapsto \BoolTy^{\dagger})
  \tag*{(commutation with $-\emkey{\alpha}{Y^{\ast}}{F^{\ast}S_{0}^{\ast}}$)} \\
  &&& ={ } && \BoolTy^{\bullet}
  \tag*{($Y^{[X]\Ty}$ preserves $\BoolTy$)}
\end{alignat*}
\begin{alignat*}{1}
  & S_{\alpha}^{[X]\Ty}\ (\lambda \mmod{F^{\ast}}\ x \mapsto \Pi\ (A\ \mmod{F^{\ast}}\ x)\ (B\ \mmod{F^{\ast}}\ x)) \\
  & \quad ={ } Y^{[X]\Ty,-1}\ (\lambda \mmod{Y^{\ast}}\ x^{\dagger} \mapsto (S_{\alpha}^{[X]\Ty}\ (\lambda x \mapsto \Pi\ (A\ \mmod{F^{\ast}}\ x)\ (B\ \mmod{F^{\ast}}\ x))\ \mmod{S_{0}^{\ast}}\ (x^{\dagger}\ \mmod{S_{0}^{\ast}}))\emkey{\alpha}{Y^{\ast}}{F^{\ast}S_{0}^{\ast}})
  \tag*{(definition of $S_{\alpha}^{[X]\Ty}$)} \\
  & \quad ={ } Y^{[X]\Ty,-1}\ (\lambda \mmod{Y^{\ast}}\ x^{\dagger} \mapsto (\Pi^{\dagger}\ (S_{\alpha}^{[X]\Ty}\ A\ \mmod{S_{0}^{\ast}}\ (x^{\dagger}\ \mmod{S_{0}^{\ast}}))\ (S_{\alpha}^{[X,\Tm]\Ty}\ B\ \mmod{S_{0}^{\ast}}\ (x^{\dagger}\ \mmod{S_{0}^{\ast}})))\emkey{\alpha}{Y^{\ast}}{F^{\ast}S_{0}^{\ast}})
  \tag*{($S_{\alpha}^{[X]\Ty}$ preserves $\Pi$)} \\
  & {\quad ={ } Y^{[X]\Ty,-1}\ (\lambda \mmod{Y^{\ast}}\ x^{\dagger} \mapsto \Pi^{\dagger}\ }
  (S_{\alpha}^{[X]\Ty}\ A\ \mmod{S_{0}^{\ast}}\ (x^{\dagger}\ \mmod{S_{0}^{\ast}}))\emkey{\alpha}{Y^{\ast}}{F^{\ast}S_{0}^{\ast}} \\
  & \hphantom{\quad ={ } Y^{[X]\Ty,-1}\ (\lambda \mmod{Y^{\ast}}\ x^{\dagger} \mapsto \Pi^{\dagger}\ }
  (S_{\alpha}^{[X,\Tm]\Ty}\ B\ \mmod{S_{0}^{\ast}}\ (x^{\dagger}\ \mmod{S_{0}^{\ast}}))\emkey{\alpha}{Y^{\ast}}{F^{\ast}S_{0}^{\ast}})
  \tag*{(commutation with $-\emkey{\alpha}{Y^{\ast}}{F^{\ast}S_{0}^{\ast}}$)} \\
  & {\quad ={ } \Pi^{\bullet}\ }
  (Y^{[X]\Ty,-1}\ (\lambda \mmod{Y^{\ast}}\ x^{\dagger} \mapsto (S_{\alpha}^{[X]\Ty}\ A\ \mmod{S_{0}^{\ast}}\ (x^{\dagger}\ \mmod{S_{0}^{\ast}}))\emkey{\alpha}{Y^{\ast}}{F^{\ast}S_{0}^{\ast}})) \\
  & \hphantom{\quad ={ } \Pi^{\bullet}\ }
  (Y^{[X,\Tm]\Ty,-1}\ (\lambda \mmod{Y^{\ast}}\ x^{\dagger} \mapsto (S_{\alpha}^{[X,\Tm]\Ty}\ B\ \mmod{S_{0}^{\ast}}\ (x^{\dagger}\ \mmod{S_{0}^{\ast}}))\emkey{\alpha}{Y^{\ast}}{F^{\ast}S_{0}^{\ast}})) \\
  \tag*{($Y^{[X]\Ty}$ preserves $\Pi$)} \\
  & \quad ={ } \Pi^{\bullet}\ (S^{[X]\Ty}_{\alpha}\ A)\ (S^{[X,\Tm]\Ty}_{\alpha}\ B)
  \tag*{(definitions of $S^{[X]\Ty}_{\alpha}$ and $S^{[X,\Tm]\Ty}_{\alpha}$)} \\
\end{alignat*}

\subsection{Induction principles}\label{sec:indPrinc}

\restateIndTerminal*
\begin{proof}
  By biinitiality of $\Init_{\Th_{\Pi,\BoolTy}}$, we have a section $S_{0}$ of the displayed model $\Init_{\Th_{\Pi,\BoolTy}}^{\dagger}$.

  By \cref{prop:inserter_terminal} $\CI(\Init_{\Th_{\Pi,\BoolTy}}^{\bullet})$ has a terminal object.
  This determines a section $\angles{\alpha}$ of $I : \CI(\Init_{\Th_{\Pi,\BoolTy}}^{\bullet}) \to \{\diamond\}$.
  We thus have a relative section $S_{\alpha}$ of $\Init_{\Th_{\Pi,\BoolTy}}^{\bullet}$.
\end{proof}

\restateIndRenamings*
\begin{proof}
  By biinitiality of $\Init_{\Th_{\Pi,\BoolTy}}$, we have a section $S_{0}$ of the displayed model $\Init_{\Th_{\Pi,\BoolTy}}^{\dagger}$.

  We now equip $\CI(\Init_{\Th_{\Pi,\BoolTy}}^{\bullet})$ with the structure of a renaming algebra.
  By \cref{prop:inserter_terminal} $\CI(\Init_{\Th_{\Pi,\BoolTy}}^{\bullet})$ has a terminal object.
  We pose
  \[ \Var^{\CI}\ A \triangleq \mmod{I^{\ast}} \to \{ a : \Var^{\CRen}\ (A\ \mmod{I^{\ast}}) \mid S_{\iota}^{\Tm}\ (\var\ a) = \var^{\bullet}\ \mmod{I^{\ast}}\ a \}. \]

  By \cref{prop:inserter_rep}, $\Var^{\CI}$ is a family of locally representable presheaves.
  Thus $\CI(\Init_{\Th_{\Pi,\BoolTy}}^{\bullet})$ has the structure of a renaming algebra.
  By biinitiality of $\CRen$, we obtain a section $\angles{\alpha}$ of $I : \CI(\Init_{\Th_{\Pi,\BoolTy}}^{\bullet}) \to \CRen$, and thus a relative section $S_{\alpha}$ of $\Init_{\Th_{\Pi,\BoolTy}}^{\bullet}$.

  The morphism $\angles{\alpha}$ of renaming algebras has an action on variables.
  \begin{alignat*}{1}
    & \angles{\alpha}^{\Var} : \forall (a : \Var^{\CRen}\ A)\ \mmod{\angles{\alpha}^{\ast}I^{\ast}} \to S_{\iota}^{\Tm}\ (\var\ a\smkey{\bullet}{\angles{\alpha}^{\ast}I^{\ast}}) = \var^{\bullet}\ \mmod{I^{\ast}}\ a\smkey{\bullet}{\angles{\alpha}^{\ast}I^{\ast}}
  \end{alignat*}
  Thus given any variable $a : \Var^{\CRen}\ A$, we know that $S_{\alpha}^{\Tm}\ (\var\ a) = (\var^{\bullet}\ \mmod{I^{\ast}})\smkey{\bullet}{\angles{\alpha}^{\ast}I^{\ast}}\ a$.
\end{proof}

\restateIndCube*
\begin{proof}
  By biinitiality of $\Init_{\mathsf{CTT}}$, we have a section $S_{0}$ of the displayed model $\Init_{\mathsf{CTT}}^{\dagger}$.

  We now equip $\CI(\Init_{\mathsf{CTT}}^{\bullet})$ with the structure of a renaming algebra.
  By \cref{prop:inserter_terminal} $\CI(\Init_{\mathsf{CTT}}^{\bullet})$ has a terminal object.
  We pose
  \[ \MI^{\CI} \triangleq \mmod{I^{\ast}} \to \{ i : \MI^{\square} \mid S_{\iota}^{\MI}\ (\mathsf{int}\ i) = \mathsf{int}^{\bullet}\ \mmod{I^{\ast}}\ i \}. \]
  Since we have $S^{\MI}_{\iota}\ 0^{\square} = 0^{\bullet}$ and $S^{\MI}_{\iota}\ 1^{\square} = 1^{\bullet}$, we can lift $0^{\square}$ and $1^{\square}$ to elements $0^{\CI},1^{\CI} : \MI^{\CI}$.

  By \cref{prop:inserter_rep}, $\MI^{\CI}$ is a locally representable presheaf.
  Thus $\CI(\Init_{\mathsf{CTT}}^{\bullet})$ has the structure of a cubical algebra.
  Furthermore $I : \CI(\Init_{\mathsf{CTT}}^{\bullet}) \to \square$ is a morphism of cubical algebras.
  By biinitiality of $\square$, we obtain a section $\angles{\alpha}$ of $I : \CI(\Init_{\mathsf{CTT}}^{\bullet}) \to \square$, and thus a relative section $S_{\alpha}$ of $\Init_{\mathsf{CTT}}^{\bullet}$.

  The morphism $\angles{\alpha}$ of cubical algebras has an action on the interval.
  \begin{alignat*}{1}
    & \angles{\alpha}^{\MI} : \forall (i : \MI^{\square})\ \mmod{\angles{\alpha}^{\ast}I^{\ast}} \to S_{\iota}^{\MI}\ (\mathsf{int}\ i\smkey{\bullet}{\angles{\alpha}^{\ast}I^{\ast}}) = \mathsf{int}^{\bullet}\ \mmod{I^{\ast}}\ i\smkey{\bullet}{\angles{\alpha}^{\ast}I^{\ast}}
  \end{alignat*}
  Thus given any interval element $i : \MI^{\square}$, we know that $S_{\alpha}^{\MI}\ (\mathsf{int}\ i) = (\mathsf{int}^{\bullet}\ \mmod{I^{\ast}})\smkey{\bullet}{\angles{\alpha}^{\ast}I^{\ast}}\ i$.
\end{proof}

\restateIndAtomicCube*
\begin{proof}
  Similar to the proofs of \cref{lem:ind_renamings} and \cref{lem:ind_cubes}.
\end{proof}


\section{Normalization}\label{sec:normalization}

In this section, we describe a normalization proof for a type theory $\Th_{\UU}$ with a hierarchy of universes indexed by natural numbers, closed under $\Pi$-types and boolean types.
This proof relies on an induction principle relative to $\CRen \to \Init_{\Th_{\UU}}$.

The proof follows the same structure as Coquand normalization proof from~\cite{CoquandNormalization}; it is an algebraic presentation of Normalization by Evaluation (NbE).
There is one important difference between the type theory $\Th_{\UU}$ and the type theory considered in~\cite{CoquandNormalization}.
Our type theory is not cumulative; type-formers at different universe level are distinct.
We believe that proving normalization for a cumulative type theory should be possible in our framework, but getting the details right is tricky and the definitions become very verbose.
These details, such as the precise definition of the normal forms, were omitted from Coquand's proof.
We prove slightly more; in addition to the existence of normal forms, we also prove uniqueness.
Proving uniqueness relies on the computation rules of relative sections.

\subsection{The type theory \texorpdfstring{$\Th_{\UU}$}{T\_{}U}}
We describe the structure of a model of $\Th_{\UU}$ over a category $\CC$, internally to $\CPsh^{\CC}$.

Types and terms are now indexed by universe levels.
\begin{alignat*}{3}
  & \Ty && :{ } && \forall (i : \Nat) \to \SPsh^{\CC} \\
  & \Tm && :{ } && \forall (i : \Nat)\ (A : \Ty_{i}) \to \SRepPsh^{\CC}
\end{alignat*}

We have lifting functions that can be used to move between different universe levels.
\begin{alignat*}{3}
  & \Lift && :{ } && \forall (i : \Nat) \to \Ty_{i} \to \Ty_{i+1} \\
  & \lift && :{ } && \forall (i : \Nat)\ (A : \Ty_{i}) \to \Tm_{i}\ A \simeq \Tm_{i+1}\ (\Lift_{i}\ A)
\end{alignat*}

We have a hierarchy of universes, indexed by universe levels.
The terms of the $i$-th universe are in bijection with the types of level $i$.
\begin{alignat*}{3}
  & \UU && :{ } && \forall (i : \Nat) \to \Ty_{i+1} \\
  & \El && :{ } && \forall (i : \Nat) \to \Tm_{i+1}\ \UU_{i} \simeq \Ty_{i}
\end{alignat*}

Finally, we have $\Pi$-types and boolean numbers types at every universe level.
The motive of the eliminator for booleans can be at any universe level.
\begin{alignat*}{3}
  & \Pi && :{ } && \forall i\ (A : \Ty_{i})\ (B : \Tm\ A \to \Ty_{i}) \to \Ty_{i} \\
  & \app && :{ } && \forall i\ A\ B \to \Tm_{i}\ (\Pi\ A\ B) \simeq ((a : \Tm\ A) \to \Tm\ (B\ a)) \\
  & - && :{ } && \elimBool\ P\ t\ f\ \true = t \\
  & - && :{ } && \elimBool\ P\ t\ f\ \false = f
\end{alignat*}
\begin{alignat*}{3}
  & \BoolTy && :{ } && \forall i \to \Ty_{i} \\
  & \true,\false && :{ } && \forall i \to \Tm_{i}\ \BoolTy \\
  & \elimBool && :{ } && \forall i\ j\ (P : \Tm\ \BoolTy \to \Ty_{j})\ (t : \Tm\ (P\ \true))\ (f : \Tm\ (P\ \false))\ (b : \Tm\ \BoolTy_{i}) \to \Tm\ (P\ b)
\end{alignat*}

We have, as for $\Th_{\Pi,\BoolTy}$, notions of displayed models without context extensions and of relative sections for $\Th_{\UU}$.
The following variant of \cref{lem:ind_renamings} can easily be proven for $\Th_{\UU}$.

\begin{definition}
  A \defemph{renaming algebra} over a model $\CS$ of $\Th_{\UU}$ is a category $\CR$ with a terminal object, along with a functor $F : \CR \to \CS$ preserving the terminal object, a locally representable dependent presheaf of variables
  \[ \Var^{\CR} : \forall i\ (A : \mmod{F^{\ast}} \to \Ty_{i}^{\CS}) \to \SRepPsh^{\CR} \]
  and an action on variables $\var : \forall i\ A\ (a : \Var_{i}\ A)\ \mmod{F^{\ast}} \to \Tm_{i}^{\CS}\ (A\ \mmod{F^{\ast}})$ that preserves context extensions.

  The category of renamings $\CRen_{\CS}$ over a model $\CS$ is defined as the biinitial renaming algebra over $\CS$.
  We denote the category of renamings of the biinitial model $\Init_{\Th_{\UU}}$ by $\CRen$.
  \lipicsEnd{}
\end{definition}

\begin{lemma}[Induction principle relative to $\CRen \to \Init_{\Th_{\UU}}$]\label{lem:ind_renamings_univ}
  Let $\Init_{\Th_{\UU}}^{\bullet}$ be a global displayed model without context extensions over $F : \CRen \to \Init_{\Th_{\UU}}$, along with, internally to $\CI(\CS^{\bullet})$, a global map
  \[ \var^{\bullet} : \forall \mmod{I^{\ast}}\ i\ (A : \mmod{F^{\ast}} \to \Ty_{i}) (a : \Var_{i}\ A) \to \Tm^{\bullet}\ (S_{\iota}^{\Ty}\ \mmod{I^{\ast}}\ A)\ (\var\ a). \]

  Then there exists a relative section $S_{\alpha}$ of $\Init_{\Th_{\UU}}^{\bullet}$.
  It satisfies the additional computation rule
  \[ S_{\alpha}^{\Tm}\ (\var\ a) = (\var^{\bullet}\ \mmod{I^{\ast}})\smkey{\bullet}{\angles{\alpha}^{\ast}I^{\ast}}\ a. \]
  \qed{}
\end{lemma}

\subsection{Normal forms}

The goal of normalization is to prove that every term admits a unique normal form.
We first need to define normal types, normal forms and neutral terms (which correspond to stuck computations).
They are defined, internally to $\CPsh^{\CRen}$, as inductive families indexed by the terms of $\Init_{\Th_{\UU}}$.
\begin{alignat*}{3}
  & \Ne && :{ } && \forall i\ (A : \mmod{F^{\ast}} \to \Ty_{i}) \to (\mmod{F^{\ast}} \to \Tm_{i}\ (A\ \mmod{F^{\ast}})) \to \SPsh^{\CRen} \\
  & \Nf && :{ } && \forall i\ (A : \mmod{F^{\ast}} \to \Ty_{i}) \to (\mmod{F^{\ast}} \to \Tm_{i}\ (A\ \mmod{F^{\ast}})) \to \SPsh^{\CRen} \\
  & \NfTy && :{ } && \forall i \to (\mmod{F^{\ast}} \to \Ty_{i}) \to \SPsh^{\CRen}
\end{alignat*}
An element of $\Ne\ a$ (\resp{} $\Nf\ a$) is a witness of the fact that the term $a$ is a neutral term (\resp{} admits a normal form).
An element of $\NfTy\ A$ is a witness that the type $A$ admits a normal form.

We list below the constructors of these inductive families.
\begin{alignat*}{3}
  & \var^{\ne} && :{ } && \forall i\ A\ (a : \Var_{i}\ A) \to \Ne_{i}\ A\ (\var\ a) \\
  & \lift^{-1,\ne} && :{ } && \forall i\ A \to \Ne_{i+1}\ (\Lift_{i}\ A)\ a \to \Ne_{i}\ (\lift^{-1}\ a) \\
  & \app^{\ne} && :{ } && \forall i\ A\ B\ f\ a \to \Ne_{i}\ f \to \Nf_{i}\ a \\
  &&&&& \to \Ne_{i}\ (\app \circleddollar A \circledast B \circledast f \circledast a) \\
  & \elimBool^{\ne} && :{ } && \forall i\ j\ P\ t\ f\ b \to ((m : \Var\ (\lambda \mmod{F^{\ast}} \mapsto \BoolTy_{i})) \to \NfTy_{j}\ (P \circledast \var\ m)) \\
  &&&&& \to \Nf\ t \to \Nf\ f \to \Ne_{i}\ b \to \Ne_{i}\ (\elimBool \circleddollar P \circledast f \circledast t \circledast b) \\
\end{alignat*}
\begin{alignat*}{3}
  & \nfty^{\nf} && :{ } && \forall i \to \NfTy_{i}\ A \to \Nf_{i+1}\ (\lambda \mmod{F^{\ast}} \mapsto \UU_{i})\ A \\
  & \ne^{\nf}_\BoolTy && :{ } && \forall i\ a \to \Ne_{i}\ (\lambda \mmod{F^{\ast}} \mapsto \BoolTy_{i})\ a \to \Nf_{i}\ a \\
  & \ne^{\nf}_{\El} && :{ } && \forall i\ A \to \Ne_{i+1}\ (\lambda \mmod{F^{\ast}} \mapsto \UU_{i})\ A \\
  &&&&& \to \Ne_{i}\ (\El \circleddollar A)\ a \to \Nf_{i}\ a \\
  & \lift^{\nf} && :{ } && \forall i\ A \to \Nf_{i}\ A\ a \to \Nf_{i+1}\ (\Lift\ A)\ (\lift\ a) \\
  & \true^{\nf} && :{ } && \forall i \to \Nf_{i}\ \true \\
  & \false^{\nf} && :{ } && \forall i \to \Nf_{i}\ \false \\
  & \lam^{\nf} && :{ } && \forall i\ A\ B\ b \to \NfTy_{i}\ A \to ((a : \Var\ A) \to \NfTy_{i}\ (B \circledast \var\ a)) \\
  &&&&& \to ((a : \Var\ A) \to \Nf_{i}\ (b \circledast \var\ a)) \\
  &&&&& \to \Nf_{i}\ (\lam \circleddollar A \circledast B \circledast b) \\
\end{alignat*}
\begin{alignat*}{3}
  & \ne^{\nfty}_{\UU} && :{ } && \forall i\ A \to \Ne_{i+1}\ (\lambda\mmod{F^{\ast}} \to \UU_{i})\ A \to \NfTy_{i+1}\ (\El \circleddollar A) \\
  & \UU^{\nfty} && :{ } && \forall i \to \NfTy_{i+1}\ (\lambda\mmod{F^{\ast}} \to \UU_{i}) \\
  & \Lift^{\nfty} && :{ } && \forall i \to \NfTy_{i}\ A \to \NfTy_{i+1}\ (\Lift\ A) \\
  & \BoolTy^{\nfty} && :{ } && \forall i \to \NfTy_{i}\ (\lambda\mmod{F^{\ast}} \to \BoolTy_{i}) \\
  & \Pi^{\nfty} && :{ } && \forall i\ A\ B \to \NfTy_{i}\ A \to ((a : \Var\ A) \to \NfTy_{i}\ (B \circledast \var\ a)) \\
  &&&&& \to \NfTy_{i}\ (\Pi \circleddollar A \circledast B)
\end{alignat*}

The construction of our normalization function will work for any algebra ($\Ne$, $\Nf$, $\NfTy$, $\Ne$, $\Nf$, $\NfTy$, $\dotsc$) with the above signature.
The choice of the initial algebra is only needed to show uniqueness of normal forms in~\cref{lem:stab_norm}.

\subsection{The normalization displayed model}

We now construct a displayed model without context extensions $\Init_{\Th_{\UU}}^{\bullet}$ over $F : \CRen \to \Init_{\Th_{\UU}}$, internally to $\CPsh^{\CRen}$.

A displayed type $A^{\bullet} : \Ty^{\bullet}\ A$ of $\Init_{\Th_{\UU}}^{\bullet}$ over a type $A : \mmod{F^{\ast}} \to \Ty_{i}$ consists of four components $(A^{\bullet}_{\nfty}, A^{\bullet}_{p}, A^{\bullet}_{\ne}, A^{\bullet}_{\nf})$.
\begin{itemize}
  \item $A^{\bullet}_{\nfty} : \NfTy_{i}\ A$ is a witness of the fact that the type $A$ admits a normal form.
  \item $A^{\bullet}_{p} : (\mmod{F^{\ast}} \to \Tm\ (A\ \mmod{F^{\ast}})) \to \SPsh^{\CRen}_{i}$ is a proof-relevant logical predicate over the terms of type $A$, valued in $i$-small presheaves.
  \item $A^{\bullet}_{\ne} : \forall a \to \Ne_{i}\ a \to A^{\bullet}_{p}\ a$ shows that neutral terms satisfy the logical predicate $A^{\bullet}_{p}$.
    The function $A^{\bullet}_{\ne}$ is often called \emph{unquote} or \emph{reflect}.
  \item $A^{\bullet}_{\nf} : \forall a \to A^{\bullet}_{p}\ a \to \Nf_{i}\ a$ shows that the terms that satisfy the logical predicate $A^{\bullet}_{p}$ admit normal forms.
    The function $A^{\bullet}_{\nf}$ is often called \emph{quote} or \emph{reify}.
\end{itemize}

A displayed term $a^{\bullet} : \Tm^{\bullet}\ A^{\bullet}\ a$ of type $A^{\bullet}$ over a term $(a : \mmod{F^{\ast}} \to \Tm\ (A\ \mmod{F^{\ast}}))$ is an inhabitant $a^{\bullet}$ of $A^{\bullet}_{p}\ a$, \ie{} a witness of the fact that $a$ satisfies the logical predicate $A^{\bullet}_{p}$.

The displayed lifted types are defined as follows.
Because the structures of $\Th_{\UU}$ are not strictly preserved by these lifting operations, $\Lift_{i}\ A$ can be seen as a record type with a projection $\lift^{-1}$ and a constructor $\lift$.
This definition would directly extend to $\Sigma$-types or to other record types.
\begin{alignat*}{3}
  & \Lift^{\bullet} && :{ } && \forall i\ A \to \Ty^{\bullet}_{i}\ A \to \Ty^{\bullet}_{i+1}\ (\Lift_{i} \circleddollar A) \\
  & {(\Lift^{\bullet}_{i}\ A^{\bullet})}_{\nfty} && \triangleq{ } && \Lift^{\nfty}\ A^{\bullet}_{\nfty} \\
  & {(\Lift^{\bullet}_{i}\ A^{\bullet})}_{p} && \triangleq{ } && \lambda a \mapsto A^{\bullet}_{p}\ (\lift^{-1} \circleddollar a) \\
  & {(\Lift^{\bullet}_{i}\ A^{\bullet})}_{\ne} && \triangleq{ } && \lambda a_{\ne} \mapsto A^{\bullet}_{\ne}\ (\lift^{-1,\ne}\ a_{\ne}) \\
  & {(\Lift^{\bullet}_{i}\ A^{\bullet})}_{\nf} && \triangleq{ } && \lambda a^{\bullet} \mapsto \lift^{\nf}\ (A^{\bullet}_{\nf}\ a^{\bullet})
\end{alignat*}

The definition of the displayed universes of the normalization displayed model is below.
\begin{alignat*}{3}
  & \UU^{\bullet}_{i} && :{ } && \Ty_{i+1}^{\bullet}\ (\lambda\ \mmod{F^{\ast}} \mapsto \UU_{i}) \\
  & \UU^{\bullet}_{i,\nfty} && \triangleq{ } && \UU^{\nfty}_{i} \\
  & \UU^{\bullet}_{i,p} && \triangleq{ } && \lambda A \mapsto \Ty^{\bullet}_{i}\ (\El \circleddollar A) \\
  & {(\UU^{\bullet}_{i,\ne}\ A\ A_{\ne})}_{\nfty} && \triangleq{ } && \ne^{\nfty}_{\UU}\ A_{\ne} \\
  & {(\UU^{\bullet}_{i,\ne}\ A\ A_{\ne})}_{p} && \triangleq{ } && \lambda a \mapsto \Ne\ a \\
  & {(\UU^{\bullet}_{i,\ne}\ A\ A_{\ne})}_{\ne} && \triangleq{ } && \lambda a_{\ne} \mapsto a_{\ne} \\
  & {(\UU^{\bullet}_{i,\ne}\ A\ A_{\ne})}_{\nf} && \triangleq{ } && \lambda a_{\ne} \mapsto \ne^{\nf}_{\El}\ a_{\ne} \\
  & \UU^{\bullet}_{i,\nf}\ A\ A^{\bullet} && \triangleq{  } && A^{\bullet}_{\nfty}
\end{alignat*}
The most interesting part is the component $\UU^{\bullet}_{i,\ne}$ that constructs a displayed type over any neutral element of the universe; any element of a neutral type is itself neutral.

For $\Pi$-types, the logical predicates are defined in the same way as for canonicity.
\begin{alignat*}{3}
  & {(\Pi^{\bullet}\ A^{\bullet}\ B^{\bullet})}_{\nfty} && \triangleq{ } &&
  \Pi^{\nfty}\ A^{\nfty}\ (\lambda a_{\var} \mapsto \mathsf{let}\ a^{\bullet} = A^{\bullet}_{\ne}\ (\var^{\ne}\ a_{\var})\ \mathsf{in}\ {(B^{\bullet}\ a^{\bullet})}_{\nfty}) \\
  & {(\Pi^{\bullet}\ A^{\bullet}\ B^{\bullet})}_{p}\ A && \triangleq{ } &&
  \forall a\ (a^{\bullet} : A^{\bullet}_{p}) \to {(B^{\bullet}\ a^{\bullet})}_{p} \\
  & {(\Pi^{\bullet}\ A^{\bullet}\ B^{\bullet})}_{\ne}\ f_{\ne} && \triangleq{ } &&
  \lambda a^{\bullet} \mapsto {(B^{\bullet}\ a^{\bullet})}_{\ne}\ (\app^{\ne}\ f_{\ne}\ (A^{\bullet}_{\nf}\ a^{\bullet})) \\
  & {(\Pi^{\bullet}\ A^{\bullet}\ B^{\bullet})}_{\nf}\ f^{\bullet} && \triangleq{ } && \lam^{\nf}\ (\lambda a_{\var} \mapsto \mathsf{let}\ a^{\bullet} = A^{\bullet}_{\ne}\ (\var^{\ne}\ a_{\var})\ \mathsf{in}\ {(B^{\bullet}\ a^{\bullet})}_{\nf}\ (f^{\bullet}\ a^{\bullet}))
\end{alignat*}

For booleans, we define an inductive family $\BoolTy^{\bullet}_{p} : (\mmod{F^{\ast}} \to \Tm\ \BoolTy) \to \SPsh^{\CRen}$ generated by
\begin{alignat*}{3}
  & \true^{\bullet} && :{ } && \BoolTy^{\bullet}_{p}\ (\lambda\ \mmod{F^{\ast}} \mapsto \true) \\
  & \false^{\bullet} && :{ } && \BoolTy^{\bullet}_{p}\ (\lambda\ \mmod{F^{\ast}} \mapsto \false) \\
  & \ne^{\BoolTy^{\bullet}_{p}} && :{ } && \forall b \to \Ne\ b \to \BoolTy^{\bullet}_{p}\ b
\end{alignat*}
This family witnesses the fact that a boolean term is an element of the free bipointed presheaf generated by the neutral boolean terms.
This extends to the following definition of the displayed boolean type in the normalization model.
\begin{alignat*}{3}
  & \BoolTy^{\bullet}_{\nfty} && \triangleq{ } && \BoolTy^{\nfty} \\
  & \BoolTy^{\bullet}_{\ne}\ n_{\ne} && \triangleq{ } && \ne^{\BoolTy^{\bullet}_{p}} \\
  & \BoolTy^{\bullet}_{\nf}\ \true^{\bullet} && \triangleq{ } && \true^{\nf} \\
  & \BoolTy^{\bullet}_{\nf}\ \false^{\bullet} && \triangleq{ } && \false^{\nf} \\
  & \BoolTy^{\bullet}_{\nf}\ (\ne^{\BoolTy^{\bullet}_{p}}\ b_{\ne}) && \triangleq{ } && \ne^{\nf}_{\BoolTy}\ b_{\ne}
\end{alignat*}

The displayed boolean eliminator $\elimBool^{\bullet}$ is defined using the induction principle of $\BoolTy^{\bullet}_{p}$.

\subsection{Normalization}

Given any displayed type $A^{\bullet}$, every variable of type $A$ satisfies the logical predicate $A^{\bullet}_{p}$; we can define
\begin{alignat*}{3}
  & \var^{\bullet} && :{ } && \forall \mmod{I^{\ast}}\ i\ (A : \mmod{F^{\ast}} \to \Ty_{i}) (a : \Var\ A) \to \Tm^{\bullet}\ (S^{\Ty}_{\iota}\ \mmod{I^{\ast}}\ A)\ (\var\ a) \\
  & \var^{\bullet} && \triangleq{ } && \lambda \mmod{I^{\ast}}\ A\ a \mapsto {(S^{\Ty}_{\iota}\ \mmod{I^{\ast}}\ A)}_{\ne}\ (\var^{\ne}\ a)
\end{alignat*}

We can now apply \cref{lem:ind_renamings_univ} to $\Init_{\Th_{\UU}}^{\bullet}$.
We obtain a relative section $S_{\alpha}$ of $\Init_{\Th_{\UU}}^{\bullet}$.

This proves the existence of normal forms, as witnessed by the following normalization function, internally to $\CPsh^{\CRen}$.
\begin{alignat*}{3}
  & \norm && :{ } && \forall i\ (A : \mmod{F^{\ast}} \to \Ty_{i})\ (a : \mmod{F^{\ast}} \to \Tm\ (A\ \mmod{F^{\ast}})) \to \Nf_{i}\ a \\
  & \norm\ A\ a && \triangleq{ } && {(S_{\alpha}^{\Ty}\ A)}_{\nf}\ (S_{\alpha}^{\Tm}\ a)
\end{alignat*}

\subsection{Stability of normalization}

It remains to show the uniqueness of normal forms.
It follows from stability of normalization:%
\begin{lemma}[Internally to $\CPsh^{\CRen}$]\label{lem:stab_norm}
  For every $a_{\ne} : \Ne_{i}\ A\ a$, we have $S_{\alpha}^{\Tm}\ a = {(S_{\alpha}^{\Ty}\ A)}_{\ne}\ a_{\ne}$, and for every $a_{\nf} : \Nf_{i}\ A\ a$, we have ${(S_{\alpha}^{\Ty}\ A)}_{\nf}\ (S_{\alpha}^{\Tm}\ a) = a_{\nf}$.
  Furthermore for every $A_{\nfty} : \NfTy_{i}\ A$, we have ${(S_{\alpha}^{\Ty}\ A)}_{\nfty} = A_{\nfty}$.
\end{lemma}
\begin{proof}
  This lemma is proven by induction on $\Ne$, $\Nf$, $\NfTy$.

  We only list some of the cases.
  \begin{description}
    \item[$\var^{\ne}\ a : \Ne\ (\var\ A\ a)$]
          \begin{alignat*}{1}
            & S_{\alpha}^{\Tm}\ (\var\ A\ a) \\
            & \quad { }= {(S_{\alpha}^{\Ty}\ A)}_{\ne}\ (\var^{\ne}\ a)
            \tag*{(computation rule of $S_{\alpha}$)}
          \end{alignat*}
    \item[$\app^{\ne}\ f_{\ne}\ a_{\nf} : \Ne\ (\app \circleddollar A \circledast B \circledast f \circledast a)$]
          \begin{alignat*}{1}
            & S_{\alpha}^{\Tm}\ (\app \circleddollar f \circledast a) \\
            & \quad { }= \app^{\bullet}\ (S_{\alpha}^{\Tm}\ f)\ (S_{\alpha}^{\Tm}\ a)
            \tag*{(computation rule of $S_{\alpha}$)} \\
            & \quad { }= (S_{\alpha}^{\Tm}\ f)\ (S_{\alpha}^{\Tm}\ a)
            \tag*{(definition of $\app^{\bullet}$)} \\
            & \quad { }= ({(S_{\alpha}^{\Ty}\ (\Pi\ A\ B))}_{\ne}\ f_{\ne})\ (S_{\alpha}^{\Tm}\ a)
            \tag*{(induction hypothesis for $f_{\ne}$)} \\
            & \quad { }= {(\Pi^{\bullet}\ (S_{\alpha}^{\Ty}\ A)\ (S_{\alpha}^{[\Tm]\Ty}\ B))}_{\ne}\ f_{\ne}\ (S_{\alpha}^{\Tm}\ a)
            \tag*{(computation rule of $S_{\alpha}$)} \\
            & \quad { }= {((S_{\alpha}^{[\Tm]\Ty}\ B)\ (S_{\alpha}^{\Tm}\ a))}_{\ne}\ {(\app^{\ne}\ f_{\ne}\ ({(S_{\alpha}^{\Ty}\ A)}_{\nf}\ (S_{\alpha}^{\Tm}\ a)))}
            \tag*{(definition of $\Pi^{\bullet}$)} \\
            & \quad { }= {(S_{\alpha}^{\Ty}\ (B \circledast a))}_{\ne}\ (\app^{\ne}\ f_{\ne}\ a_{\nf})
            \tag*{(induction hypothesis for $a_{\nf}$)}
          \end{alignat*}
    \item[$\lam^{\nf}\ b_{\nf} : \Nf\ (\lam \circleddollar \{A\} \circledast \{B\} \circledast b)$]
          \begin{alignat*}{1}
            & {(S_{\alpha}^{\Ty}\ (\Pi\ A\ B))}_{\nf}\ (S_{\alpha}^{\Tm}\ (\lam \circleddollar b)) \\
            & \quad { }= {(\Pi^{\bullet}\ (S_{\alpha}^{\Ty}\ A)\ (S_{\alpha}^{[\Tm]\Ty}\ B))}_{\nf}\ (\lam^{\bullet}\ (S_{\alpha}^{[\Tm]\Tm}\ b))
            \tag*{(computation rule of $S_{\alpha}$)} \\
            & \quad { }= {(\Pi^{\bullet}\ (S_{\alpha}^{\Ty}\ A)\ (S_{\alpha}^{[\Tm]\Ty}\ B))}_{\nf}\ (\lambda\ a^{\bullet} \mapsto (S_{\alpha}^{[\Tm]\Tm}\ b)\ a^{\bullet})
            \tag*{(definition of $\lam^{\bullet}$)} \\
            & \quad { }= \lam^{\nf}\ (\lambda\ a_{\var} \mapsto \mathsf{let}\ a^{\bullet} = {(S_{\alpha}^{\Ty}\ A)}_{\ne}\ (\var^{\ne}\ a_{\var})\ \mathsf{in}\ {(S_{\alpha}^{[\Tm]\Ty}\ B\ a^{\bullet})}_{\nf}\ (S_{\alpha}^{[\Tm]\Tm}\ b\ a^{\bullet}))
            \tag*{(definition of $\Pi^{\bullet}$)} \\
            & \quad { }= \lam^{\nf}\ (\lambda\ a_{\var} \mapsto \mathsf{let}\ a^{\bullet} = S_{\alpha}^{\Tm}\ (\var\ a_{\var})\ \mathsf{in}\ {(S_{\alpha}^{[\Tm]\Ty}\ B\ a^{\bullet})}_{\nf}\ (S_{\alpha}^{[\Tm]\Tm}\ b\ a^{\bullet}))
            \tag*{(computation rule of $S_{\alpha}$)} \\
            & \quad { }= \lam^{\nf}\ (\lambda\ a_{\var} \mapsto {(S_{\alpha}^{\Ty}\ (B\ \circledast \var\ a_{\var}))}_{\nf}\ (S_{\alpha}^{\Tm}\ (b\ \circledast \var\ a_{\var})))
            \tag*{(computation rule of $S_{\alpha}$)} \\
            & \quad { }= \lam^{\nf}\ (\lambda\ a_{\var} \mapsto b_{\nf}\ a_{\var})
            \tag*{(induction hypothesis for $b_{\nf}$)} \\
            & \quad { }= \lam^{\nf}\ b_{\nf}
          \end{alignat*}
    \item[$\Pi^{\nfty}\ A_{\nfty}\ B_{\nfty} : \NfTy\ (\Pi \circleddollar A \circledast B)$]
          \begin{alignat*}{1}
            & {(S_{\alpha}^{\Ty}\ (\Pi \circleddollar A \circledast B))}_{\nfty} \\
            & \quad { }= {(\Pi^{\bullet}\ (S_{\alpha}^{\Ty}\ A)\ (S_{\alpha}^{[\Tm]\Ty}\ B))}_{\nfty}
            \tag*{(computation rule of $S_{\alpha}$)} \\
            & \quad { }= \Pi^{\nfty}\ {(S_{\alpha}^{\Ty}\ A)}_{\nfty}\ (\lambda a_{\var} \mapsto \mathsf{let}\ a^{\bullet} = {(S_{\alpha}^{\Ty}\ A)}_{\ne}\ (\var^{\ne}\ a_{\var})\ \mathsf{in}\ {(S_{\alpha}^{[\Tm]\Ty}\ B\ a^{\bullet})}_{\nfty})
            \tag*{(definition of $\Pi^{\bullet}$)} \\
            & \quad { }= \Pi^{\nfty}\ {(S_{\alpha}^{\Ty}\ A)}_{\nfty}\ (\lambda a_{\var} \mapsto \mathsf{let}\ a^{\bullet} = s^{\Tm}_{\alpha}\ (\var\ a_{\var})\ \mathsf{in}\ {(S_{\alpha}^{[\Tm]\Ty}\ B\ a^{\bullet})}_{\nfty})
            \tag*{(computation rule of $S_{\alpha}$)} \\
            & \quad { }= \Pi^{\nfty}\ {(S_{\alpha}^{\Ty}\ A)}_{\nfty}\ (\lambda a_{\var} \mapsto {(S_{\alpha}^{\Ty}\ (B\ \circledast \var\ a_{\var}))}_{\nfty})
            \tag*{(computation rule of $S_{\alpha}$)} \\
            & \quad { }= \Pi^{\nfty}\ {(S_{\alpha}^{\Ty}\ A)}_{\nfty}\ (\lambda a_{\var} \mapsto B_{\nfty}\ a_{\var})
            \tag*{(induction hypothesis for $B_{\nfty}$)} \\
            & \quad { }= \Pi^{\nfty}\ A_{\nfty}\ B_{\nfty}
            \tag*{(induction hypothesis for $A_{\nfty}$)} \\
          \end{alignat*}
  \end{description}
\end{proof}


\end{document}